\documentclass[11pt,letterpaper]{article}
\usepackage[margin=0.8in]{geometry}

\usepackage{amsmath,amsfonts,amssymb}
\usepackage{array}
\usepackage{textcomp}
\usepackage{url}
\usepackage{verbatim}
\usepackage{graphicx}
\usepackage{cite}
\usepackage{color}
\usepackage{bm}
\usepackage{qcircuit}
\usepackage{braket}
\usepackage{latexsym,mathrsfs,mathtools,enumerate}
\usepackage{hyperref}
\usepackage{placeins}
\usepackage{authblk}

\allowdisplaybreaks[4]
\hypersetup{colorlinks,linkcolor={blue},citecolor={blue},urlcolor={red}}

\newtheorem{definition}{Definition}
\newtheorem{question}[definition]{Question}
\newtheorem{lemma}[definition]{Lemma}
\newtheorem{remark}[definition]{Remark}
\newtheorem{theorem}[definition]{Theorem}
\newtheorem{example}[definition]{Example}
\newtheorem{proposition}[definition]{Proposition}

\newtheorem{corollary}[definition]{Corollary}
\newtheorem{conjecture}[definition]{Conjecture}

\newtheorem{memo}[definition]{Memo}

\def\squareforqed{\hbox{\rlap{$\sqcap$}$\sqcup$}}
\def\qed{\ifmmode\squareforqed\else{\unskip\nobreak\hfil
\penalty50\hskip1em\null\nobreak\hfil\squareforqed
\parfillskip=0pt\finalhyphendemerits=0\endgraf}\fi}
\def\endenv{\ifmmode\;\else{\unskip\nobreak\hfil
\penalty50\hskip1em\null\nobreak\hfil\;
\parfillskip=0pt\finalhyphendemerits=0\endgraf}\fi}
\newenvironment{proof}{\noindent \textbf{{Proof.~} }}{\qed}
\def\Dbar{\leavevmode\lower.6ex\hbox to 0pt
{\hskip-.23ex\accent"16\hss}D}
\makeatletter
\def\url@leostyle{%
  \@ifundefined{selectfont}{\def\UrlFont{\sf}}{\def\UrlFont{\small\ttfamily}}}
\makeatother

\def\bcj{\begin{conjecture}}
\def\ecj{\end{conjecture}}
\def\bcr{\begin{corollary}}
\def\ecr{\end{corollary}}
\def\bd{\begin{definition}}
\def\ed{\end{definition}}
\def\bea{\begin{eqnarray}}
\def\eea{\end{eqnarray}}
\def\beq{\begin{equation}}
\def\eeq{\end{equation}}
\def\bal{\begin{aligned}}
\def\eal{\end{aligned}}
\def\bem{\begin{enumerate}}
\def\eem{\end{enumerate}}
\def\bex{\begin{example}}
\def\eex{\end{example}}
\def\bim{\begin{itemize}}
\def\eim{\end{itemize}}
\def\bl{\begin{lemma}}
\def\el{\end{lemma}}
\def\bma{\begin{bmatrix}}
\def\ema{\end{bmatrix}}
\def\bpf{\begin{proof}}
\def\epf{\end{proof}}
\def\bpp{\begin{proposition}}
\def\epp{\end{proposition}}
\def\bqu{\begin{question}}
\def\equ{\end{question}}
\def\br{\begin{remark}}
\def\er{\end{remark}}
\def\bt{\begin{theorem}}
\def\et{\end{theorem}}
\def\bmm{\begin{memo}}
\def\emm{\end{memo}}

\def\btb{\begin{tabular}}
\def\etb{\end{tabular}}

\newcommand{\nc}{\newcommand}

\nc{\bbA}{\mathbb{A}} \nc{\bbB}{\mathbb{B}} \nc{\bbC}{\mathbb{C}}
 \nc{\bbD}{\mathbb{D}} \nc{\bbE}{\mathbb{E}} \nc{\bbF}{\mathbb{F}}
 \nc{\bbG}{\mathbb{G}} \nc{\bbH}{\mathbb{H}} \nc{\bbI}{\mathbb{I}}
 \nc{\bbJ}{\mathbb{J}} \nc{\bbK}{\mathbb{K}} \nc{\bbL}{\mathbb{L}}
 \nc{\bbM}{\mathbb{M}} \nc{\bbN}{\mathbb{N}} \nc{\bbO}{\mathbb{O}}
 \nc{\bbP}{\mathbb{P}} \nc{\bbQ}{\mathbb{Q}} \nc{\bbR}{\mathbb{R}}
 \nc{\bbS}{\mathbb{S}} \nc{\bbT}{\mathbb{T}} \nc{\bbU}{\mathbb{U}}
 \nc{\bbV}{\mathbb{V}} \nc{\bbW}{\mathbb{W}} \nc{\bbX}{\mathbb{X}}
 \nc{\bbZ}{\mathbb{Z}}

 \nc{\bA}{{\bf A}} \nc{\bB}{{\bf B}} \nc{\bC}{{\bf C}}
 \nc{\bD}{{\bf D}} \nc{\bE}{{\bf E}} \nc{\bF}{{\bf F}}
 \nc{\bG}{{\bf G}} \nc{\bH}{{\bf H}} \nc{\bI}{{\bf I}}
 \nc{\bJ}{{\bf J}} \nc{\bK}{{\bf K}} \nc{\bL}{{\bf L}}
 \nc{\bM}{{\bf M}} \nc{\bN}{{\bf N}} \nc{\bO}{{\bf O}}
 \nc{\bP}{{\bf P}} \nc{\bQ}{{\bf Q}} \nc{\bR}{{\bf R}}
 \nc{\bS}{{\bf S}} \nc{\bT}{{\bf T}} \nc{\bU}{{\bf U}}
 \nc{\bV}{{\bf V}} \nc{\bW}{{\bf W}} \nc{\bX}{{\bf X}}
 \nc{\bZ}{{\bf Z}}

\nc{\cA}{{\cal A}} \nc{\cB}{{\cal B}} \nc{\cC}{{\cal C}}
\nc{\cD}{{\cal D}} \nc{\cE}{{\cal E}} \nc{\cF}{{\cal F}}
\nc{\cG}{{\cal G}} \nc{\cH}{{\cal H}} \nc{\cI}{{\cal I}}
\nc{\cJ}{{\cal J}} \nc{\cK}{{\cal K}} \nc{\cL}{{\cal L}}
\nc{\cM}{{\cal M}} \nc{\cN}{{\cal N}} \nc{\cO}{{\cal O}}
\nc{\cP}{{\cal P}} \nc{\cQ}{{\cal Q}} \nc{\cR}{{\cal R}}
\nc{\cS}{{\cal S}} \nc{\cT}{{\cal T}} \nc{\cU}{{\cal U}}
\nc{\cV}{{\cal V}} \nc{\cW}{{\cal W}} \nc{\cX}{{\cal X}}
\nc{\cZ}{{\cal Z}}

\nc{\hA}{{\hat{A}}} \nc{\hB}{{\hat{B}}} \nc{\hC}{{\hat{C}}}
\nc{\hD}{{\hat{D}}} \nc{\hE}{{\hat{E}}} \nc{\hF}{{\hat{F}}}
\nc{\hG}{{\hat{G}}} \nc{\hH}{{\hat{H}}} \nc{\hI}{{\hat{I}}}
\nc{\hJ}{{\hat{J}}} \nc{\hK}{{\hat{K}}} \nc{\hL}{{\hat{L}}}
\nc{\hM}{{\hat{M}}} \nc{\hN}{{\hat{N}}} \nc{\hO}{{\hat{O}}}
\nc{\hP}{{\hat{P}}} \nc{\hR}{{\hat{R}}} \nc{\hS}{{\hat{S}}}
\nc{\hT}{{\hat{T}}} \nc{\hU}{{\hat{U}}} \nc{\hV}{{\hat{V}}}
\nc{\hW}{{\hat{W}}} \nc{\hX}{{\hat{X}}} \nc{\hZ}{{\hat{Z}}}

\nc{\hn}{{\hat{n}}}

\nc{\as}{{\cal AS}}
	\nc{\app}{{\cal AP}}

\def\max{\mathop{\rm max}}
\def\min{\mathop{\rm min}}

\def\dg{\dagger}

\providecommand{\bra}[1]{\langle#1|}
\providecommand{\ket}[1]{|#1\rangle}

\allowdisplaybreaks[4]

\title{High-Dimensional Carrier-Assisted Entanglement Purification Based on Mutually Unbiased Bases}

\author[1]{Zihua Song}
\author[1,*]{Lin Chen}
\author[1,*]{Yongge Wang}

\affil[1]{LMIB (Beihang University), Ministry of Education,
and School of Mathematical Sciences, Beihang University,
Beijing 100191, China}

\affil[*]{Corresponding authors:
linchen@buaa.edu.cn (Lin Chen);
wangyongge@buaa.edu.cn (Yongge Wang)}

\date{September 5, 2026}

\begin{document}

\maketitle

\begin{abstract}
Distilling high-dimensional entanglement under general asymmetric Pauli noise remains challenging. For the unprocessed qutrit single-carrier recurrence studied here, severe noise asymmetry can prevent the convergence condition from being satisfied. In this paper, we investigate carrier-assisted entanglement purification protocols, namely CAEPP and mCAEPP, first for two-qutrit systems, and show that the unprocessed single-carrier recurrence is bottlenecked by marginal $X$-error probabilities. To overcome this limitation, we introduce MUB-guided preprocessing that combines randomized inversion-pair symmetrization with deterministic maximum-line alignment. For any qutrit Pauli channel with initial fidelity $p_{00}>1/3$, we prove that, conditioned on successful syndrome outcomes, the MUB-adapted recurrence converges for every sufficiently large fixed carrier number and that its fixed-point fidelity approaches one in the subsequent large-carrier limit. We further extend the algebraic carrier-assisted framework and the asymmetric-noise bottleneck to arbitrary qudit dimensions, and show that in prime-power dimensions the finite-field preprocessing gives the sufficient threshold $p_{00}>1/d$. A certified equal-target comparison over a finite deterministic optimization domain further shows that neither the MUB-adapted strategy nor complete Pauli-label twirling uniformly dominates the other in finite-resource cost.
\end{abstract}

\begin{center}
\end{center}

\section{Introduction}
\label{sec:introduction}

Entanglement distillation, the extraction of pure entangled states from mixed ensembles, is a fundamental prerequisite for reliable quantum information processing \cite{bbc93, werner89}. While the basic distillability of bipartite systems has been extensively mapped out via the positive partial transpose (PPT) and reduction criteria \cite{peres1996, hhh96, Horodecki1997Reduction, 1997Inseparable, hhh1998}, characterizing and distilling full-rank mixed states under general noise remains notoriously difficult \cite{2002.03233}. Traditionally, standard local twirling operations have been employed to convert bipartite states into highly symmetric forms, such as Werner states \cite{2000Evidence}, to simplify the collective processing of multiple copies \cite{watrous04,vd06,Chen2018Generalized}. However, to circumvent severe mathematical complexities, much of this foundational effort was strictly limited to systems with low matrix rank \cite{1999Rank,Lin2008Rank,cd16pra,cd11jpa}, alongside various quantitative estimations mostly confined to low-dimensional protocols \cite{rains1999,rains2001,dw2005,Kwiat2001Experimental,ch12ijmpb}. 

High-dimensional entanglement offers increased information capacity and, in some settings, enhanced noise robustness, as reviewed in \cite{Erhard2020NatRev}. Recent demonstrations of high-dimensional quantum teleportation \cite{Luo2019PRL} and heralded photon--photon quantum gates \cite{Liu2026NatPhot} indicate continuing progress in multilevel coherent control. Nevertheless, entanglement distillation under general asymmetric noise remains difficult. Recent no-go results delimit universal entanglement purification \cite{Zang2025NoGo} and catalysis-assisted distillation \cite{Lami2024Catalysis}, but they do not imply that every recurrence protocol fails under asymmetric noise. Here we instead identify a convergence obstruction for the particular unprocessed qutrit single-carrier recurrence studied below. Related work has addressed distillation with noisy measurements \cite{Kim2025}, high-dimensional state certification \cite{Lib2025PRL}, symmetric projection constraints \cite{Morelli2023PRL}, and entanglement-based communication under high noise \cite{Hu2021PRL}.

Within the carrier-assisted setting considered here, overcoming the asymmetric-noise obstruction motivates preprocessing that retains selected directional information. In the qutrit setting, the MUB-guided strategy combines randomized inversion-pair symmetrization with a deterministic Clifford rotation that aligns a maximum-weight MUB line with the primary axis. This produces the symmetry and spectral separation used to prove convergence under general asymmetric qutrit Pauli noise. In prime-power dimensions, the corresponding finite-field construction uses maximum-line alignment followed by multiplicative symmetrization, which preserves the aligned-axis weight and equalizes nonzero column totals but does not generally preserve every MUB-line weight separately.

In this paper, we extend the carrier-assisted entanglement purification protocol (CAEPP) proposed for qubits in \cite{qubitCAEPP} to three-dimensional systems (qutrits) subject to general asymmetric noise, and then discuss its higher-dimensional extension. The remainder of this paper is organized as follows. In Section \ref{se:pre}, we establish the theoretical framework by introducing qutrit Pauli channels, Bell-diagonal states, and entanglement fidelity. Specifically, Lemma~\ref{le:EB1} shows that $F(\rho)>1/3$ rules out an entanglement-breaking channel, while the associated Bell-diagonal state violates the reduction criterion and is $1$-distillable. Section \ref{sec:single} details the single-carrier-assisted entanglement purification protocol. Through rigorous error propagation analysis, Proposition \ref{pr:inf} demonstrates that without adaptive preprocessing, the protocol's convergence is strictly bottlenecked by the marginal $X$-error probabilities. This identifies a convergence obstruction for the unprocessed qutrit single-carrier recurrence and motivates the preprocessing developed below. To overcome this limitation, Section \ref{sec:multi} introduces the multi-carrier generalized protocol (mCAEPP) via stabilizer codes. Lemma~\ref{le:cnot1} gives the elementary qutrit SUM-gate conjugation rule, while Lemma~\ref{le:cnot2} records a separate parallel-SUM conjugation identity. The star-code encoder and its syndrome-extraction rules are specified explicitly in Section~\ref{subsec:sca}. Building on this framework, Theorem~\ref{thm:mcaepp_convergence} proves that, for symmetric depolarizing channels, the conditional recurrence converges for every fixed finite $m$, while its fixed-point fidelity approaches one as $m\to\infty$. To tackle general asymmetric qutrit noise, we introduce MUB-guided preprocessing consisting of randomized inversion-pair symmetrization and deterministic maximum-line alignment. Governed by the discrete geometric properties of the qutrit phase space established in Lemma \ref{lemma:mub_dominance}, this step actively rotates the phase space to enforce a primary-axis error dominance. Our main qutrit result, Theorem~\ref{thm:universal_purification}, proves that for every qutrit Pauli channel with $p_{00}>1/3$, the conditional MUB-adapted recurrence converges for every sufficiently large fixed carrier number and its fixed-point fidelity approaches one in the subsequent large-carrier limit. Relative to the qubit carrier architecture in \cite{qubitCAEPP}, the principal new ingredients are the qutrit MUB-adapted recurrence and its asymmetric-noise spectral analysis, followed by the finite-field convolution construction in prime-power dimensions. We also compare the MUB-adapted strategy with complete Pauli-label twirling at equal target fidelity, including cumulative success probability, restart-aware channel uses, and the shared randomness of the specified implementations. Finally, Section \ref{sec:higher} extends the algebraic carrier-assisted framework to higher-dimensional qudits. Proposition \ref{pr:inf_generalized} shows that the single-carrier convergence bottleneck persists in arbitrary dimension $d$, where convergence requires the target fidelity to dominate all marginal shift-error probabilities. Equation~\eqref{eq:finite_field_line_sum} gives the line-sum identity from which the corresponding maximum-line lower bound follows in prime-power dimensions. Building on finite-field line alignment and multiplicative symmetrization, Theorem~\ref{thm:prime_power_mcaepp} proves that $p_{00}>1/d$ is sufficient for the stated fixed-$m$ convergence and ordered large-$m$ limit, while composite dimensions that are not prime powers require additional ideas beyond this finite-field construction.

\section{Preliminaries}

\label{se:pre}

To characterize three-dimensional bipartite systems and generalized Pauli channels, we first introduce the set of single-qutrit generalized Pauli operators. This set is constructed from the cyclic-shift operator $X$ and the phase-shift operator $Z$, whose actions on the computational basis are defined as $X\ket{j}=\ket{(j+1)\bmod 3}$ and $Z\ket{j}=\omega^j\ket{j}$, respectively, with $\omega=e^{2\pi i/3}$. These operators satisfy the fundamental commutation relation $ZX=\omega XZ$. Consequently, any operator within this set can be uniquely expressed in the form $X^u Z^v$ for $u, v \in \{0,1,2\}$.

The qutrit Fourier transform is defined by
\begin{equation}
    H_3\ket{j} = \frac{1}{\sqrt{3}} \sum_{k=0}^{2}\omega^{jk}\ket{k},
    \qquad
    \omega=e^{2\pi i/3},
\end{equation}
or equivalently,
\begin{equation}
    H_3 = \frac{1}{\sqrt{3}}
    \begin{pmatrix}
        1 & 1 & 1\\
        1 & \omega & \omega^2\\
        1 & \omega^2 & \omega
    \end{pmatrix}.
\end{equation}
Under this convention, the qutrit Pauli operators satisfy
\begin{equation}
    H_3XH_3^\dagger=Z,
    \qquad
    H_3ZH_3^\dagger=X^{-1},
    \qquad
    H_3^\dagger ZH_3=X.
\end{equation}
 
The fundamental resources for our protocol are the nine three-dimensional Bell states, i.e., maximally entangled qutrit pairs. They form a complete orthonormal basis and are uniformly defined as
\begin{equation}
\ket{\Phi^{n,m}} = \frac{1}{\sqrt{3}}\sum_{j=0}^2\omega^{mj}\ket{j, (j+n)\bmod 3},
\end{equation}
where $n, m \in \{0,1,2\}$ represent the shift and phase indices, respectively. Notably, applying an arbitrary local Pauli operation $I \otimes X^u Z^v$ to $\ket{\Phi^{n,m}}$ deterministically shifts its indices, enabling mutual conversions among all nine orthogonal states up to a global phase factor.

Any three-dimensional bipartite Bell-diagonal state can be written as
\begin{equation}
\label{eq:sigma0}
\rho=\sum_{n,m=0}^2p_{nm}\ket{\Phi^{n,m}}\bra{\Phi^{n,m}},
\end{equation}
where $\sum_{n,m=0}^2p_{nm}=1$.

A three-dimensional Pauli channel $\mathcal{N}$ is parameterized by these probabilities $p_{nm}$ as
\begin{equation}
\label{eq:channel1}
    \mathcal{N}(\cdot)=\sum_{n,m=0}^2 p_{nm} Z^m X^n (\cdot) X^{-n} Z^{-m}.
\end{equation}
When Alice sends one qutrit of $\ket{\Phi^{0,0}}$ through $\mathcal{N}$, the resulting shared state, known as the CJ state of the channel, is given by
\begin{equation}
\label{eq:sigma1}
    \sigma=(I_3\otimes\mathcal{N})\left(\ket{\Phi^{0,0}}\!\bra{\Phi^{0,0}}\right)=\sum_{n,m=0}^2p_{nm}\ket{\Phi^{n,m}}\!\bra{\Phi^{n,m}}.
\end{equation}

The entanglement fidelity for three-dimensional states $\rho$ is defined as 
\begin{equation}
\label{eq:fidelity1}
F(\rho)=\langle\Phi^{0,0}|\rho|\Phi^{0,0}\rangle=p_{00}.
\end{equation}

The following lemma is a direct consequence of the results in \cite{2003Entanglement}.

\begin{lemma}
\label{le:EB1}
    An $N$-dimensional channel $\mathcal{N}$ is EB only if its corresponding CJ state $\rho$ has an entanglement fidelity $F(\rho) \le 1/N$. 
\end{lemma}

For the qutrit Bell-diagonal state in \eqref{eq:sigma0}, $p_{00}>1/3$ implies violation of the reduction criterion and hence one-copy distillability \cite{Horodecki1997Reduction}; in particular, the state is NPT entangled.

Based on these preliminaries, we first present the single-carrier CAEPP protocol in Section \ref{sec:single}, followed by the multi-carrier generalization and stabilizer code implementation in Section \ref{sec:multi}.

\section{single-carrier-assisted Entanglement Purification in Three Dimensions}
\label{sec:single}

The three-dimensional CAEPP in this section extends the qubit CAEPP recently introduced in \cite{qubitCAEPP}. Both protocols share a fundamental operational architecture that avoids consuming additional entangled pairs by utilizing an ancillary carrier initialized in the state $\ket{0}$ to traverse the noisy channel. Furthermore, both employ collective bilateral encoding and decoding operations and rely on $Z$-basis measurements on the carrier for syndrome extraction and post-selection. In both dimensions, a basic preprocessing step involving local unitaries is applied to the shared state to minimize pure phase errors prior to encoding.

However, generalizing this framework to three-dimensional systems introduces algebraic complexities and severe noise asymmetry challenges. The standard CNOT gates and Pauli operators are replaced by high-dimensional analogues such as modulo-$3$ SUM gates $U_{CN}$ and generalized qutrit Pauli operators. Moreover, the qutrit Pauli error space contains eight nonidentity errors, allowing more varied asymmetric error distributions than in the qubit case. Although a fixed local preprocessing step is sufficient to handle general noise in the qubit case, we show in this section that sufficiently asymmetric noise can violate the convergence condition of the unprocessed qutrit single-carrier recurrence. This obstruction motivates the MUB-guided preprocessing developed in Section \ref{sec:multi}.

\subsection{Protocol Overview}

The single-carrier purification protocol proceeds through the following sequential steps.

1. Initialization: Alice and Bob share a three-dimensional Bell-diagonal state $\rho$ with fidelity $F=p_{00}>1/3$, as defined in \eqref{eq:fidelity1}. By Lemma~\ref{le:EB1}, this condition ensures that the corresponding qutrit channel is not entanglement-breaking.

2. preprocessing: Alice and Bob apply local unitaries to the shared pair.

3. Encoding: Alice prepares a single carrier qutrit in the initial state $\ket{0}$. Then she applies a three-dimensional CNOT gate $U_{CN}$, using her shared qutrit as the control and the carrier as the target. Following this encoding step, she transmits the carrier through the quantum channel.

4. Decoding: When receiving the carrier, Bob applies a reverse CNOT gate $U_{CN'}$, using his shared qutrit as the control and the carrier as the target. Subsequently, the carrier is measured in the $Z$-basis.

5. Decision: Alice and Bob retain the purified shared state if the syndrome measurement yields $0$, indicating no detected error; otherwise, the state is discarded and the protocol restarts.

Figure \ref{fig:3dcaepp} presents the detailed schematic diagram of the single-carrier operational workflow.

\begin{figure}[htbp]
    \centering
    \[
    \Qcircuit @C=1.5em @R=2.2em {
        \mbox{\hspace{1em}} & \mbox{} & A_{\text{in}} & \mbox{} & \qw & \gate{U_A^{\text{pre}}} & \ctrl{1} & \qw & \qw & \qw & \qw & \qw & \mbox{} \\
        \rho_{\text{in}} & \mbox{} \ar@{-}[urr] \ar@{-}[drr] & \mbox{} & \ket{0} & \mbox{} & \qw & \targ & \gate{\mathcal{N}} & \targ & \meter \ar@{}[u]|-{\raisebox{-2em}{$\text{out}=0$}} & \qw & \qw & \mbox{} & \rho_{\text{out}} \\
        \mbox{} & \mbox{} & B_{\text{in}} & \mbox{} & \qw & \gate{U_B^{\text{pre}}} & \qw & \qw & \ctrl{-1} & \qw & \qw & \qw & \mbox{}
        \gategroup{1}{12}{3}{12}{0.6em}{\}}
    }
    \]
    \vspace{0.5em}
    \caption{Quantum circuit of the single-carrier CAEPP protocol}
    \label{fig:3dcaepp}
    \vspace{0.3em}
    \small
    \centering
    \begin{minipage}{0.9\textwidth}
    \textbf{Notation Description:}
    \begin{itemize}
        \item $A_{\text{in}}$: Input qutrit of the shared entangled pair at Alice's side
        \item $B_{\text{in}}$: Input qutrit of the shared entangled pair at Bob's side
        \item $\rho_{\text{in}}$: Initial noisy shared Bell-diagonal entangled state
        \item $\ket{0}$: Initial carrier qutrit
        \item $U_A^{\text{pre}}, U_B^{\text{pre}}$: Local preprocessing unitary operations
        \item $\mathcal{N}$: Noisy qutrit channel acting on the carrier
        \item $\text{out}=0$: Success condition of Z-basis measurement on the carrier
        \item $\rho_{\text{out}}$: Purified output entangled state after successful protocol
    \end{itemize}
    \end{minipage}
\end{figure}
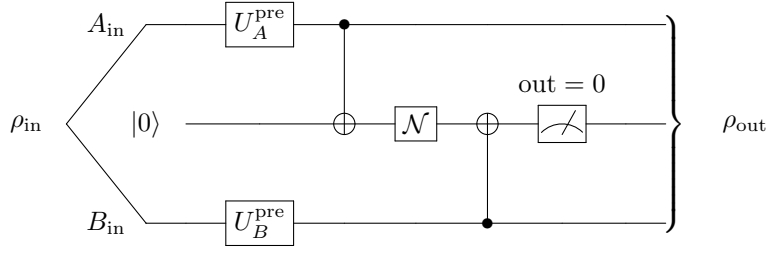

\subsection{Detailed Protocol Design}
\label{subsub:dotp}

We now evaluate how the protocol affects the shared state $\sigma$ in \eqref{eq:sigma1}. Linearity permits an independent analysis of each Bell-diagonal component $\ket{\Phi^{n,m}}\bra{\Phi^{n,m}}$. The derivation proceeds as follows.

Specifically, the action of the three-dimensional CNOT gate $U_{CN}$ is defined as
\begin{equation}
    U_{CN} \ket{j}_C\otimes\ket{k}_T=\ket{j}_C \otimes \ket{(j+k)\bmod 3}_T,
\end{equation}
where the subscripts $C$ and $T$ denote the control and target qutrits, respectively. 

Before the gate application, the joint state of the shared Bell pair and the initial carrier is
\begin{equation}
    \ket{\Phi^{n,m}}_{AB}\otimes\ket{0}_{T}=\frac{1}{\sqrt{3}}\sum_{j=0}^2\omega^{mj}\ket{j}_{A}\otimes\ket{(j+n)\bmod 3}_{B}\otimes\ket{0}_{T}.
\end{equation}
In our protocol, Alice's qutrit $A$ serves as the control and the ancillary carrier $T$ acts as the target, while Bob's qutrit $B$ remains unaffected as a spectator.

Applying $U_{CN}$ between Alice's qutrit $A$ and the carrier $T$ transforms the target state from $\ket{0}_T$ to $\ket{(j+0) \bmod3}_T=\ket{j}_T$, yielding the joint state
\begin{equation}
    U_{CN} \left(\ket{\Phi^{n,m}}_{AB}\otimes\ket{0}_{T}\right)=\frac{1}{\sqrt{3}}\sum_{j=0}^2\omega^{mj}\ket{j}_{A}\otimes\ket{(j+n)\bmod 3}_{B}\otimes\ket{j}_{T}.
\end{equation}

The carrier $T$ is then transmitted through the quantum channel to Bob. Upon receiving it, Bob applies a reverse CNOT gate $U_{CN'}$ using his qutrit $B$ as the control and the carrier $T$ as the target. The action of $U_{CN'}$ is defined by
\begin{equation}
U_{CN'} \ket{j}_B\otimes\ket{k}_T=\ket{j}_B \otimes \ket{(k-j)\bmod 3}_T.
\end{equation}
Applying $U_{CN'}$ to the joint state yields the final pre-measurement state
\begin{equation}
\label{eq:outputstate}
\frac{1}{\sqrt{3}}\sum_{j=0}^2\omega^{mj}\ket{j}_{A}\otimes\ket{(j+n)\bmod 3}_{B}\otimes\ket{(3-n)\bmod 3}_{T}.
\end{equation}

Finally, Bob measures the carrier in the $Z$-basis and classically communicates the outcome to Alice. If the measurement yields $0$, the protocol succeeds and they retain the shared pair; otherwise, they declare failure and discard the state.

\subsection{Two-Round Purification over Noiseless Channels}

Consider the ideal scenario where the ancillary carrier is transmitted through a noiseless channel with $\mathcal{N}=I$. Upon successful post-selection in the first round of the CAEPP, all $X$-errors with $n \in \{1,2\}$ are deterministically filtered out. The shared state thus reduces to
\begin{equation}
\rho_{\mathrm{out}}=\sum_{m=0}^2p'_{0m}\ket{\Phi^{0,m}}\bra{\Phi^{0,m}},
\end{equation}
with the normalized probabilities given by
\begin{equation}
p'_{0m}=\dfrac{p_{0m}}{\sum_{j=0}^2p_{0j}}.
\end{equation}
Because this intermediate state contains only pure phase errors, executing a second round of the CAEPP in the conjugate basis achieves complete purification. Consequently, conditioned on successful outcomes in both rounds, the retained shared state is the maximally entangled qutrit state $\ket{\Phi^{0,0}}$.

Figure \ref{fig:2round_noiseless} illustrates the corresponding quantum circuit. In Round $1$, the initial qutrits $A_{\text{in}}$ and $B_{\text{in}}$ undergo preprocessing operations $U_A^{\text{pre1}}$ and $U_B^{\text{pre1}}$, and interact with the first carrier initialized in state $\ket{0}_1$, yielding the partially purified state $\rho_{\text{out}}$ upon a successful measurement outcome of $\text{out}_1=0$. Round $2$ applies the subsequent preprocessing operations $U_A^{\text{pre2}}$ and $U_B^{\text{pre2}}$, and utilizes the second carrier initialized in state $\ket{0}_2$. Upon obtaining the second successful measurement outcome $\mathrm{out}_2=0$, Alice and Bob retain the maximally entangled qutrit state $\ket{\Phi^{0,0}}$.

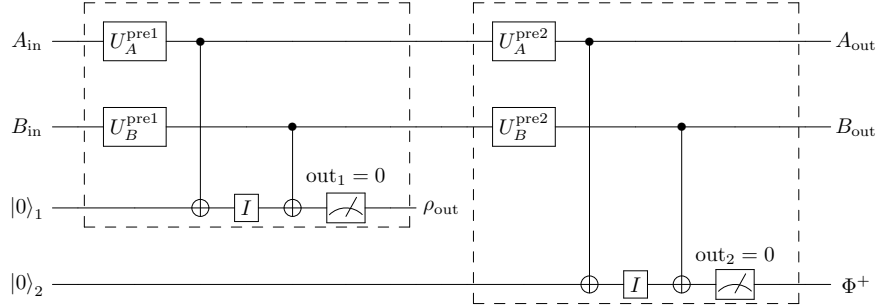
\begin{figure}[htbp]
    \centering
    \scalebox{0.8}{
    $
    \Qcircuit @C=1.2em @R=2.2em {
        A_{\text{in}} & \mbox{\hspace{0.2em}} & \qw & \gate{U_A^{\text{pre1}}} & \ctrl{2} & \qw & \qw & \qw & \qw & \qw & \qw & \qw & \gate{U_A^{\text{pre2}}} & \ctrl{3} & \qw & \qw & \qw & \qw & \qw & \qw & A_{\text{out}}\\ 
        B_{\text{in}} & \mbox{\hspace{0.2em}} & \qw & \gate{U_B^{\text{pre1}}} & \qw & \qw & \ctrl{1} & \qw & \qw & \qw & \qw & \qw & \gate{U_B^{\text{pre2}}} & \qw & \qw & \ctrl{2} & \qw & \qw & \qw & \qw & B_{\text{out}}\\
        \ket{0}_{1} & \mbox{\hspace{0.2em}} & \qw & \qw & \targ & \gate{I} & \targ & \meter \ar@{}[u]|-{\raisebox{-2.2em}{$\text{out}_1=0$}} & \qw & \qw & \rho_{\text{out}} & \mbox{} & \mbox{} & \mbox{} & \mbox{} & \mbox{} & \mbox{} & \mbox{} & \mbox{} & \mbox{}\\
        \ket{0}_{2} & \mbox{\hspace{0.2em}} & \qw & \qw & \qw & \qw & \qw & \qw & \qw & \qw & \qw & \qw & \qw & \targ & \gate{I} & \targ & \meter \ar@{}[u]|-{\raisebox{-2.2em}{$\text{out}_2=0$}} & \qw & \qw & \qw & \ket{\Phi^{0,0}}
        \gategroup{1}{4}{3}{9}{1.8em}{-} 
        \gategroup{1}{13}{4}{18}{1.8em}{-} 
    }
    $
    }
    \vspace{0.5em}
    \caption{Quantum circuit for two-round purification over noiseless channels}
    \label{fig:2round_noiseless}
    \vspace{0.3em}
    \small
    \centering
    \begin{minipage}{0.9\textwidth}
    \textbf{Notation Description:}
    \begin{itemize}
        \item $A_{\text{in}}, B_{\text{in}}$: Input qutrits of the initial noisy state shared between Alice and Bob.
        \item $\ket{0}_{1}, \ket{0}_{2}$: Initial states of the carrier states.
        \item $U_A^{\text{pre}k}, U_B^{\text{pre}k}$: preprocessing unitaries for Round $k$.
        \item $I$: Noiseless quantum channel for shared state transmission.
        \item $\text{out}_k=0$: Success condition of $Z$-basis measurement on the carrier for Round $k$.
        \item $\rho_{\text{out}}$: Partially purified state on shared states after Round $1$ success.
        \item $\ket{\Phi^{0,0}}$: The maximally entangled qutrit state retained after successful post-selection in both rounds.
    \end{itemize}
    \end{minipage}
\end{figure}

\subsection{Robustness under Noisy Carriers}

We now evaluate the CAEPP under practical conditions where the single-qutrit carrier is subjected to noise. Motivated by realistic communication setups, we assume the ancillary carrier traverses the same noisy channel $\mathcal{N}$ used to distribute the initial shared state. This scenario is mathematically modeled as two independent uses of the identical channel, represented by $\mathcal{N} \otimes \mathcal{N}$.

During initialization, Alice transmits one qutrit of a maximally entangled pair through $\mathcal{N}$, yielding the shared Bell-diagonal state defined in \eqref{eq:sigma1}, with $\mathbf{q}^{(0)}=\mathbf{p}$. For a fixed carrier channel $\mathcal{N}$, let $F_N:=q_{00}^{(N)}$ denote the fidelity conditioned on $N$ consecutive successful rounds of the single-carrier CAEPP. Whenever the conditional recurrence converges, we define its asymptotic convergent fidelity by 
\begin{equation} 
\label{eq:convergent_fidelity}
    F_\ast(\mathcal{N}):=\lim_{N\to\infty}q_{00}^{(N)}.  
\end{equation}

\subsection{Single-Carrier Protocol under Asymmetric Noise}
\label{subsec:singleCAEPPunderNoise}

We define the mixed error operator $Y=ZX$ through the cyclic-shift operator $X$ and the phase-shift operator $Z$,
\begin{align}
    &X\ket{j}=\ket{(j+1)\bmod3},\\
    &Z\ket{j}=\omega^j\ket{j},\\
    &Y\ket{j}=ZX\ket{j}=\omega^{j+1}\ket{(j+1)\bmod3}.
\end{align}
Then \eqref{eq:channel1} can be rewritten as
\begin{align}
    \nonumber
    \mathcal{N}(\cdot)=&p_{00}(\cdot)+p_{01}Z(\cdot)Z^{-1}+p_{02}Z^2(\cdot)Z^{-2}\\
    \nonumber
    &+p_{10}X(\cdot)X^{-1}+p_{20}X^2(\cdot)X^{-2}\\
    \nonumber
    &+p_{11}Y(\cdot)Y^{-1}+p_{22}Y^2(\cdot)Y^{-2}\\
    &+p_{12}ZY(\cdot)Y^{-1}Z^{-1}+p_{21}Z^2Y^2(\cdot)Y^{-2}Z^{-2}.
    \label{eq:XYZ}
\end{align}

As discussed in Section \ref{subsub:dotp}, a single round of testing cannot eliminate pure $Z$-errors with $n=0$. Consequently, we restrict our analysis to Pauli channels characterized by exactly seven non-zero probabilities by setting $p_{01}=p_{02}=0$. Letting the initial fidelity be $F_0 = p_{00}$, this configuration can be parameterized as
\begin{eqnarray}
\label{eq:p00}
    (p_{00},0,0,p_{10},p_{11},p_{12},p_{20},p_{21},p_{22})
\end{eqnarray}
where $\sum_{n=1}^2\sum_{m=0}^2 p_{nm}=1-p_{00}$.

We now evaluate the asymptotic performance of the single-carrier protocol in the absence of preprocessing, establishing its fundamental convergence threshold in the following proposition.

\begin{proposition}
\label{pr:inf}
    Consider a two-qutrit state transmitted with initial fidelity $p_{00}$ and error probabilities $p_{nm}$, where $p_{01} = p_{02} = 0$. In the absence of any preprocessing, the fidelity $F_N$ after $N$ successful rounds of the single-carrier protocol converges to $1$ as $N \to \infty$ if and only if
    \begin{equation}
        p_{00} > \max \{ p_{10} + p_{11} + p_{12}, p_{20} + p_{21} + p_{22} \}.
    \end{equation}
\end{proposition}

\begin{proof}
For $i \in \{1,2\}$, suppose the shared state has an error with $n=i$ and $m=j$, which occurs with probability $p_{ij}$. To survive the final measurement, the ancillary carrier must independently acquire an error with the identical shift index $n=i$. This occurs with the marginal $X$-error probability $p_{X=i} := \sum_{m=0}^2 p_{im}$. Let $P_i$ denote the probability that the protocol succeeds with a residual shift error $n=i$. The error-free branch requires both the shared state and the ancillary carrier to independently remain flawless, yielding $P_0 = p_{00}^2$. 

For error branches $i \in \{1,2\}$, survival dictates that the ancillary carrier must independently acquire a shift error $n=i$ identical to that of the shared state. Defining the marginal $X$-error probability as $p_{X=i} := \sum_{m=0}^2 p_{im}$, and summing over all possible initial phase errors $j$, we obtain,
\begin{equation}
    P_i = \sum_{j=0}^2 p_{ij} p_{X=i} = p_{X=i}^2.
\end{equation}
The post-selection fidelity $F_1$ is then simply the conditional probability,
\begin{equation}
\label{eq:fidelity2}
    F_1 = \frac{P_0}{P_0 + P_1 + P_2} = \frac{p_{00}^2}{p_{00}^2 + p_{X=1}^2 + p_{X=2}^2}.
\end{equation}

To evaluate the fidelity $F_2$ after the second round of CAEPP, we first express the residual error probabilities $p'_{nm}$ generated from the first round for $n \in \{1,2\}$ and $m \in \{0,1,2\}$ as
\begin{equation}
\label{eq:fidelity3}
    p_{nm}' = \frac{\sum_{j=0}^2 p_{n,j} p_{n,(m-j)\bmod 3}}{p_{00}^2 + p_{X=1}^2 + p_{X=2}^2}.
\end{equation}

In the second round, the protocol operates exclusively on the purified state obtained from the first round, assisted by a new ancillary carrier traversing the identical original channel. The updated fidelity is then given by
\begin{equation}
\label{eq:fidelity4}
    F_2 = \frac{p_{00}'p_{00}}{p_{00}'p_{00} + \sum_{n=1}^2 \left( \sum_{m=0}^2 p_{nm}' \right) p_{X=n}}.
\end{equation}

Substituting \eqref{eq:fidelity2} and \eqref{eq:fidelity3} into \eqref{eq:fidelity4} and performing the algebraic simplification, the fidelity after the second round compactly becomes
\begin{equation}
    F_2 = \frac{p_{00}^3}{p_{00}^3 + p_{X=1}^3 + p_{X=2}^3}.
\end{equation}

By mathematical induction, this analytical structure generalizes to the $N$-th round. The fidelity $F_N$ can then be directly expressed as
\begin{equation}
\label{eq:fidelity5}
    F_N = \frac{p_{00}^{N+1}}{p_{00}^{N+1} + p_{X=1}^{N+1} + p_{X=2}^{N+1}} = \left[ 1 + \left(\frac{p_{X=1}}{p_{00}}\right)^{N+1} + \left(\frac{p_{X=2}}{p_{00}}\right)^{N+1} \right]^{-1}.
\end{equation}

To guarantee an asymptotic convergence of $\lim_{N \to \infty} F_N = 1$, the exponential bases in the denominator must be strictly less than $1$, yielding the fundamental algebraic constraint
\begin{equation}
    p_{00} > \max\{ p_{X=1}, p_{X=2} \}.
\end{equation}
\end{proof}

This exposes a critical vulnerability to noise asymmetry. Under the shift-symmetric Pauli-noise condition $p_{X=1}=p_{X=2}$, which gives $p_{X=1}=p_{X=2}=(1-p_{00})/2$ under the restricted assumptions of Proposition~\ref{pr:inf}, the convergence condition reduces to $p_{00}>1/3$, as illustrated in Figs.~\ref{fig:f=0.33,A=0.5} and~\ref{fig:f=0.34,A=0.5}. However, an imbalance between the two marginal shift-error probabilities can violate the convergence condition even when $p_{00}>1/3$, as shown in Fig.~\ref{fig:f=0.34,A=0.48}. In the highly asymmetric limit $p_{X=1}\to1-p_{00}$, convergence requires $p_{00}>1/2$, as demonstrated in Fig.~\ref{fig:f=0.51,A=0.01}.

These observations motivate the MUB preprocessing introduced in Section~\ref{subsec:ada}. By aligning a maximum-weight phase-space direction with the primary purification axis, the resulting MUB-adapted mCAEPP establishes convergence for every $p_{00}>1/3$ within the qutrit Pauli-channel setting considered there.

\begin{figure}[htbp]
    \centering
    \includegraphics[width=0.8\textwidth]{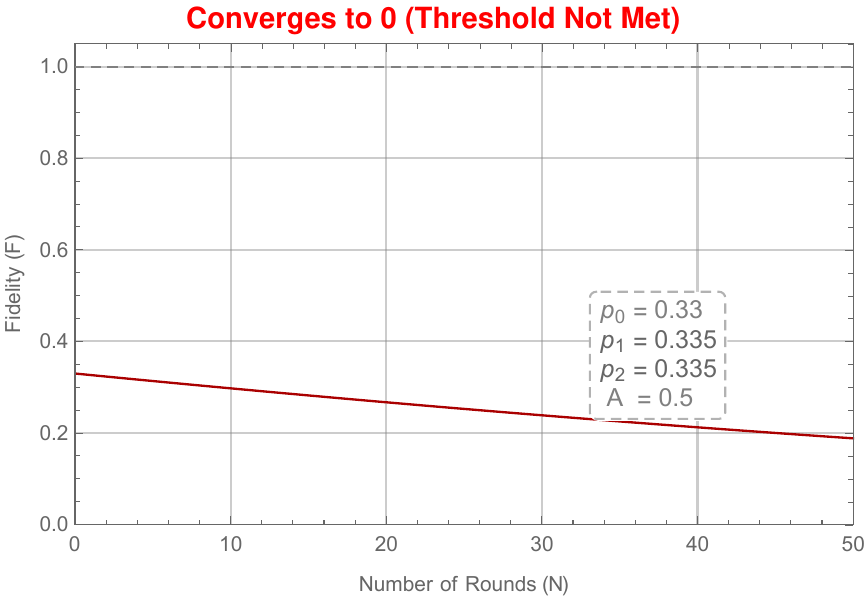} 
    \caption{Asymptotic fidelity $F_N$ of the single-carrier CAEPP without preprocessing. Inset parameters: $p_0 := p_{00}$, $p_{1,2} := p_{X=1,2}$ and $A := p_1/(p_1+p_2)$. For symmetric noise with $A = 0.5$, $p_0 = 0.33$, and $p_1 = p_2 = 0.335$, the convergence threshold $p_0 > \max\{p_1, p_2\}$ is violated, causing the fidelity to converge to $0$.}
    \label{fig:f=0.33,A=0.5}
\end{figure}

\begin{figure}[htbp]
    \centering
    \includegraphics[width=0.8\textwidth]{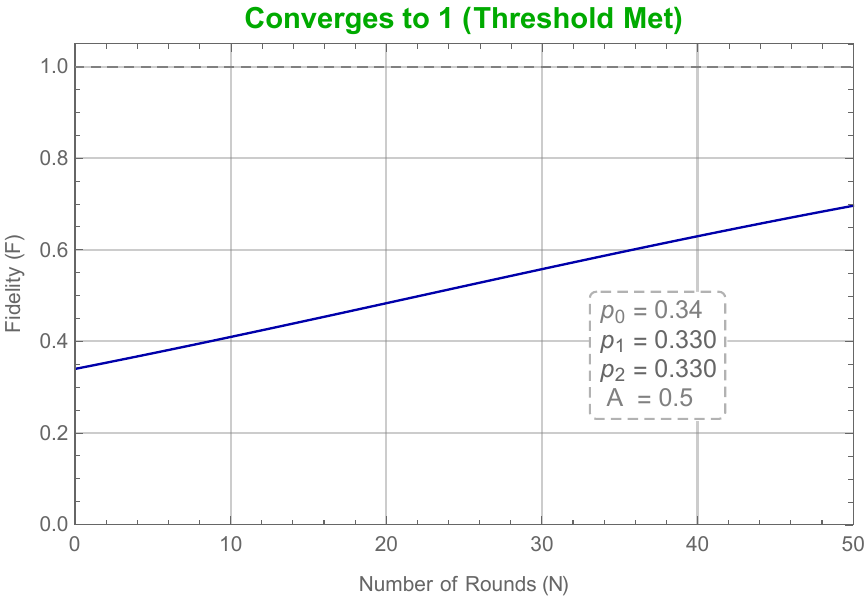} 
    \caption{Asymptotic fidelity $F_N$ without preprocessing. For symmetric noise with $A = 0.5$, $p_0 = 0.34$, and $p_1 = p_2 = 0.330$, the threshold condition $p_0 > \max\{p_1, p_2\}$ is satisfied, enabling the fidelity to asymptotically converge to $1$.}
    \label{fig:f=0.34,A=0.5}
\end{figure}

\begin{figure}[htbp]
    \centering
    \includegraphics[width=0.8\textwidth]{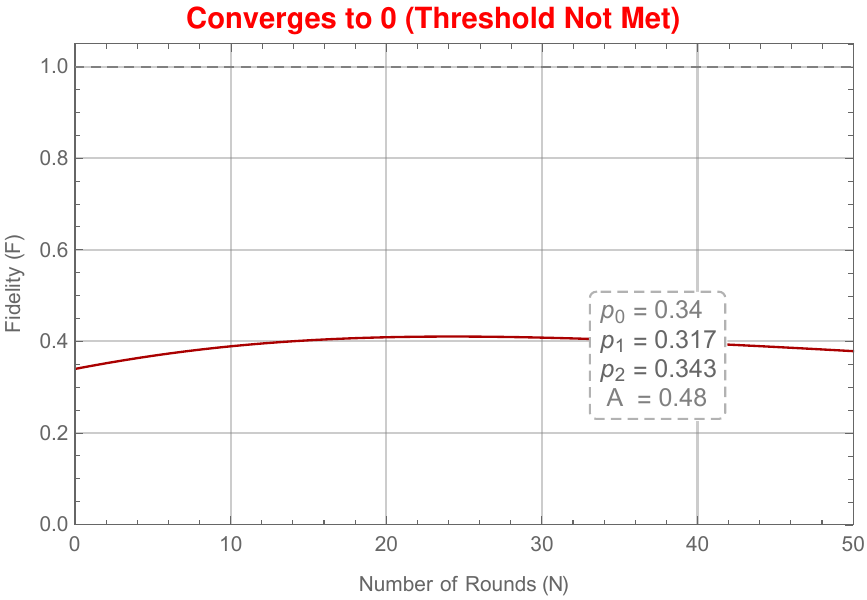} 
    \caption{Asymptotic fidelity $F_N$ without preprocessing. Under slightly asymmetric noise with $A = 0.48$ and the same initial $p_0 = 0.34$, the marginal errors shift to $p_1 = 0.317$ and $p_2 = 0.343$. The threshold condition $p_0 > \max\{p_1, p_2\}$ is consequently violated, causing the fidelity to converge to $0$ despite $p_0 > 1/3$.}
    \label{fig:f=0.34,A=0.48}
\end{figure}

\begin{figure}[htbp]
    \centering
    \includegraphics[width=0.8\textwidth]{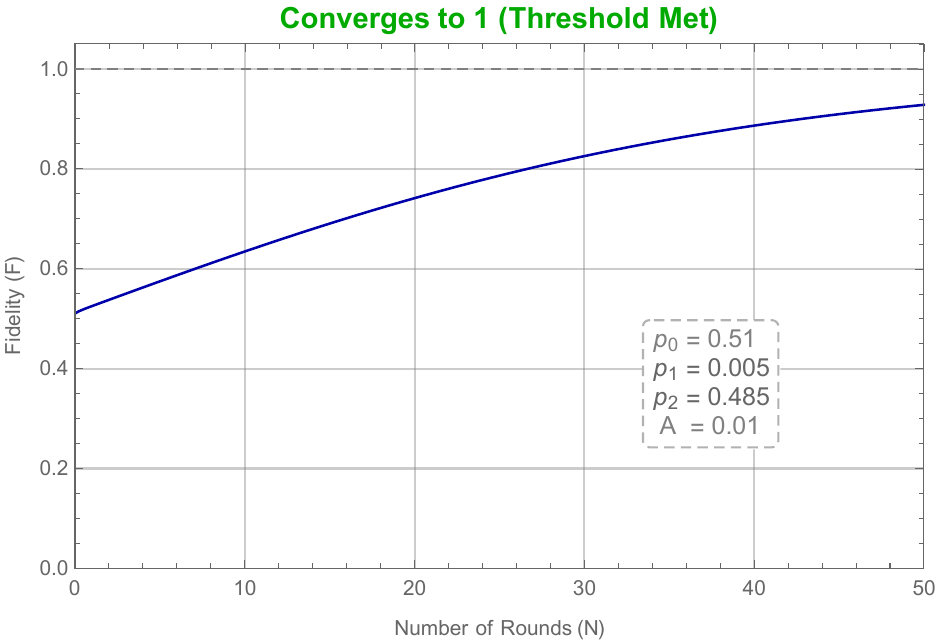} 
\caption{Asymptotic fidelity $F_N$ without preprocessing. Under highly asymmetric noise with $A = 0.01$ and $p_0 = 0.51$, the marginal errors are $p_1 = 0.005$ and $p_2 = 0.485$. The threshold condition $p_0 > \max\{p_1, p_2\}$ is satisfied, enabling the fidelity to asymptotically converge to $1$.}
    \label{fig:f=0.51,A=0.01}
\end{figure}

\section{Multi-Carrier Extensions and Stabilizer Code Implementations}
\label{sec:multi}

We now extend the single-carrier CAEPP of Section~\ref{sec:single} to $m$ carriers. The resulting mCAEPP uses a star stabilizer code to couple one shared qutrit pair to $m$ independently transmitted carriers. We first specify the stabilizer and encoding operations, illustrate the construction for $m=2$, and then state the general convergence result.

\subsection{Stabilizer Code for Qutrit Purification}

The purification protocol with multiple carriers can be constructed using a stabilizer code specified by a commuting set of multi-qutrit Pauli generators, each of which is a tensor product of local operators drawn from $\{I_3,Z^k,X^l,Z^kX^l\}$, where $k,l\in\{1,2\}$. In each round, Alice initializes the $m$ carrier qutrits in $\ket{0}^{\otimes m}$ and applies the encoding operation jointly to her qutrit of the shared pair and the carriers. She then transmits the carriers independently through the Pauli channel $\mathcal{N}$. Upon receiving them, Bob applies the corresponding decoding operation jointly to his qutrit of the shared pair and the carriers, and subsequently measures all carriers in the $Z$ basis.

Environmental noise during transmission induces quantum state alterations, which are mathematically characterized by the generalized qutrit Pauli operators $X$, $Y$, and $Z$. In the stabilizer formalism, a Pauli error $E$ is detected if it fails to commute with at least one generator $S_i$, i.e., $ES_i=\omega^kS_iE,\ \ \ k \bmod 3\ne0$ for some $i$. Errors commuting with all generators pass the stabilizer check and remain undetected.

\subsection{Encoding and Decoding Operations}

All scenarios discussed in this section are based on the three-dimensional case, i.e., qutrits. The carrier qutrits are initially prepared as $\ket{0}^{\otimes m}$, which is a code state of the stabilizer generators
\begin{equation}
\label{eq:stabilizer}
    \{Z_1,Z_2,\dots,Z_m\}.
\end{equation}
In this notation, $Z_i$ denotes the phase-shift operator $Z=\mathrm{diag}\{1,\omega,\omega^2\}$ acting exclusively on the $i$-th subsystem. The identity operator $I_3$ is implicitly applied to the remaining subsystems.

Here the code state is a common eigenvector of \eqref{eq:stabilizer}. Specifically, the single-qutrit basis state $\ket{0}$ is an eigenstate of the phase-shift operator $Z$ with eigenvalue $1$. Consequently, the tensor product state $\ket{0}^{\otimes m}$ remains invariant under the action of any local operator $Z_i$, making it a simultaneous eigenstate of the entire stabilizer group with eigenvalue $1$.

For a generator set $\{S_i\}$, there exists a unitary $U_{enc}$ such that
\begin{equation}
\label{eq:encoding_stabilizer_relation}
    U_{enc} Z_i U_{enc}^\dg = S_i \ \ \forall i \in \{1,2,\dots,m\}.
\end{equation}
Any unitary satisfying \eqref{eq:encoding_stabilizer_relation} is an encoding operation, and its adjoint serves as the corresponding decoding operation.

\begin{lemma}
\label{le:cnot1}
    The three-dimensional CNOT gate $U_{CN}$ satisfies
\begin{equation}
    U_{CN} (I_3 \otimes Z) U_{CN}^\dg = Z^2 \otimes Z.
\end{equation}
\end{lemma}

\begin{proof}
For a computational-basis vector $\ket{j,k}$, where $j$ and $k$ label the control and target qutrits, respectively, we have
    \begin{align}
    \nonumber
        U_{CN} (I_3 \otimes Z) U_{CN}^\dg \ket{j,k} &= U_{CN} (I_3 \otimes Z) (U_{CN}^\dg \ket{j,k})\\
        &= U_{CN} (\ket{j} \otimes Z \ket{(k-j) \bmod 3})\\
        \nonumber
        &= \omega^{k-j}\ket{j,k},
    \end{align}
    \begin{align}
        Z^2 \otimes Z\ket{j,k} &= Z^{-1} \otimes Z\ket{j,k}= \omega^{k-j}\ket{j,k}.
    \end{align}

Thus the two operators act identically on every computational-basis vector, proving the stated identity.
\end{proof}

For completeness, applying Lemma~\ref{le:cnot1} successively to a parallel family of forward SUM gates gives the following conjugation identity. 

\begin{lemma}
\label{le:cnot2}
    Let the first qutrit be the common control and let the $m$ carrier qutrits be the respective targets. Define
    \begin{equation}
        U_{\mathrm{par}}^{(m)} = U_{CN}^{(1)} \circ U_{CN}^{(2)} \circ\cdots\circ U_{CN}^{(m)},
    \end{equation}
    where
    \begin{equation}
        U_{CN}^{(i)} \ket{j,k_1,k_2,\dots,k_m} = \ket{j,k_1,\dots,(j+k_i)\bmod 3,\dots,k_m}.
    \end{equation}
    Then
    \begin{equation}
        U_{\mathrm{par}}^{(m)} \left( I_3\otimes\bigotimes_{i=1}^{m}Z_i \right) \left(U_{\mathrm{par}}^{(m)}\right)^\dg = Z_0^{2m}\otimes\bigotimes_{i=1}^{m}Z_i.
    \end{equation}
\end{lemma}

\begin{proof}
For any computational-basis state $\ket{j,k_1,k_2,\dots,k_m}$, the definition of $U_{\mathrm{par}}^{(m)}$ gives
\begin{equation}
    \left(U_{\mathrm{par}}^{(m)}\right)^\dg \ket{j,k_1,k_2,\dots,k_m} = \ket{j,(k_1-j)\bmod 3,(k_2-j)\bmod 3,\dots,(k_m-j)\bmod 3}.
\end{equation}
Therefore,
\begin{align}
    \nonumber
    &U_{\mathrm{par}}^{(m)} \left( I_3\otimes\bigotimes_{i=1}^{m}Z_i \right) \left(U_{\mathrm{par}}^{(m)}\right)^\dg \ket{j,k_1,k_2,\dots,k_m} \\
    \nonumber
    &\qquad= \omega^{\sum_{i=1}^{m}(k_i-j)} \ket{j,k_1,k_2,\dots,k_m} \\
    &\qquad= \omega^{\sum_{i=1}^{m}k_i-mj} \ket{j,k_1,k_2,\dots,k_m}.
\end{align}
On the other hand,
\begin{align}
    \nonumber
    &\left( Z_0^{2m}\otimes\bigotimes_{i=1}^{m}Z_i \right) \ket{j,k_1,k_2,\dots,k_m} \\
    \nonumber
    &\qquad= \omega^{2mj+\sum_{i=1}^{m}k_i} \ket{j,k_1,k_2,\dots,k_m} \\
    &\qquad= \omega^{\sum_{i=1}^{m}k_i-mj} \ket{j,k_1,k_2,\dots,k_m},
\end{align}
where the last equality follows from $\omega^{3mj}=1$. Since the two operators have the same action on every computational-basis state, the stated identity follows.
\end{proof}

\subsection{Case Study: Two-Carrier Instance}

To illustrate the mechanism of the three-dimensional mCAEPP, we consider a two-carrier qutrit protocol with $m=2$. We choose the stabilizer generators
\begin{equation}
\label{eq:X1X22}
    \{X_1X_2^2,\; Z_0Z_1Z_2\}.
\end{equation}
Here, $X_i$ and $Z_i$ denote the qutrit shift and phase operators acting on the $i$-th subsystem, respectively, with the identity applied to all other subsystems. The index $0$ refers to the shared qutrit, while indices $1$ and $2$ label the two carrier qutrits.

For these operators to form a valid stabilizer group, they must commute mutually. This is readily verified by explicitly calculating their product,
\begin{equation}
    X_1X_2^2Z_0Z_1Z_2 = \omega \cdot \omega^2 Z_0Z_1Z_2X_1X_2^2 = Z_0Z_1Z_2X_1X_2^2.
\end{equation}
Having established their commutativity, we now proceed to evaluate their error detection capabilities. 

The generator $S_1 = X_1X_2^2$ precisely detects single-carrier $Z$- and $Y$-type errors. The corresponding nontrivial commutation relations are explicitly given by
\begin{align}
    Z_1S_1 &= \omega S_1Z_1,\\
    Z_2S_1 &= \omega^2 S_1Z_2.
\end{align}
Given that $Y=ZX$, the detection signature of $Y$-errors is mathematically equivalent to that of pure $Z$-errors.

The generator $S_2=Z_0Z_1Z_2$ detects an error configuration whenever the sum of its shift exponents is nonzero modulo $3$. The corresponding nontrivial commutation relations generating the nontrivial phases are
\begin{align}
    S_2 X_i &= \omega X_iS_2, \quad &i&\in\{0,1,2\},\\
    S_2X_kX_l &= \omega^2 X_kX_lS_2, \quad &k,l&\in\{0,1,2\}, k\ne l.
\end{align}

\subsection{Scalability: Beyond Two Carriers}

\label{subsec:sca}

We now analyze the scalable mCAEPP for qutrit depolarizing channels above the non-entanglement-breaking threshold. The star-code encoder $U_{enc}$ considered below is defined independently by the circuit in Fig.~\ref{fig:m_carrier_caepp}; it is distinct from the parallel-SUM unitary $U_{\mathrm{par}}^{(m)}$ in Lemma~\ref{le:cnot2}.

To describe the corresponding multi-carrier stabilizer construction, we define the `star' stabilizer generators for $m$ carriers as
\begin{equation}
\label{eq:generators}
    \{X_1X_m^2, X_2X_m^2, \dots, X_{m-1}X_m^2, Z_0Z_1\cdots Z_m\}.
\end{equation}

The connection between the stabilizer generators in \eqref{eq:generators} and the circuit in Fig.~\ref{fig:m_carrier_caepp} follows directly from the conjugation of the initial carrier stabilizers in \eqref{eq:stabilizer}. For each $i=1,\ldots,m-1$, the gate $H_3^\dg$ first transforms $Z_i$ into $X_i$. The subsequent reverse CNOT gate $U_{CN'}$, with the $i$-th carrier as the control and the $m$-th carrier as the target, transforms it into $X_iX_m^2$. Similarly, the successive reverse CNOT gates $U_{CN'}$, with qutrits $0,1,\ldots,m-1$ as controls and the $m$-th carrier as the common target, transform $Z_m$ into $Z_0Z_1\cdots Z_m$. Therefore,
\begin{align}
    U_{enc}Z_iU_{enc}^\dg &= X_iX_m^2, \qquad i=1,\ldots,m-1,\\
    U_{enc}Z_mU_{enc}^\dg &= Z_0Z_1\cdots Z_m.
\end{align}
Thus, the encoding circuit maps the initial stabilizer set \eqref{eq:stabilizer} exactly onto the star stabilizer set \eqref{eq:generators}. Since Bob applies $U_{dec}=U_{enc}^\dg$, measuring all carriers in the $Z$ basis after decoding implements the corresponding stabilizer checks. The round is accepted when all measurement outcomes are zero.

The conditional convergence of the resulting protocol is stated in the following theorem.

\begin{theorem}
\label{thm:mcaepp_convergence}
    Let $q_{00}^{(n,m)}$ denote the fidelity of the shared state conditioned on $n$ consecutive successful rounds of the mCAEPP with $m$ carriers per round, where each round begins with the bilateral Clifford preprocessing $W_A^*\otimes W_B$, with $W=H_3^\dagger P_3^\dagger H_3$ and $P_3=\operatorname{diag}(1,1,\omega)$. For a qutrit depolarizing channel with initial fidelity $p_{00}>1/3$, the conditional recurrence converges for every fixed finite $m$. Moreover,
    \begin{equation}
    \label{eq:ordered_limit_theorem}
        \lim_{m\to\infty} \left( \lim_{n\to\infty}q_{00}^{(n,m)} \right) =1.
    \end{equation}
\end{theorem}

\begin{proof}
    The exact conditional recurrence, its convergence for fixed $m$, and the subsequent large-$m$ limit of its attracting fixed point are proved in Appendix~\ref{app:sca}.
\end{proof}

Theorem~\ref{thm:mcaepp_convergence} is a conditional asymptotic statement. Cumulative success probability and restart-aware finite-resource costs at a prescribed target fidelity are defined in Section~\ref{sec:finite_resource_comparison}.

\begin{figure}[htbp]
\centering
\resizebox{0.8\textwidth}{!}{
$
\Qcircuit @C=1.15em @R=2.2em {
A_{\mathrm{in}}
 & \mbox{\hspace{0.2em}}
 & \gate{W^*}
 & \qw & \qw & \qw & \qw & \qw
 & \ctrl{6}
 & \qw
 & \qw & \qw & \qw & \qw & \qw & \qw
 & \qw & \qw & \qw
\\
\ket{0}_{1}
 & \mbox{\hspace{0.2em}}
 & \gate{H_3^{\dagger}}
 & \ctrl{5} & \qw & \qw & \qw & \qw
 & \qw
 & \gate{\mathcal{N}}
 & \qw & \qw & \qw & \qw & \qw & \ctrl{5}
 & \gate{H_3}
 & \meter
   \ar@{}[u]|-{\raisebox{-2.2em}{
   $\mathrm{out}_1=0$}}
 & \qw
\\
\ket{0}_{2}
 & \mbox{\hspace{0.2em}}
 & \gate{H_3^{\dagger}}
 & \qw & \ctrl{4} & \qw & \qw & \qw
 & \qw
 & \gate{\mathcal{N}}
 & \qw & \qw & \qw & \qw & \ctrl{4} & \qw
 & \gate{H_3}
 & \meter
   \ar@{}[u]|-{\raisebox{-2.2em}{
   $\mathrm{out}_2=0$}}
 & \qw
\\
\ket{0}_{3}
 & \mbox{\hspace{0.2em}}
 & \gate{H_3^{\dagger}}
 & \qw & \qw & \ctrl{3} & \qw & \qw
 & \qw
 & \gate{\mathcal{N}}
 & \qw & \qw & \qw & \ctrl{3} & \qw & \qw
 & \gate{H_3}
 & \meter
   \ar@{}[u]|-{\raisebox{-2.2em}{
   $\mathrm{out}_3=0$}}
 & \qw
\\
\vdots
 & \mbox{\hspace{0.2em}}
 & \mbox{}
 & \mbox{} & \mbox{} & \mbox{} & \cdots & \mbox{}
 & \mbox{}
 & \vdots
 & \mbox{} & \mbox{} & \cdots
 & \mbox{} & \mbox{} & \mbox{}
 & \mbox{}
 & \vdots
 & \mbox{\hspace{2em}
   $\xrightarrow{\ \forall i:\,
   \mathrm{out}_i=0\ }\rho_{\mathrm{out}}$}
\\
\ket{0}_{m-1}
 & \mbox{\hspace{0.2em}}
 & \gate{H_3^{\dagger}}
 & \qw & \qw & \qw & \qw & \ctrl{1}
 & \qw
 & \gate{\mathcal{N}}
 & \qw & \ctrl{1} & \qw & \qw & \qw & \qw
 & \gate{H_3}
 & \meter
   \ar@{}[u]|-{\raisebox{-2.2em}{
   $\mathrm{out}_{m-1}=0$}}
 & \qw
\\
\ket{0}_{m}
 & \mbox{\hspace{0.2em}}
 & \qw
 & \gate{X^{-1}}
 & \gate{X^{-1}}
 & \gate{X^{-1}}
 & \qw
 & \gate{X^{-1}}
 & \gate{X^{-1}}
 & \gate{\mathcal{N}}
 & \targ
 & \targ
 & \qw
 & \targ
 & \targ
 & \targ
 & \qw
 & \meter
   \ar@{}[u]|-{\raisebox{-1.6em}{
   $\mathrm{out}_m=0$}}
 & \qw
\\
B_{\mathrm{in}}
 & \mbox{\hspace{0.2em}}
 & \gate{W}
 & \qw & \qw & \qw & \qw & \qw
 & \qw
 & \qw
 & \ctrl{-1}
 & \qw & \qw & \qw & \qw & \qw
 & \qw & \qw & \qw
}
$
}
\vspace{0.4em}

\caption{Explicit encoding and decoding circuit for the scalable mCAEPP with $m$ carriers. Alice applies $H_3^\dg$ to carriers $1,\ldots,m-1$, followed by the reverse qutrit CNOT gates $U_{CN'}$ that implement $U_{enc}$. After transmission through the carrier channels, Bob applies the inverse decoding sequence $U_{dec}=U_{enc}^\dg$ using the qutrit CNOT gates $U_{CN}$ and $H_3$, and then measures all carriers in the $Z$ basis.}
\label{fig:m_carrier_caepp}

\vspace{0.3em}
\small
\begin{minipage}{0.96\textwidth}
\textbf{Notation description.}
\begin{itemize}

\item $A_{\mathrm{in}}$ and $B_{\mathrm{in}}$ are the two qutrits of the initial noisy state shared by Alice and Bob.

\item $\ket{0}_{1},\ldots,\ket{0}_{m}$ are the initial carrier states.

\item The local preprocessing is the bilateral qutrit Clifford $W_A^*\otimes W_B$, where $P_3:=\operatorname{diag}(1,1,\omega)$ and $W:=H_3^\dagger P_3^\dagger H_3$. It induces the Bell-label shear $(s,t)\mapsto(s\oplus t,t)$.

\item $H_3^{\dagger}$ and $H_3$ are inverse qutrit Fourier transforms satisfying $H_3^{\dagger}ZH_3=X$.

\item A control dot connected to an $X^{-1}$ box denotes the reverse qutrit CNOT gate $U_{CN'}$, while a control dot connected to a circled-plus target denotes the qutrit CNOT gate $U_{CN}$.

\item $\mathcal{N}$ is the identical qutrit Pauli channel acting independently on each carrier.

\item The round is accepted only when $\mathrm{out}_i=0$ for every $i=1,\ldots,m$; $\rho_{\mathrm{out}}$ denotes the corresponding conditionally purified output state. For each fixed $m$, the accepted recurrence converges as the number of successful rounds tends to infinity, and the fidelity of its attracting fixed point approaches one as $m\to\infty$, as stated in Theorem~\ref{thm:mcaepp_convergence}.

\item The gates $W_A^*$ and $W_B$ are included in each complete round considered in Theorem~\ref{thm:mcaepp_convergence}. Omitting these two gates gives the raw star-stabilizer check $\mathcal C_Z^{(m)}$ used as a building block in the alternating-basis protocols of Theorems~\ref{thm:universal_purification} and~\ref{thm:prime_power_mcaepp}.

\end{itemize}
\end{minipage}
\end{figure}

\subsection{MUB-Guided preprocessing for General Pauli Channels}
\label{subsec:ada}

We now consider a memoryless qutrit Pauli channel with Bell-error probabilities $\{p_{xz}\}_{x,z\in\mathbb{Z}_3}$. Instead of completely averaging the eight nonidentity Pauli errors, we first apply an MUB-preserving pair symmetrization that equalizes the two opposite nonzero points on every MUB line while retaining the total weight of each line.

Define the single-qutrit Clifford operator $J$ by
\begin{equation}
    J\ket{j}=\ket{-j\bmod 3}.
\end{equation}
It satisfies
\begin{equation}
    JXJ^\dagger=X^{-1},
    \qquad
    JZJ^\dagger=Z^{-1}.
\end{equation}

For each shared pair, Alice and Bob jointly draw one uniformly random bit $r\in\{0,1\}$ and apply $J_A^r\otimes J_B^r$. For each carrier transmission, they jointly draw a new random bit, independently of all other carrier transmissions. Alice applies $J^r$ to the carrier immediately before the channel, and Bob applies $J^r$ immediately after the channel and before decoding. The two operations cancel in the absence of a channel error. After averaging over the shared random bit, the effective Pauli probabilities become
\begin{equation}
\label{eq:mub_pair_symmetrization}
    \widetilde p_{xz} = \frac{p_{xz}+p_{-x,-z}}{2},
\end{equation}
where all indices are evaluated modulo $3$. This operation preserves $p_{00}$ and the total probability carried by every MUB line.

The four MUB lines through the origin have weights
\begin{align}
    L_1&=p_{00}+p_{10}+p_{20},\notag\\
    L_2&=p_{00}+p_{01}+p_{02},\notag\\
    L_3&=p_{00}+p_{11}+p_{22},\notag\\
    L_4&=p_{00}+p_{12}+p_{21}.
\end{align}

\begin{lemma}
\label{lemma:mub_dominance}
If $p_{00}>1/3$, then
\begin{equation}
    L_{\max}:=\max\{L_1,L_2,L_3,L_4\}>\frac{1}{2}.
\end{equation}
\end{lemma}

\begin{proof}
The four MUB lines intersect only at $(0,0)$ and partition the eight nonzero points of $\mathbb{Z}_3^2$. Consequently,
\begin{equation}
    \sum_{i=1}^{4}L_i = 4p_{00}+\sum_{(x,z)\neq(0,0)}p_{xz} = 1+3p_{00}.
\end{equation}
If $p_{00}>1/3$, then $\sum_{i=1}^{4}L_i>2$, and hence at least one of the four line weights is strictly larger than $1/2$.
\end{proof}

The pair symmetrization in \eqref{eq:mub_pair_symmetrization} leaves all four line weights unchanged. Alice and Bob next apply a bilateral Clifford transformation $U_A\otimes U_B^*$ that maps a line attaining $L_{\max}$ to the line $L_1$ associated with $z=0$. The corresponding conjugation is also applied to every carrier channel use so that the carrier Pauli labels undergo the same phase-space permutation. After this alignment, the effective Pauli probabilities can be written as
\begin{equation}
\label{eq:mub_aligned_distribution}
    P=
    \begin{pmatrix}
        p_{00}&p_{01}&p_{02}\\
        p_{10}&p_{11}&p_{12}\\
        p_{20}&p_{21}&p_{22}
    \end{pmatrix}
    =
    \begin{pmatrix}
        a&c&c\\
        b&d&e\\
        b&e&d
    \end{pmatrix},
\end{equation}
where
\begin{equation}
    a=p_{00},
    \qquad
    a+2(b+c+d+e)=1,
    \qquad
    b=\max\{b,c,d,e\}.
\end{equation}
In particular, the aligned dominant line has weight
\begin{equation}
    a+2b=L_{\max}>\frac{1}{2}.
\end{equation}

Let $\mathcal C_Z^{(m)}$ denote one successful raw $m$-carrier star-stabilizer check in the aligned coordinate system. It consists of the carrier encoding, carrier transmission, decoding, and all-zero syndrome post-selection, and does not include the bilateral shear $W_A^*\otimes W_B$ used in Theorem~\ref{thm:mcaepp_convergence}. A complete adaptive cycle consists of a first check $\mathcal{C}_Z^{(m)}$, the bilateral Fourier rotation
\begin{equation}
    R_{AB}=H_{3,A}\otimes H_{3,B}^*,
\end{equation}
a second check $\mathcal{C}_Z^{(m)}$ using a fresh independent set of $m$ carriers, and the inverse rotation
\begin{equation}
    R_{AB}^{-1}=H_{3,A}^\dagger\otimes H_{3,B}^T.
\end{equation}
The inverse rotation returns the accepted output to the original Bell coordinate system. The complete cycle is accepted only when every carrier measurement in both checks gives outcome zero.

\begin{theorem}
\label{thm:universal_purification}
Let $\mathcal{N}$ be a memoryless qutrit Pauli channel with entanglement fidelity $p_{00}>1/3$. Apply the MUB-preserving pair symmetrization independently to the initial shared pair and to every carrier channel use, and map a maximum-weight MUB line to the aligned line $z=0$. Let $q_{00}^{(n,m)}$ be the fidelity conditioned on $n$ consecutive successful complete cycles, with $m$ fresh carriers used in each of the two checks of every cycle. Then there exists an integer $m_0$ such that, for every fixed $m\geq m_0$, the conditional recurrence starting from the pre-processed channel state converges to a fixed point $\mathbf{q}_*^{(m)}$. Moreover,
\begin{equation}
\label{eq:ordered_limit_general_channel}
    \lim_{m\to\infty} \left( \lim_{n\to\infty}q_{00}^{(n,m)} \right) = \lim_{m\to\infty}q_{*,00}^{(m)} = 1.
\end{equation}
\end{theorem}

The proof of Theorem \ref{thm:universal_purification} is given in Appendix \ref{app:ada}. The conclusion is conditional on successful post-selection and uses the order of limits displayed in \eqref{eq:ordered_limit_general_channel}. It does not assert unit fidelity for an arbitrary fixed finite carrier number.

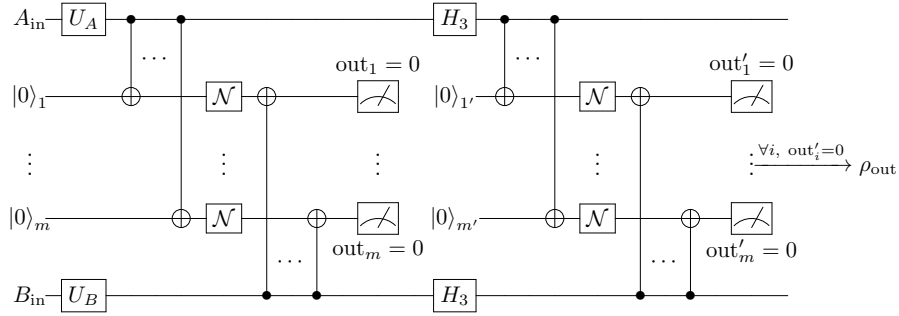
\begin{figure}[htbp]
\centering
\hspace*{3em}
\resizebox{0.8\columnwidth}{!}{
$
\Qcircuit @C=1.25em @R=2.9em {
\lstick{A_{\mathrm{in}}}
& \gate{U_A}
& \multigate{2}{\mathcal{C}_Z^{(m)}}
& \qw
& \qw
& \qw
& \qw
& \gate{H_3}
& \multigate{2}{\mathcal{C}_Z^{(m)}}
& \qw
& \qw
& \qw
& \gate{H_3^\dagger}
& \qw
& \qw
& \mbox{}
& \mbox{}
\\
\lstick{\ket{0}_T^{\otimes m}}
& \qw
& \ghost{\mathcal{C}_Z^{(m)}}
& \qw
& \qw
& \meter
\ar@{}[u]|-{\raisebox{-2.0em}{$\forall i:\mathrm{out}_i=0$}}
& \push{\ket{0}_{T'}^{\otimes m}}
& \qw
& \ghost{\mathcal{C}_Z^{(m)}}
& \qw
& \qw
& \meter
\ar@{}[u]|-{\raisebox{-2.0em}{$\forall i:\mathrm{out}'_i=0$}}
& \mbox{}
& \mbox{}
& \mbox{}
& \mbox{}
& \hspace{-2em}\rho_{\mathrm{out}}
\\
\lstick{B_{\mathrm{in}}}
& \gate{U_B^*}
& \ghost{\mathcal{C}_Z^{(m)}}
& \qw
& \qw
& \qw
& \qw
& \gate{H_3^{*}}
& \ghost{\mathcal{C}_Z^{(m)}}
& \qw
& \qw
& \qw
& \gate{H_3^{T}}
& \qw
& \qw
& \mbox{}
& \mbox{}
\gategroup{1}{15}{3}{15}{0.6em}{\}}
}
$
}
\vspace{0.5em}
\caption{Quantum circuit for one alternating-basis cycle of the adaptive mCAEPP. The one-time local Clifford operations $U_A$ and $U_B^*$ align the maximum-weight MUB axis $L_{\max}$ with the primary axis $L_1$. The first block $\mathcal C_Z^{(m)}$ performs the raw star-stabilizer check corresponding to the encoding, carrier-transmission, decoding, and syndrome-post-selection part of Fig.~\ref{fig:m_carrier_caepp}; it excludes the bilateral shear $W_A^*\otimes W_B$ used in Theorem~\ref{thm:mcaepp_convergence}. Conditioned on all first-stage syndrome outcomes being zero, the bilateral Fourier rotation $R_{AB}=H_{3,A}\otimes H_{3,B}^*$ exchanges the shift- and phase-error axes. The same raw check is then repeated using a fresh set of $m$ carriers. Finally, the inverse rotation $R_{AB}^{-1}=H_{3,A}^{\dagger}\otimes H_{3,B}^{T}$ returns the accepted state to the original Bell coordinate system. The output is retained only when both sets of syndrome measurements are zero.}
\label{fig:ada_mcaepp}
\end{figure}

\subsection{Clifford-Twirled Extension}

Complete single-qutrit Clifford twirling provides an alternative fully symmetric preprocessing strategy. Let $\mathrm{Cliff}_3$ denote the single-qutrit Clifford group modulo global phases. For a memoryless qutrit Pauli channel $\mathcal{N}$, define its Clifford-twirled channel by
\begin{equation}
    \overline{\mathcal{N}}(\cdot):=\frac{1}{|\mathrm{Cliff}_3|}\sum_{C\in\mathrm{Cliff}_3}C^\dagger\mathcal{N}\!\left(C(\cdot)C^\dagger\right)C.
\label{eq:clifford_twirl}
\end{equation}

\begin{theorem}
\label{thm:clifford_twirled_mcaepp}
Let $\mathcal{N}$ be a memoryless qutrit Pauli channel with entanglement fidelity $p_{00}>1/3$. Apply the Clifford twirl in \eqref{eq:clifford_twirl} independently to the channel use employed to distribute the initial shared pair and to every subsequent carrier transmission. Then the effective channel is the qutrit depolarizing channel
\begin{equation}
\label{eq:twirled_depolarizing_channel}
    \overline{\mathcal{N}}(\cdot)=p_{00}(\cdot)+\frac{1-p_{00}}{8}\sum_{\substack{x,z\in\mathbb{Z}_3\\(x,z)\neq(0,0)}}Z^zX^x(\cdot)X^{-x}Z^{-z}.
\end{equation}
After the Clifford twirl, Alice and Bob apply in every purification round the bilateral Clifford preprocessing $W_A^*\otimes W_B$, where $W=H_3^\dagger P_3^\dagger H_3$, before performing the $m$-carrier check. For every fixed finite carrier number $m$, the conditional mCAEPP resulting recurrence generated by $\overline{\mathcal{N}}$ converges to a fixed point $\mathbf{q}_*^{(m)}$. Moreover,
\begin{equation}
\label{eq:twirled_unit_fidelity}
    \lim_{m\to\infty}\left(\lim_{n\to\infty}q_{00}^{(n,m)}\right)=\lim_{m\to\infty}q_{*,00}^{(m)}=1.
\end{equation}
Consequently, for every prescribed target fidelity $F_{\mathrm{tar}}<1$, there exist finite integers $m$ and $N$ such that $q_{00}^{(N,m)}\geq F_{\mathrm{tar}}$, with a strictly positive cumulative success probability.
\end{theorem}

\begin{proof}
For every $C\in\mathrm{Cliff}_3$, conjugation by $C$ maps each qutrit Pauli operator to another qutrit Pauli operator up to an irrelevant phase,
\begin{equation}
\label{eq:clifford_pauli_permutation}
    C^\dagger Z^zX^xC=e^{\mathrm{i}\theta_C(x,z)}Z^{z'}X^{x'}.
\end{equation}
The identity label $(0,0)$ is fixed, whereas the Clifford action is transitive on the eight nonidentity labels in $\mathbb{Z}_3^2\setminus\{(0,0)\}$~\cite{Gross2007UnitaryDesigns}. Therefore, uniform averaging over $\mathrm{Cliff}_3$ leaves the identity-error probability $p_{00}$ unchanged and distributes the total nonidentity probability $1-p_{00}$ uniformly among the remaining eight Pauli errors. This proves \eqref{eq:twirled_depolarizing_channel}. The effective channel thus satisfies exactly the depolarizing-channel assumptions of Theorem~\ref{thm:mcaepp_convergence}. Since the recurrence specified above also includes the same per-round bilateral preprocessing $W_A^*\otimes W_B$, the fixed-$m$ convergence and the ordered limit in \eqref{eq:twirled_unit_fidelity} follow directly from Theorem~\ref{thm:mcaepp_convergence}. Finally, since the attracting fixed-point fidelity approaches one as $m\to\infty$, for every $F_{\mathrm{tar}}<1$ one can first choose a finite $m$ such that $q_{*,00}^{(m)}>F_{\mathrm{tar}}$ and then choose a finite number of successful rounds $N$ such that $q_{00}^{(N,m)}\geq F_{\mathrm{tar}}$. Each of these finitely many rounds has a nonzero success probability, so the corresponding cumulative success probability is strictly positive.
\end{proof}

Equation~\eqref{eq:clifford_twirl} defines the target Pauli-label map by a full Clifford average. For the finite-resource comparison, we implement the same label map by uniformly sampling the eight Clifford representatives induced by the following $Q_8$ matrices, whose Clifford lifts exist by the symplectic--Clifford correspondence \cite{Hostens2005Clifford}:
\begin{equation}
    Q_8=\{\pm I,\pm A,\pm B,\pm AB\}\subset SL(2,\mathbb F_3),
    \quad
    A=\begin{pmatrix}0&1\\2&0\end{pmatrix},
    \quad
    B=\begin{pmatrix}1&1\\1&2\end{pmatrix}.
\end{equation}
Direct multiplication gives $A^2=B^2=-I$ and $AB=-BA$, so the displayed eight matrices form a subgroup isomorphic to $Q_8$. Every nonidentity element is either $-I$ or has square $-I$. The element $-I$ fixes no nonzero vector over $\mathbb F_3$. For any remaining nonidentity element $G$, if $Gv=v$, then $G^2=-I$ gives $v=G^2v=-v$, and hence $v=0$. Therefore, the action on the eight nonzero labels is free. Since both the group and the set of nonzero labels have eight elements, the action is regular. The uniform $Q_8$ average therefore maps every nonzero Pauli label to the uniform distribution and produces the same Pauli-label map as the full Clifford average in \eqref{eq:clifford_twirl}. A fresh independent three-bit shared seed is used for every randomized channel use, including the use that distributes the initial shared pair, while Alice and Bob use the same seed within that use. The three-bit value is the cost of this specified exact implementation; no channel-specific minimum is claimed.

\begin{remark}
    Theorem~\ref{thm:universal_purification} and Theorem~\ref{thm:clifford_twirled_mcaepp} provide two alternative conditional purification strategies under the same fidelity assumption $p_{00}>1/3$. For the specific Clifford twirl considered here, the identity-error probability $p_{00}$ is preserved, and the resulting Choi state belongs to the qutrit isotropic family. This family contains no bound-entangled interval: the twirled state is separable for $p_{00}\leq1/3$, whereas it is NPT and distillable for $p_{00}>1/3$~\cite{Horodecki1999Teleportation,Horodecki1997Reduction}. Hence, within the parameter regime assumed in this work, Clifford twirling cannot convert the initial state into an undistillable bound-entangled state, and no stronger fidelity threshold is required to preserve distillability. This conclusion is specific to the fidelity-preserving Clifford twirl and does not imply that arbitrary noise-tailoring maps preserve entanglement. Complete twirling replaces all eight nonidentity Pauli probabilities by $(1-p_{00})/8$ and therefore erases the directional information contained in the original asymmetric error distribution. In contrast, the MUB-preserving pair symmetrization retains the four MUB-line weights and permits a maximum-weight error axis to be aligned with the primary purification axis. The equal-target finite-resource comparison below gives opposite cost orderings among the audited deterministic profiles under the stated $2\leq m\leq100$ optimization, so neither strategy uniformly dominates the other on this audited set.
\end{remark}

\subsection{Equal-target finite-resource comparison}
\label{sec:finite_resource_comparison}

We restrict the comparison to $F_{\mathrm{tar}}>p_{00}$. For a candidate that reaches the target, let $N$ be the first index such that $q_{00}^{(N,m)}\geq F_{\mathrm{tar}}$; candidates for which no such finite $N$ exists are regarded as infeasible. Explicitly,
\begin{equation}
    q_{00}^{(j,m)}<F_{\mathrm{tar}}\quad(0\leq j<N),
    \qquad 
    q_{00}^{(N,m)}\geq F_{\mathrm{tar}}.
\end{equation}
For MUB preprocessing this condition is checked separately for every legal maximum-line alignment, but no alignment parameter is attached to the Clifford protocol. Let $R_0=1$, and let $R_n$ be the probability that one attempt reaches check $n$ for $\mathsf C$, or complete cycle $n$ for $\mathsf M$. Let $s_n$ be the conditional success probability of that check or complete cycle. Thus $R_{n+1}=R_ns_n$, and
\begin{equation}
\label{eq:total_success_probability}
    P_{\mathrm{tot}}^{(N,m)} :=R_N =\prod_{n=0}^{N-1}s_n.
\end{equation}
When the protocol must be distinguished, we write $P_{\mathrm{tot},\mathsf C}^{(N,m)}$ or $P_{\mathrm{tot},\mathsf M}^{(N,m)}$ for the corresponding value. If $s_{n,1}$ is the conditional success probability of the first check in MUB cycle $n$, early abort gives
\begin{equation}
\label{eq:finite_attempt_costs}
    E_{\mathrm{att},\mathsf C}=m\sum_{n=0}^{N-1}R_n,
    \qquad
    E_{\mathrm{att},\mathsf M}=m\sum_{n=0}^{N-1}R_n(1+s_{n,1}).
\end{equation}
After any failure the pair is discarded and the protocol restarts from a newly distributed input pair. Hence
\begin{equation}
\label{eq:finite_restart_costs}
    \overline C_{\mathrm{car},\mathsf P} =\frac{E_{\mathrm{att},\mathsf P}} {P_{\mathrm{tot},\mathsf P}^{(N,m)}},
    \qquad
    \overline C_{\mathrm{all},\mathsf P} =\frac{1+E_{\mathrm{att},\mathsf P}} {P_{\mathrm{tot},\mathsf P}^{(N,m)}},
\end{equation}
where $\mathsf P\in\{\mathsf C,\mathsf M\}$ and the additional one is the initial-pair distribution in every attempt. For the specified universal, exact, independently randomized per-use implementations,
\begin{equation}
\label{eq:finite_shared_bits}
    B_{\mathsf M}=\overline C_{\mathrm{all},\mathsf M},
    \qquad
    B_{\mathsf C}=3\overline C_{\mathrm{all},\mathsf C}.
\end{equation}
These are implementation-specific charges, not unconditional minimum randomness costs for a known channel.

We optimize the three reported costs separately over every integer $2\leq m\leq100$ and, for MUB, over all legal alignments of a maximum-weight line. The three objectives select the same canonical parameters in the deterministic cases shown below; equality of the carrier and total-use selections is a certified data-set fact, while the fixed implementation factors in \eqref{eq:finite_shared_bits} make the shared-bit and total-use minimizers coincide within each protocol. Clifford $N$ counts single checks, while MUB $N$ counts complete two-check cycles. All costs include early abort and full restart, and equal target fidelity does not require equal attained fidelity. In the table header, $P_{\mathrm{tot}}$ is shorthand for the protocol-specific $P_{\mathrm{tot},\mathsf P}^{(N,m)}$ at the reported $(m^*,N^*)$.

Order the qutrit MUB lines as $L_1:z=0$, $L_2:x=0$, $L_3:z=x$, and $L_4:z=-x$, and write $w_i=(L_i-p_{00})/(1-p_{00})$. The inversion-symmetric inputs in Table~\ref{tab:finite_resources_corrected} are determined by
\begin{equation}
    p_{10}=p_{20}=\frac{(1-p_{00})w_1}{2},\quad
    p_{01}=p_{02}=\frac{(1-p_{00})w_2}{2},\quad
    p_{11}=p_{22}=\frac{(1-p_{00})w_3}{2},\quad
    p_{12}=p_{21}=\frac{(1-p_{00})w_4}{2}.
\end{equation}
The four profiles are deliberately selected witnesses rather than a statistical sample.

\begin{table}[t]
\centering
\caption{Certified equal-target comparison at $F_{\rm tar}=0.99$ for four deliberately selected deterministic witnesses. Costs are expected per accepted output; the optimization domain, restart model, meanings of $N$, and randomization model are specified in the surrounding text.}
\label{tab:finite_resources_corrected}
\small
\setlength{\tabcolsep}{4pt}
\renewcommand{\arraystretch}{1.08}
\begin{tabular}{ccclrrrrr}
\hline
$p_{00}$ & $(w_1,w_2,w_3,w_4)$ & Protocol & $(m^*,N^*)$ &
$F_{\rm out}$ & $P_{\rm tot}$ & $\overline C_{\rm car}$ &
$\overline C_{\rm all}$ & $B$\\
\hline
0.7 & $(0.55,0.15,0.15,0.15)$ & Clifford & $(3,4)$ & 0.992526957 & $1.04663\times10^{-2}$ & 402.398 & 497.942 & 1493.83\\
    &                             & MUB      & $(3,4)$ & 0.993297822 & $2.62427\times10^{-4}$ & 16953.8 & 20764.4 & 20764.4\\
\hline
0.7 & $(0.45,0.45,0.05,0.05)$ & Clifford & $(3,4)$ & 0.992526957 & $1.04663\times10^{-2}$ & 402.398 & 497.942 & 1493.83\\
    &                           & MUB      & $(2,3)$ & 0.991795496 & $1.09129\times10^{-2}$ & 316.672 & 408.306 & 408.306\\
\hline
0.7 & $(0.25,0.25,0.25,0.25)$ & Clifford & $(3,4)$ & 0.992526957 & $1.04663\times10^{-2}$ & 402.398 & 497.942 & 1493.83\\
    &                           & MUB      & $(3,2)$ & 0.990098227 & $1.04922\times10^{-2}$ & 401.512 & 496.821 & 496.821\\
\hline
0.9 & $(0.25,0.25,0.25,0.25)$ & Clifford & $(2,2)$ & 0.997151877 & 0.592643 & 5.91309 & 7.60044 & 22.8013\\
    &                           & MUB      & $(2,1)$ & 0.997151877 & 0.592643 & 5.91309 & 7.60044 & 7.60044\\
\hline
\end{tabular}
\end{table}

The first two profiles have the same $p_{00}=0.7$ but opposite cost orderings, which rules out a uniform ordering among the audited deterministic profiles under the stated $2\leq m\leq100$ optimization. The third profile gives nearly equal carrier costs, and the fourth gives equal channel-use costs but different charges under the two specified randomization implementations. For the already depolarizing fourth input, this randomness difference is not unavoidable if redundant channel-specific randomization is omitted. No conclusion is drawn here for $m>100$, arbitrary protocols, random channel ensembles, or experimental gate and memory costs.

\section{Extension to Higher Dimensions}
\label{sec:higher}

In this section, we first extend the algebraic single-carrier CAEPP and its marginal-error convergence condition to arbitrary integer dimensions $d\geq2$ using the generalized qudit Pauli operators defined over indices modulo $d$. We then develop a separate MUB-adapted multi-carrier construction for prime-power dimensions. The latter requires a finite-field formulation of the Pauli operators, Bell states, carrier operations, and discrete phase space.

\subsection{Generalization of Algebraic and State Spaces}

The single-qudit generalized Pauli operators are constructed from the cyclic-shift operator $X$ and the phase-shift operator $Z$. Their actions on the computational basis states $\ket{j}$ for $j \in \{0, 1, \dots, d-1\}$ are governed by $X\ket{j} = \ket{(j+1) \bmod d}$ and $Z\ket{j} = \omega^j\ket{j}$ where the phase factor is $\omega = e^{2\pi i/d}$. These operators obey the canonical commutation relation $ZX = \omega XZ$. 

The bipartite state space is spanned by $d^2$ maximally entangled generalized Bell states. We define these basis states uniformly as 
\begin{equation}
    \ket{\Phi^{n,m}} = \dfrac{1}{\sqrt{d}} \sum_{j=0}^{d-1} \omega^{mj} \ket{j, (j+n) \bmod d}
\end{equation}
with indices $n, m \in \{0, 1, \dots, d-1\}$. A general $d$-dimensional Bell-diagonal state is therefore expressed as 
\begin{equation}
    \rho = \sum_{n,m=0}^{d-1} p_{nm} \ket{\Phi^{n,m}} \bra{\Phi^{n,m}}
\end{equation}
under the normalization condition $\sum_{n,m=0}^{d-1} p_{nm} = 1$. The initial entanglement fidelity corresponds exactly to the target state probability $F(\rho) = p_{00}$.

The bilateral operations utilize the modulo-$d$ SUM gate 
\begin{equation}
    U_{\mathrm{SUM}}\ket{j, k} = \ket{j, (j+k) \bmod d}.
\end{equation}

\subsection{Convergence Bottleneck under Asymmetric Noise}

Following the established framework in Section \ref{subsec:singleCAEPPunderNoise}, we evaluate the asymptotic performance of the qudit single-carrier protocol to demonstrate how environmental noise asymmetry limits purification. This analysis directly extends the three-dimensional failure threshold established in Proposition \ref{pr:inf} to arbitrary $d$-dimensional Hilbert spaces.

\begin{proposition}
\label{pr:inf_generalized}
    Consider a memoryless $d$-dimensional Pauli channel, with arbitrary integer $d\geq2$, that is used both to distribute the initial shared Bell pair and to transmit every carrier. Let its Pauli-error probabilities be $\{p_{xz}\}_{x,z=0}^{d-1}$, and suppose that $p_{0z}=0$ for every $z\neq0$. Define the marginal shift-error probabilities by
    \begin{equation}
        p_{X=x}:=\sum_{z=0}^{d-1}p_{xz}.
    \end{equation}
    In the absence of preprocessing, the fidelity $F_N$ conditioned on $N$ consecutive successful rounds of the single-carrier CAEPP satisfies
    \begin{equation}
    \label{eq:generalized_single_carrier_fidelity}
        F_N=\frac{p_{00}^{N+1}}{p_{00}^{N+1}+\displaystyle\sum_{x=1}^{d-1}p_{X=x}^{N+1}}.
    \end{equation}
    Consequently,
    \begin{equation}
        \lim_{N\to\infty}F_N=1
    \end{equation}
    if and only if
    \begin{equation}
    \label{eq:generalized_single_carrier_condition}
        p_{00}>\max_{x\in\{1,\ldots,d-1\}}p_{X=x}.
    \end{equation}
\end{proposition}

\begin{proof}
Because $p_{0z}=0$ for every $z\neq0$, the marginal probability of the zero-shift sector is $p_{X=0}=p_{00}$. In one successful round, a shared-state component with shift label $x$ is retained precisely when the independently transmitted carrier acquires the same shift label $x$. Summing over all phase labels, the unnormalized total weight of the shift sector $x$ is therefore multiplied by $p_{X=x}$ in every successful round. Since the initial shared state is distributed through the same channel, induction over the number of successful rounds shows that the unnormalized weight of the sector $x$ after $N$ rounds is $p_{X=x}^{N+1}$. Normalizing these $d$ sector weights gives
\begin{equation}
F_N=\frac{p_{X=0}^{N+1}}{\displaystyle\sum_{x=0}^{d-1}p_{X=x}^{N+1}}=\left[1+\sum_{x=1}^{d-1}\left(\frac{p_{X=x}}{p_{00}}\right)^{N+1}\right]^{-1},
\end{equation}
which proves \eqref{eq:generalized_single_carrier_fidelity}. If $p_{00}>p_{X=x}$ for every $x\neq0$, every ratio in the denominator converges to zero and hence $F_N\to1$. Conversely, if $p_{X=x}\geq p_{00}$ for some $x\neq0$, the corresponding term does not converge to zero, so $F_N$ cannot converge to $1$. This proves \eqref{eq:generalized_single_carrier_condition}.
\end{proof}

Proposition~\ref{pr:inf_generalized} shows that the marginal-shift-error bottleneck persists in arbitrary integer dimensions. The construction below addresses this bottleneck in prime-power dimensions by combining finite-field MUB alignment with an axis-preserving multiplicative symmetrization.

\subsection{MUB-Adapted mCAEPP in Prime-Power Dimensions}
\label{subsec:prime_power_mcaepp}

The generalized qudit formulation based on indices modulo $d$ applies to arbitrary integer dimensions. We now restrict attention to prime-power dimensions $d=\ell^r$, where $\ell$ is prime and $r$ is a positive integer. In this case, the computational basis is indexed by the elements of the finite field $\mathbb{F}_d$, and the construction below uses finite-field arithmetic rather than arithmetic modulo $d$. A complete set of $d+1$ mutually unbiased bases exists in every prime-power dimension~\cite{WoottersFields1989,GibbonsHoffmanWootters2004}.

For $u\in\mathbb{F}_d$, define the field trace and the corresponding additive character by
\begin{equation}
    \operatorname{Tr}(u):=\operatorname{Tr}_{\mathbb{F}_d/\mathbb{F}_\ell}(u)=u+u^\ell+\cdots+u^{\ell^{r-1}},
\end{equation}
and
\begin{equation}
\label{eq:finite_field_character}
    \chi(u):=\exp\left[\frac{2\pi\mathrm{i}}{\ell}\operatorname{Tr}(u)\right],
\end{equation}
where the value of $\operatorname{Tr}(u)\in\mathbb{F}_\ell$ is identified with an integer in $\{0,\ldots,\ell-1\}$. The finite-field shift and phase operators are defined by
\begin{equation}
\label{eq:finite_field_pauli}
    X(a)\ket{u}=\ket{u+a},
    \qquad
    Z(b)\ket{u}=\chi(bu)\ket{u},
    \qquad
    a,b,u\in\mathbb{F}_d.
\end{equation}
They satisfy
\begin{equation}
    Z(b)X(a)=\chi(ab)X(a)Z(b).
\end{equation}
The finite-field additive-character and Pauli-operator formulation applies to both odd and even prime-power dimensions~\cite{Durt2005PrimePowerMUB}. A finite-field Pauli channel therefore takes the form
\begin{equation}
\label{eq:finite_field_pauli_channel}
    \mathcal{N}(\rho)=\sum_{x,z\in\mathbb{F}_d}p_{xz}X(x)Z(z)\rho Z(z)^\dagger X(x)^\dagger,
    \qquad
    \sum_{x,z\in\mathbb{F}_d}p_{xz}=1.
\end{equation}

The corresponding generalized Bell basis is
\begin{equation}
\label{eq:finite_field_bell_states}
    \ket{\Phi^{s,t}}=\frac{1}{\sqrt d}\sum_{u\in\mathbb{F}_d}\chi(tu)\ket{u,u+s},
    \qquad
    s,t\in\mathbb{F}_d.
\end{equation}
In particular, transmission of the second subsystem of $\ket{\Phi^{0,0}}$ through \eqref{eq:finite_field_pauli_channel} produces the Bell-diagonal state
\begin{equation}
    \rho=\sum_{s,t\in\mathbb{F}_d}p_{st}\ket{\Phi^{s,t}}\bra{\Phi^{s,t}},
\end{equation}
whose entanglement fidelity is $p_{00}$.

The carrier encoding and decoding operations use the finite-field SUM gate
\begin{equation}
\label{eq:finite_field_sum}
    U_{\mathrm{SUM}}\ket{u,v}=\ket{u,u+v},
    \qquad
    u,v\in\mathbb{F}_d.
\end{equation}
Its inverse acts as
\begin{equation}
    U_{\mathrm{SUM}}^{-1}\ket{u,v}=\ket{u,v-u}.
\end{equation}
Together with computational-basis measurements whose outcomes are labeled by $\mathbb{F}_d$, these operations implement the finite-field generalization of the carrier syndrome checks used in the qutrit mCAEPP.

Appendix~\ref{app:prime_power_mcaepp} gives the complete state propagation for this raw star check, including the syndrome signs and accepted Bell label.

In the finite-field phase-space formulation, the $d+1$ MUB directions are represented by the $d+1$ lines through the origin in $\mathbb{F}_d^2$~\cite{GibbonsHoffmanWootters2004},
\begin{equation}
\label{eq:finite_field_mub_lines}
    \mathcal{L}_\infty:=\{(0,z):z\in\mathbb{F}_d\},
    \qquad
    \mathcal{L}_a:=\{(x,ax):x\in\mathbb{F}_d\},
    \quad
    a\in\mathbb{F}_d.
\end{equation}
Any two distinct lines in \eqref{eq:finite_field_mub_lines} intersect only at $(0,0)$, and every nonzero point of $\mathbb{F}_d^2$ belongs to exactly one of them. Let $L_\infty$ and $L_a$ denote the corresponding sums of Pauli-error probabilities. Their total satisfies
\begin{equation}
\label{eq:finite_field_line_sum}
    L_\infty+\sum_{a\in\mathbb{F}_d}L_a=1+dp_{00}.
\end{equation}

We select a line of maximum weight and map it to the axis $z=0$. For $\mathcal L_a$ and $\mathcal L_\infty$, respectively, the required phase-space maps are
\begin{equation}
\label{eq:finite_field_line_alignment_matrices}
    S_a=\begin{pmatrix}1&0\\-a&1\end{pmatrix},\qquad
    S_\infty=\begin{pmatrix}0&1\\-1&0\end{pmatrix},
\end{equation}
which satisfy $S_a(x,ax)^{\mathsf T}=(x,0)^{\mathsf T}$ and $S_\infty(0,z)^{\mathsf T}=(z,0)^{\mathsf T}$. The standard symplectic--Clifford correspondence supplies the corresponding unitary implementation; the finite-field embedding, including even characteristic, is given in Appendix~\ref{app:prime_power_mcaepp} \cite{Hostens2005Clifford}. The same label map is applied to the initial pair and every carrier use, so $p_{00}$ is unchanged. We relabel the aligned probabilities as $\{p_{xz}\}$.

The line alignment alone does not provide the symmetry needed for the multi-carrier spectral analysis. After the alignment, we therefore apply an additional multiplicative symmetrization that preserves the axes $x=0$ and $z=0$. For every $\lambda\in\mathbb{F}_d^\times$, define the permutation unitary
\begin{equation}
\label{eq:multiplicative_unitary}
    V_\lambda\ket{u}=\ket{\lambda u}.
\end{equation}
It conjugates the finite-field Pauli operators according to
\begin{equation}
\label{eq:multiplicative_pauli_action}
    V_\lambda X(x)V_\lambda^\dagger=X(\lambda x),
    \qquad
    V_\lambda Z(z)V_\lambda^\dagger=Z(\lambda^{-1}z).
\end{equation}
Consequently, the corresponding Bell-error labels are permuted as
\begin{equation}
\label{eq:multiplicative_label_action}
    (x,z)\longmapsto(\lambda x,\lambda^{-1}z).
\end{equation}
In particular,
\begin{equation}
    \mathcal L_a\longmapsto\mathcal L_{\lambda^{-2}a}.
\end{equation}
Thus this symmetrization generally permutes, rather than separately preserves, the non-axis MUB-line weights. The properties used below are the preserved aligned-axis weight and the equalized total probabilities of the nonzero $z$-columns.

Alice and Bob use shared classical randomness to choose $\lambda$ uniformly from $\mathbb{F}_d^\times$ and apply the corresponding coordinated bilateral operation to the shared Bell pair. The corresponding channel conjugation is applied independently to every carrier transmission. After averaging over $\lambda$, the effective Pauli probabilities are
\begin{equation}
\label{eq:multiplicative_symmetrization}
    \widetilde p_{xz}=\frac{1}{d-1}\sum_{\lambda\in\mathbb{F}_d^\times}p_{\lambda^{-1}x,\lambda z}.
\end{equation}
This operation fixes the origin and preserves the total probability of the aligned axis $z=0$. It also makes all nonzero points on that axis equiprobable and equalizes the total probabilities of the $d-1$ nonzero $z$-columns.

Let
\begin{equation}
\label{eq:prime_power_a_b}
    a:=\widetilde p_{00}=p_{00},
    \qquad
    b:=\widetilde p_{x0}\quad\text{for every }x\neq0,
\end{equation}
and define
\begin{equation}
\label{eq:prime_power_alpha_beta_gamma}
    \alpha:=a+(d-1)b,
    \qquad
    \beta:=a-b,
    \qquad
    \gamma:=\frac{1-\alpha}{d-1}.
\end{equation}
Here $\alpha$ is the total weight of the aligned axis, whereas
\begin{equation}
\label{eq:prime_power_column_weight}
    \sum_{x\in\mathbb{F}_d}\widetilde p_{xz}=\gamma
    \qquad
    \text{for every }z\neq0.
\end{equation}
The difference between the dominant nontrivial Fourier coefficient and the remaining column weight is
\begin{equation}
\label{eq:prime_power_spectral_inequality}
    \beta-\gamma =a-b-\frac{1-a-(d-1)b}{d-1} =\frac{da-1}{d-1}.
\end{equation}
Therefore,
\begin{equation}
\label{eq:prime_power_beta_gamma_order}
    p_{00}>\frac{1}{d}
    \quad
    \Longrightarrow
    \quad
    \beta>\gamma.
\end{equation}
If $b=0$, the maximality of the aligned line implies that every nonzero phase-space probability vanishes, so the channel is noiseless. In every nontrivial noisy case, $b>0$, and equations \eqref{eq:prime_power_alpha_beta_gamma} and \eqref{eq:prime_power_beta_gamma_order} give
\begin{equation}
\label{eq:prime_power_parameter_order}
    0\leq\gamma<\beta<\alpha.
\end{equation}

Define the finite-field Fourier transformation by
\begin{equation}
\label{eq:finite_field_fourier}
    H_d\ket{u}=\frac{1}{\sqrt d}\sum_{v\in\mathbb{F}_d}\chi(uv)\ket{v}.
\end{equation}
The bilateral operation
\begin{equation}
    R_{AB}:=H_{d,A}\otimes H_{d,B}^*
\end{equation}
acts on the Bell labels as
\begin{equation}
\label{eq:finite_field_fourier_bell_action}
    R_{AB}\ket{\Phi^{s,t}}=\chi(-st)\ket{\Phi^{t,-s}}.
\end{equation}
The phase factor in \eqref{eq:finite_field_fourier_bell_action} has no effect on Bell-state probabilities. One complete purification cycle consists of an $m$-carrier check, the bilateral Fourier transformation $R_{AB}$, a second check using $m$ fresh carriers, and the inverse transformation $R_{AB}^{-1}$.

\begin{theorem}
\label{thm:prime_power_mcaepp}
    Let $d=\ell^r$ be a prime-power dimension, and let $\mathcal{N}$ be a memoryless finite-field Pauli channel with entanglement fidelity $p_{00}>1/d$. Apply a finite-field Clifford transformation that maps a maximum-weight MUB line to the axis $z=0$, followed by the multiplicative symmetrization in \eqref{eq:multiplicative_symmetrization}, to the initial shared pair and to every carrier channel use. Let $q_{00}^{(N,m)}$ denote the fidelity conditioned on $N$ consecutive successful complete cycles, with $m$ fresh carriers used in each of the two checks of every cycle. Then there exists an integer $m_0$ such that, for every fixed $m\geq m_0$, the conditional recurrence converges to a fixed point $\mathbf{q}_*^{(m)}$. Moreover,
    \begin{equation}
    \label{eq:prime_power_ordered_limit}
        \lim_{m\to\infty}\left(\lim_{N\to\infty}q_{00}^{(N,m)}\right) =\lim_{m\to\infty}q_{*,00}^{(m)} =1.
    \end{equation}
    Consequently, for every prescribed target fidelity $F_{\mathrm{tar}}<1$, there exist finite integers $m$ and $N$ such that $q_{00}^{(N,m)}\geq F_{\mathrm{tar}}$, with a strictly positive cumulative success probability.
\end{theorem}

The proof of Theorem~\ref{thm:prime_power_mcaepp} is given in Appendix~\ref{app:prime_power_mcaepp}. The proof first derives the exact finite-field convolution generated by one successful $m$-carrier check. Equation~\eqref{eq:prime_power_parameter_order} then makes the contributions from the nonzero $z$-columns asymptotically smaller than the leading correction generated by the aligned column. Combining two checks related by \eqref{eq:finite_field_fourier_bell_action} produces a unique dominant Bell-sector eigendirection associated with $\ket{\Phi^{0,0}}$. Finite-dimensional matrix perturbation and Perron--Frobenius analysis then give the fixed-$m$ convergence and the ordered limit in \eqref{eq:prime_power_ordered_limit}.

For $d=3$, the multiplicative group is $\mathbb{F}_3^\times=\{1,-1\}$. The multiplicative symmetrization therefore reduces to the pair symmetrization used in the qutrit protocol, and the condition $p_{00}>1/d$ reproduces the qutrit threshold $p_{00}>1/3$. For composite dimensions that are not prime powers, the finite-field phase-space construction used above is unavailable. Theorem~\ref{thm:prime_power_mcaepp} therefore does not directly apply to those dimensions.

\section{Conclusion}

We have studied carrier-assisted entanglement purification in high-dimensional systems subject to general asymmetric Pauli noise. The main lesson is that noise asymmetry is not merely a quantitative perturbation of the symmetric case; it can create a convergence obstruction for the unprocessed single-carrier recurrence analyzed here. In particular, even when the target Bell component exceeds the basic qutrit threshold, unfavorable marginal error distributions may prevent single-carrier purification from converging. The MUB-guided preprocessing introduced here addresses this obstruction through randomized inversion-pair symmetrization and deterministic maximum-line alignment. For every qutrit Pauli channel with $p_{00}>1/3$, the accepted recurrence converges for every sufficiently large fixed carrier number, and its fixed-point fidelity approaches one in the subsequent large-carrier limit.

The higher-dimensional analysis further separates two aspects of the protocol. The algebraic carrier-assisted construction and the associated marginal-error bottleneck extend naturally to arbitrary qudit dimensions. By contrast, the high-dimensional preprocessing uses the finite-field line structure available in prime-power dimensions, followed by multiplicative symmetrization. This shows that high-dimensional purification is not obtained by a direct dimensional replacement of qubit or qutrit protocols; the discrete phase-space structure of the Bell-error indices becomes part of the protocol design. The qutrit inversion average preserves every MUB-line weight, whereas the general prime-power multiplicative symmetrization preserves the aligned-axis weight and equalizes the nonzero column totals while generally permuting the other MUB lines. This structure yields the sufficient condition $p_{00}>1/d$ and the stated ordered large-$m$ limit; composite dimensions that are not prime powers require different ideas.

At equal target fidelity, the certified deterministic comparison over all legal maximum-line alignments and $2\leq m\leq100$ gives witnesses with opposite finite-resource orderings. Thus neither the MUB-adapted strategy nor complete Pauli-label twirling uniformly dominates the other among the audited deterministic profiles under the stated optimization; the preferred strategy depends on the full Pauli-error profile. The shared-bit comparison refers to the specified exact, independently randomized per-use implementations and is not an unconditional minimum for a known channel.

Several issues remain open. For composite dimensions that are not prime powers, a complete set of $d+1$ MUBs is not generally available, so alternative preprocessing principles may be needed. It is also important to move beyond Pauli-diagonal noise and examine whether analogous adaptive strategies can be developed for continuous or experimentally reconstructed noise models. Finally, the present thresholds and finite-resource comparison assume ideal local operations and perfect syndrome extraction. Quantifying the tolerance to noisy gates, imperfect measurements, storage errors, and experimental implementation costs will be necessary before the protocol can be assessed in realistic high-dimensional quantum-network settings; a fully certified random-ensemble comparison is also left for future work.

\section*{Acknowledgment} 

ZHS and LC were supported by the NNSF of China (Grant No. 12471427). The authors thank Joonwoo Bae for his helpful feedback and stimulating discussions.

\appendix

\section{Proof of Theorem~\ref{thm:mcaepp_convergence}}

\label{app:sca}

Fix the number of carriers $m$. Let the Bell-error distribution of the shared state after $n$ consecutive successful rounds be
\begin{equation}
\label{eq:qnm_definition}
    \mathbf{q}^{(n,m)} =
    \big(
    q_{00}^{(n,m)},q_{01}^{(n,m)},q_{02}^{(n,m)},
    q_{10}^{(n,m)},q_{11}^{(n,m)},q_{12}^{(n,m)},
    q_{20}^{(n,m)},q_{21}^{(n,m)},q_{22}^{(n,m)}
    \big)^T.
\end{equation}
For notational simplicity, the superscripts $(n,m)$ are suppressed throughout the derivation of a single conditional round.

At the preprocessing step, Alice and Bob apply the bilateral
qutrit Clifford operation
\begin{equation}
    W_A^*\otimes W_B,
    \qquad
    W:=H_3^\dagger P_3^\dagger H_3,
    \qquad
    P_3:=\operatorname{diag}(1,1,\omega).
\end{equation}
The conjugation action of $W$ on the qutrit Pauli generators is
\begin{equation}
    W X W^\dagger=X,
    \qquad
    W Z W^\dagger=XZ.
\end{equation}
Consequently,
\begin{equation}
    W X^sZ^tW^\dagger
    \doteq X^{s\oplus t}Z^t,
\end{equation}
where $\doteq$ denotes equality up to an irrelevant global phase.
Under the Bell-state convention used in this paper, the bilateral
operation therefore induces the symplectic shear
\begin{equation}
    S:(s,t)\longmapsto(s\oplus t,t),
    \qquad
    S=
    \begin{pmatrix}
        1&1\\
        0&1
    \end{pmatrix}
    \in SL(2,\mathbb F_3).
\end{equation}
Hence the probability distribution after preprocessing is the
pullback of this label transformation,
\begin{equation}
\label{eq:r=q}
    r_{s,t}=q_{s\ominus t,t},
\end{equation}
where $\ominus$ denotes subtraction modulo $3$. Explicitly,
\begin{align}
\nonumber
    \mathbf r
    &=(r_{00},r_{01},r_{02},
       r_{10},r_{11},r_{12},
       r_{20},r_{21},r_{22})\\
    &=(q_{00},q_{21},q_{12},
       q_{10},q_{01},q_{22},
       q_{20},q_{11},q_{02}).
\end{align}

The Pauli errors acting on the shared state and the \(m\) carriers can be specified by a \(2(m+1)\)-component label vector over \(\mathbb{Z}_3\),
\begin{equation}
    \boldsymbol{e} := (x_0,z_0,x_1,z_1,\ldots,x_m,z_m) \in \mathbb{Z}_3^{2(m+1)}.
\end{equation}
The corresponding joint error operator is
\begin{equation}
\label{eq:errorstring1}
    E(\boldsymbol{e}) = E_0\otimes E_1\otimes\cdots\otimes E_m,
    \qquad
    E_j=Z_j^{z_j}X_j^{x_j},
    \quad j=0,1,\ldots,m.
\end{equation}
We write $E:=E(\boldsymbol e)$ when no ambiguity arises.

For $S_i=X_iX_m^2,i=1,\dots,m-1$ in \eqref{eq:generators}, $ES_i=S_iE$ holds provided that the following condition is satisfied,
\begin{align}
&\omega^{z_i+2z_m}=\omega^0, \\
    \Leftrightarrow&z_i+2z_m\equiv0\mod3,\\
    \Leftrightarrow&z_i=z_m \mod3.
\end{align}

For $S_m=Z_0Z_1 \cdots Z_m$, $ES_m=S_mE$ holds provided that the following condition is satisfied,
\begin{equation}
\omega^{\sum_{j=0}^mx_j}=\omega^0\Leftrightarrow\sum_{j=0}^mx_j\equiv0\mod3.
\end{equation}

To sum up, the general conditions for the success of the purification protocol are as follows,
\begin{equation}
\label{condi:succ1}
    \begin{cases}
        z_i=z_m \mod3, i=1,2,\dots,m-1.\\
        \sum_{j=0}^mx_j\equiv0\mod3.
    \end{cases}
\end{equation}

Each of the $m$ carrier qutrits is transmitted independently through the same qutrit depolarizing channel with $p:=p_{00}$. Therefore, for every $k=1,\ldots,m$, the carrier error pair $(x_k,z_k)$ has the joint probability mass function
\begin{equation}
\label{eq:carrier_distribution}
    P_k(i,j) := P(x_k=i,z_k=j) =
    \begin{cases}
        p, & (i,j)=(0,0),\\
        \dfrac{1-p}{8}, & \text{otherwise}.
    \end{cases}
\end{equation}
The carrier error pairs $\{(x_k,z_k)\}_{k=1}^{m}$ are independent and identically distributed. They are also independent of the error pair $(x_0,z_0)$ associated with the current shared state, whose distribution after preprocessing is
\begin{equation}
\label{eq:shared_distribution}
    P_0(i,j) := P(x_0=i,z_0=j) = r_{ij}.
\end{equation}
In general, $P_0$ need not coincide with $P_k$.

To evaluate the success probability $P_{\mathrm{succ}}$ and the conditional output distribution after Round $n+1$, define
\begin{equation}
    \mathbf{q}' :=
    \mathbf{q}^{(n+1,m)} =
    (q'_{00},q'_{01},q'_{02},
    q'_{10},q'_{11},q'_{12},
    q'_{20},q'_{21},q'_{22})^T.
\end{equation}
we introduce a success indicator $\mathbf{1}_{\mathrm{succ}}$ defined over the error strings specified in \eqref{eq:errorstring1}. This indicator takes the value $1$ if the round ends in \textbf{Success}, and $0$ otherwise. 

We can now proceed to construct the three-dimensional success indicator $\mathbf{1}_{\mathrm{succ}}$ by \eqref{condi:succ1},
\begin{equation}
\label{eq:indi4}
    \mathbf{1}_{\mathrm{succ}} = \left( \prod_{i=1}^{m-1} \delta_0(z_i-z_m)\right) \cdot \delta_0(x_0+x_1+\cdots+x_m).
\end{equation}

Afterwards, based on the properties of the indicator function, to calculate the probability of protocol success, we only need to compute the expectation of $\mathbf{1}_{\mathrm{succ}}$,
\begin{equation}
    P_{\mathrm{succ}}=\mathbb{E}[\mathbf{1}_{\mathrm{succ}}]=P(\bigcap_{i=1}^{m-1}\{ z_i=z_m \mod3 \} \cap \{ \sum_{k=0}^m x_k\equiv0\mod3 \}).
\end{equation}

We will calculate $P_{\mathrm{succ}}$ in the following steps.

(i) Variable distribution analysis; 

(ii) Decomposition of the total probability formula;

(iii) Simplify the expression;

(iv) Calculate each part separately;

(v) Calculate $P_{\mathrm{succ}}$.

\vspace{1cm}

Here is the first step.

(i) Variable distribution analysis. 

For the carrier error pairs $\{(x_k,z_k)\}_{k=1}^{m}$, the common joint probability mass function is given by \eqref{eq:carrier_distribution}. We refer to the marginal pmf as the sum of the joint pmf over all possible values of another random variable.  

The marginal pmf of $z_k$ for $c\in\{0,1,2\}$ is
\begin{align}
\label{eq:pmfc0}
    C_c:=P(z_k=c) = \sum_{i=0}^2P_k(i,c)=
    \begin{cases}
        \dfrac{3p+1}{4}, &c=0\\
        \dfrac{3(1-p)}{8}, &c=1,2\\
    \end{cases} 
\end{align}
and the marginal pmf of $x_k$ for $t\in\{0,1,2\}$ is
\begin{align}
\label{eq:pmft0}
    T_t:=P(x_k=t) = \sum_{j=0}^2P_k(t,j)=
    \begin{cases}
        \dfrac{3p+1}{4}, &t=0\\
        \dfrac{3(1-p)}{8}, &t=1,2\\
    \end{cases} 
\end{align}

Let the conditional pmf $r_c(t)=P(x_k=t|z_k=c)$ for $c,t\in\{0,1,2\}$, then
\begin{align}
\label{eq:vda2}
    r_c(t)=
    \begin{cases}
        r_0(t)=
        \begin{cases}
            r_0(0)=\dfrac{4p}{3p+1},\\
            r_0(1)=\dfrac{1-p}{2(3p+1)},\\
            r_0(2)=\dfrac{1-p}{2(3p+1)},
        \end{cases}\\
        \\
        r_1(t)=\dfrac{1}{3},\ t\in\{0,1,2\},\\
        \\
        r_2(t)=\dfrac{1}{3},\ t\in\{0,1,2\}.
    \end{cases}
\end{align}

(ii) Decomposition of the total probability formula.

\begin{align}
    P_{\mathrm{succ}} &=P(\bigcap_{i=1}^{m-1}\{ z_i=z_m \mod3 \} \cap \{ \sum_{k=0}^m x_k\equiv0\mod3 \})\\
    &= \sum_{c=0}^2 \left( \prod_{j=1}^m P(z_j=c) P( \sum_{k=0}^m x_k\equiv0\mod3|z_1=\cdots=z_m=c)\right)\\
    &=\sum_{c=0}^2 C_c^m P( \sum_{k=0}^m x_k\equiv0\mod3|z_1=\cdots=z_m=c).
\end{align}

(iii) Simplify the expression.

\begin{align}
    &P( \sum_{k=0}^m x_k\equiv0\mod3|z_1=\cdots=z_m=c)\\
    &=\sum_{t_0=0}^2 \left[P(x_0=t_0)  \sum_{\mathbf{t} \in \mathbb{Z}_3^m} \delta_0(\sum_{j=0}^mt_j)\prod_{i=1}^mP(x_i=t_i|z_i=c)\right],
\end{align}
where $\mathbf{t} = (t_1, \dots, t_m)$ and
\begin{align}
    \delta_0(\sum_{j=0}^mt_j)&=\dfrac{1}{3}\sum_{k=0}^2\omega^{k \sum_{j=0}^mt_j}=\dfrac{1}{3}\sum_{k=0}^2\prod_{j=0}^m\omega^{k \cdot t_j}.
\end{align}

Then
\begin{align}
    P_{\mathrm{succ}} &= \sum_{c=0}^2 C_c^m P \left( \sum_{k=0}^m x_k\equiv0\mod3|z_1=\cdots=z_m=c \right)\\
    &= \sum_{c=0}^2 C_c^m \sum_{t_0=0}^2 \left[P(x_0=t_0)  \sum_{\mathbf{t} \in \mathbb{Z}_3^m}\left(\dfrac{1}{3}\sum_{k=0}^2\prod_{j=0}^m \omega^{k t_j} \prod_{i=1}^mP(x_i=t_i|z_i=c)\right)\right]\\
    &=\dfrac{1}{3}\sum_{c=0}^2 C_c^m \sum_{k=0}^2\sum_{t_0=0}^2P(x_0=t_0)\omega^{kt_0} \sum_{\mathbf{t} \in \mathbb{Z}_3^m}\left[\prod_{i=1}^m\omega^{k t_i}P(x_i=t_i|z_i=c)\right],
\end{align}
where the inner sum simplifies as
\begin{align}
     \sum_{\mathbf{t} \in \mathbb{Z}_3^m}\left[\prod_{i=1}^m\omega^{k t_i}P(x_i=t_i|z_i=c)\right]&=\prod_{i=1}^m\sum_{t_i=0}^2\left[\omega^{k t_i}P(x_i=t_i|z_i=c)\right]\\
    &=\left[\sum_{t=0}^2\omega^{k t}P(x=t|z=c)\right]^m.
\end{align}

For a random variable $Y$ taking values in $\{0,1,2\}$, the modulo-3 characteristic function is defined as
\begin{equation}
    \phi_Y(k):=\mathbb{E}(\omega^{k Y})=\sum_{y=0}^2P(Y=y)\omega^{k y}, k\in \{0,1,2\}.
\end{equation}

Then, recognizing the definitions within our equation, we obtain
\begin{equation}
\label{eq:psucc}
    P_{\mathrm{succ}}=\dfrac{1}{3}\sum_{c=0}^2C_c^m \sum_{k=0}^2 \phi_{x_0}(k) [\phi_{x|z=c}(k)]^m.
\end{equation}

(iv) Calculate each part separately.

Using the specific values of $C_c$ from \eqref{eq:pmfc0}, we evaluate the carrier conditional characteristic function, 
\begin{equation}
    \phi_{x|z=c}(k) = \frac{1}{C_c}\sum_{t=0}^2 P_k(t,c) \omega^{kt}.
\end{equation}
Evaluating this for each $c \in \{0,1,2\}$ yields
\begin{equation}
    \phi_{x|z=c}(k) =
    \begin{cases}
        1, & k=0 \ (\text{for all } c), \\
        \frac{9p-1}{2(3p+1)}, & k=1,2 \ \text{and} \ c=0, \\
        0, & k=1,2 \ \text{and} \ c \in \{1,2\}.
    \end{cases}
\end{equation}

Next, we evaluate the characteristic function $\phi_{x_0}(k)$ of the shared state $x_0$. We define the probabilities $\widetilde{R}_i$ for $i \in \{0,1,2\}$ as
\begin{equation}
\label{def:R}
\widetilde{R}_i := P(x_0=i) = \sum_{j=0}^2 r_{ij}.
\end{equation}
The characteristic function is then directly expressed as
\begin{equation}
\label{eq:phi x0}
    \phi_{x_0}(k) = \sum_{i=0}^2 \widetilde{R}_i \omega^{ki}.
\end{equation}

Evaluating \eqref{eq:phi x0} at $k=0$ trivially yields the normalization condition $\phi_{x_0}(0) = \sum_{i=0}^2 \widetilde{R}_i = 1$. For the nontrivial modes $k \in \{1,2\}$, we algebraically decouple the sum into its symmetric and antisymmetric components. Expanding the characteristic functions yields
\begin{align}
    \phi_{x_0}(1) &= \widetilde{R}_0 + \frac{1}{2}(\omega+\omega^2)(\widetilde{R}_1+\widetilde{R}_2) + \frac{1}{2}(\omega-\omega^2)(\widetilde{R}_1-\widetilde{R}_2), \\
    \phi_{x_0}(2) &= \widetilde{R}_0 + \frac{1}{2}(\omega+\omega^2)(\widetilde{R}_1+\widetilde{R}_2) - \frac{1}{2}(\omega-\omega^2)(\widetilde{R}_1-\widetilde{R}_2).
\end{align}
This explicit decomposition clearly demonstrates that $\phi_{x_0}(1)$ and $\phi_{x_0}(2)$ are complex conjugates, where the term proportional to $(\omega-\omega^2)$ constitutes the pure imaginary component.

When evaluating the single-round success probability $P_{\mathrm{succ}}$, the pure imaginary components proportional to $(\omega-\omega^2)$ yield no net contribution. They either vanish upon multiplication (since $\phi_{x|z=c}(k) = 0$ for $c \neq 0$) or cancel perfectly in pairs due to the symmetry of the carrier characteristic function, $\phi_{x|z=0}(1) = \phi_{x|z=0}(2)$. By retaining only the effective real part and applying the root identity $\omega + \omega^2 = -1$ alongside the probability normalization $\sum_i \widetilde{R}_i = 1$, the characteristic function directly simplifies to
\begin{equation}
    \phi_{x_0}(0) = 1, \qquad \phi_{x_0}(1)+\phi_{x_0}(2) = 3\widetilde{R}_0-1.
\end{equation}

(v) Calculate $P_{\mathrm{succ}}$.

Substituting the previously derived characteristic functions into the expression for $P_{\mathrm{succ}}$, we evaluate the sum to obtain
\begin{equation}
    P_{\mathrm{succ}} = \frac{1}{3}\left(\frac{3p+1}{4}\right)^m + \frac{3\widetilde{R}_0-1}{3}\left(\frac{9p-1}{8}\right)^m + \frac{2}{3}\left(\frac{3(1-p)}{8}\right)^m.
\end{equation}

To render the subsequent analysis more readable, we introduce three auxiliary exponential decay parameters:
\begin{equation}
\label{def:ABC_m}
    A_m := \left(\frac{3p+1}{4}\right)^m, \quad 
    B_m := \left(\frac{9p-1}{8}\right)^m, \quad 
    C_m := \left(\frac{3(1-p)}{8}\right)^m.
\end{equation}

The single-round conditional success probability then takes the compact form
\begin{equation}
\label{def:Psucc}
    P_{\mathrm{succ}}=\dfrac{1}{3}[A_m+(3\widetilde{R}_0-1)B_m+2C_m].
\end{equation}

\vspace{1cm}

To determine the output Bell label $(s,t)$, we introduce indicator functions for the measurement outcomes. For the bit-flip label $s=i$, we define
\begin{equation}
    \delta(s=i) := \delta_0(i-x_0) = \frac{1}{3}\sum_{k=0}^2 \omega^{k(i-x_0)}.
\end{equation}

By the single-carrier phase-propagation rule in Lemma \ref{le:cnot1}, and using the successful syndrome condition that all carrier phase-error labels coincide, the output phase-flip label satisfies $t\equiv z_0-z_1 \pmod 3$. Accordingly, its indicator is defined as
\begin{equation}
    \delta(t=j) := \delta_0(j-z_0+z_1) = \frac{1}{3}\sum_{k=0}^2 \omega^{k(j-z_0+z_1)}.
\end{equation}

The joint Bell indicator factors as
\begin{equation}
\label{eq:q'st}
    \delta_{ij} := \mathbf{1}_{\mathrm{succ}} \delta(s=i) \delta(t=j),
\end{equation}
yielding the conditional probability $q'_{st} = \mathbb{E}[\delta_{st}] / P_{\mathrm{succ}}$.

Our primary objective is to evaluate $q'_{00}$ and establish that $q'_{00} \to 1$ in the asymptotic limit. Other required conditional probabilities $q'_{st}$ will be computed as necessary to facilitate this proof.

By substituting the explicit forms of the indicators into \eqref{eq:q'st}, the joint indicator for the target outcome $(s,t)=(0,0)$ expands as
\begin{equation}
    \delta_{00} = \underbrace{\left( \prod_{i=1}^{m-1} \delta_0(z_i-z_m)\right) \delta_0\left(\sum_{k=0}^m x_k\right)}_{\mathbf{1}_{\mathrm{succ}}} \underbrace{\delta_0(x_0)}_{s=0} \underbrace{\delta_0(z_0-z_1)}_{t=0}.
\end{equation}

The target conditions $s=0$ and $t=0$ deterministically enforce $x_0=0$ and $z_0=z_1$, respectively. Absorbing these constraints into the success indicator $\mathbf{1}_{\mathrm{succ}}$ simplifies the expression to
\begin{equation}
    \delta_{00} = \left( \prod_{i=0}^{m-1} \delta_0(z_i-z_m) \right) \delta_0\left(\sum_{k=1}^m x_k\right) \delta_0(x_0).
\end{equation}

Consequently, the joint indicator $\delta_{00}$ is non-vanishing if and only if all particles (index $0$ for the shared state and indices $1,\dots,m$ for the carriers) simultaneously satisfy the following constraints:
\begin{equation}
    \begin{cases}
        z_0 = z_1 = \cdots = z_m = \gamma, \quad \gamma \in \{0,1,2\}, \\
        x_0 = 0 \quad \text{and} \quad \sum_{k=1}^m x_k \equiv 0 \pmod 3.
    \end{cases}
\end{equation}

Next, we decompose the expectation value $\mathbb{E}[\delta_{00}]$ by conditioning on the mutually exclusive values of the common $Z$-error, $\gamma \in \{0,1,2\}$. Since the shared state and the carriers are independent after preprocessing, the joint probability factorizes. Letting $r_{ij} := P_0(i,j)$ denote the initial error probabilities of the shared state, the expectation expands as
\begin{align}
    \mathbb{E}[\delta_{00}] &= \sum_{\gamma=0}^2 P_0(0,\gamma) P\left( \sum_{k=1}^m x_k \equiv 0 \pmod 3 \;\middle|\; \forall k, z_k=\gamma \right) P(\forall k, z_k=\gamma) \notag \\
    &= \frac{1}{3} \sum_{\gamma=0}^2 r_{0\gamma} C_{\gamma}^m \sum_{k=0}^2 [\phi_{x|z=\gamma}(k)]^m \notag \\
    &= \frac{r_{00}}{3}(A_m + 2B_m) + \frac{r_{01} + r_{02}}{3}C_m.
\end{align}
Dividing by $P_{\mathrm{succ}}$ yields the normalized conditional probability,
\begin{equation}
    q'_{00} = \frac{1}{3P_{\mathrm{succ}}} \big[ r_{00}(A_m+2B_m) + (r_{01}+r_{02})C_m \big].
\end{equation}

For the outcomes $q'_{01}$ and $q'_{02}$, the target bit-flip label $s=0$ remains fixed, but the phase-flip target $t=j$ (where $j \in \{1,2\}$) modifies the constraint to $z_0 - \gamma \equiv j \pmod 3$. Physically, this simply induces a cyclic permutation of the shared state's $Z$-error indices $r_{0\gamma}$. Leaving the rest of the derivation structurally identical, we immediately obtain
\begin{align}
    q'_{01} &= \frac{1}{3P_{\mathrm{succ}}} \big[ r_{01}(A_m+2B_m) + (r_{02}+r_{00})C_m \big], \\
    q'_{02} &= \frac{1}{3P_{\mathrm{succ}}} \big[ r_{02}(A_m+2B_m) + (r_{00}+r_{01})C_m \big].
\end{align}

When evaluating the cases where $s \neq 0$, the $X$-error constraint shifts. For instance, fixing $s=1$ requires $x_0=1$, which modifies the carrier sum condition to $\sum_{k=1}^m x_k \equiv -1 \pmod 3$. The joint indicator $\delta_{10}$ becomes
\begin{equation}
    \delta_{10} = \left( \prod_{i=0}^{m-1} \delta_0(z_i-z_m) \right) \delta_0\left(\sum_{k=1}^m x_k + 1\right) \delta_0(x_0 - 1).
\end{equation}
Expanding the indicator for the carrier sum, $\delta_0\left(\sum_{k=1}^m x_k + 1\right)$, explicitly factors out an extra phase $\omega^k$ due to the constant offset $+1$. Since the characteristic function is symmetric, $\phi_{x|z=0}(1)=\phi_{x|z=0}(2)$, applying the identity $\omega+\omega^2=-1$ precisely flips the sign of the $B_m$ term in the expectation,
\begin{align}
    \mathbb{E}[\delta_{10}] &= \frac{1}{3} \sum_{\gamma=0}^2 r_{1\gamma} C_{\gamma}^m \sum_{k=0}^2 \omega^{k} [\phi_{x|z=\gamma}(k)]^m \notag \\
    &= \frac{r_{10}}{3}(A_m - B_m) + \frac{r_{11} + r_{12}}{3}C_m.
\end{align}
Normalizing yields $q'_{10}$. Applying the same cyclic permutation argument for $t \in \{1,2\}$ provides the complete set for $s=1$,
\begin{align}
    q'_{10} &= \frac{1}{3P_{\mathrm{succ}}} \big[ r_{10}(A_m-B_m) + (r_{11}+r_{12})C_m \big], \\
    q'_{11} &= \frac{1}{3P_{\mathrm{succ}}} \big[ r_{11}(A_m-B_m) + (r_{12}+r_{10})C_m \big], \\
    q'_{12} &= \frac{1}{3P_{\mathrm{succ}}} \big[ r_{12}(A_m-B_m) + (r_{10}+r_{11})C_m \big].
\end{align}

Finally, the analysis for $s=2$ proceeds in complete analogy. Fixing $s=2$ requires $x_0=2$, which shifts the carrier sum condition to $\sum_{k=1}^m x_k \equiv 1 \pmod 3$. Expanding the corresponding indicator yields a conjugate phase factor $\omega^{-k}$ (or equivalently $\omega^{2k}$). Due to the symmetry of the characteristic function, this phase factor again evaluates to $\omega^2 + \omega = -1$, yielding the identical $(A_m - B_m)$ dependence. Applying the cyclic permutation of the $Z$-error indices $r_{2\gamma}$ for the phase-flip target $t \in \{0,1,2\}$ provides the final set of conditional probabilities,
\begin{align}
    q'_{20} &= \frac{1}{3P_{\mathrm{succ}}} \big[ r_{20}(A_m-B_m) + (r_{21}+r_{22})C_m \big], \\
    q'_{21} &= \frac{1}{3P_{\mathrm{succ}}} \big[ r_{21}(A_m-B_m) + (r_{22}+r_{20})C_m \big], \\
    q'_{22} &= \frac{1}{3P_{\mathrm{succ}}} \big[ r_{22}(A_m-B_m) + (r_{20}+r_{21})C_m \big].
\end{align}

For convenience, we summarize the conditional probabilities $q'_{st}$ derived above. Recognizing the cyclic permutation in the phase-flip index $t$, the nine expressions can be compactly written as a generalized formula,
\begin{equation}
\label{eq:arrange}
    q'_{st} = \frac{1}{3P_{\mathrm{succ}}}
    \begin{cases}
        r_{st}(A_m+2B_m) + (r_{s, t\oplus 1} + r_{s, t\oplus 2})C_m, & s=0, \\
        r_{st}(A_m-B_m) + (r_{s, t\oplus 1} + r_{s, t\oplus 2})C_m, & s \in \{1,2\},
    \end{cases}
\end{equation}
where the addition in the second subscript is performed modulo 3 ($\oplus$), and the parameters $A_m, B_m, C_m$, and $P_{\mathrm{succ}}$ are given by \eqref{def:ABC_m} and \eqref{def:Psucc}. One can readily verify that these probabilities are properly normalized, $\sum_{s,t=0}^2 q'_{st} = 1$.

\vspace{1cm}

To investigate whether the target state probability $q'_{00}$ asymptotically approaches unity as $m$ increases, we treat the transformation as a discrete dynamical map. By substituting the shared state's initial error probabilities $r_{ij}$ with the iterative parameters $q_{ij}$ according to the permutation rules in \eqref{eq:r=q}, the nine recursion relations can be compactly grouped by the output bit-flip index $s$ and expressed in matrix form. Defining the mapped input state vectors as 
\begin{equation}
\label{eq:mapped_input_vectors}
\begin{aligned}
    \mathbf v^{(0)}
        &=(q_{00},q_{21},q_{12})^T,\\
    \mathbf v^{(1)}
        &=(q_{10},q_{01},q_{22})^T,\\
    \mathbf v^{(2)}
        &=(q_{20},q_{11},q_{02})^T.
\end{aligned}
\end{equation}
the updated probability vectors $\mathbf{q}'^{(s)} = (q'_{s0}, q'_{s1}, q'_{s2})^T$ are given by
\begin{equation}
\label{eq:matrix_map}
    \mathbf{q}'^{(s)} = \frac{1}{3P_{\mathrm{succ}}}
    \begin{pmatrix}
        \alpha_s & C_m & C_m \\
        C_m & \alpha_s & C_m \\
        C_m & C_m & \alpha_s
    \end{pmatrix}
    \mathbf{v}^{(s)},
\end{equation}
where the diagonal parameter $\alpha_s$ is determined by the target subspace, defined as
\begin{equation}
    \alpha_s =
    \begin{cases}
        A_m + 2B_m, & s=0, \\
        A_m - B_m, & s \in \{1,2\}.
    \end{cases}
\end{equation}

We now prove that the conditional recurrence actually converges for every fixed finite $m$. Define
\begin{equation}
    D_0:=
    \begin{pmatrix}
        A_m+2B_m & C_m & C_m\\
        C_m & A_m+2B_m & C_m\\
        C_m & C_m & A_m+2B_m
    \end{pmatrix},
    \qquad
    D_1:=
    \begin{pmatrix}
        A_m-B_m & C_m & C_m\\
        C_m & A_m-B_m & C_m\\
        C_m & C_m & A_m-B_m
    \end{pmatrix}.
\end{equation}
Let $\Pi_S$ denote the permutation matrix induced by the
symplectic shear in \eqref{eq:r=q}. Thus,
\begin{equation}
\label{eq:PiS_definition}
    \Pi_S\mathbf q=\mathbf r
    =
    (q_{00},q_{21},q_{12},
     q_{10},q_{01},q_{22},
     q_{20},q_{11},q_{02})^T.
\end{equation}
Equivalently,
\begin{equation}
    \Pi_S\mathbf e_{s,t}
    =\mathbf e_{s\oplus t,t}.
\end{equation}
Define
\begin{equation}
\label{eq:Mm_definition}
    M_m:=\operatorname{diag}(D_0,D_1,D_1)\Pi_S.
\end{equation}
The transition relations in \eqref{eq:matrix_map} can then be written as the normalized linear recurrence
\begin{equation}
\label{eq:normalized_recurrence}
    \mathbf{q}^{(n+1,m)} = \frac{M_m\mathbf{q}^{(n,m)}} {\mathbf{1}^{T}M_m\mathbf{q}^{(n,m)}},
\end{equation}
where $\mathbf{1}^{T}=(1,1,\dots,1)$ denotes the $9$‑dimensional row vector with all entries equal to one. Indeed, using \eqref{def:Psucc}, one verifies that
\begin{equation}
    \mathbf{1}^{T}M_m\mathbf{q}^{(n,m)} = 3P_{\mathrm{succ}}\!\left(\mathbf{q}^{(n,m)}\right).
\end{equation}
Iterating \eqref{eq:normalized_recurrence} gives
\begin{equation}
\label{eq:normalized_matrix_power}
    \mathbf{q}^{(n,m)} = \frac{M_m^n\mathbf{q}^{(0,m)}} {\mathbf{1}^{T}M_m^n\mathbf{q}^{(0,m)}}.
\end{equation}

For the cumulative success probability defined in \eqref{eq:total_success_probability}, the identity $\mathbf{1}^{T}M_m\mathbf{q}^{(n,m)}=3P_{\mathrm{succ}}(\mathbf{q}^{(n,m)})$ and \eqref{eq:normalized_matrix_power} give
\begin{equation}
    P_{\mathrm{succ}}\!\left(\mathbf{q}^{(n,m)}\right)=\frac{\mathbf{1}^{T}M_m^{n+1}\mathbf{q}^{(0,m)}}{3\mathbf{1}^{T}M_m^n\mathbf{q}^{(0,m)}}.
\end{equation}
Consequently, the product in \eqref{eq:total_success_probability} telescopes to
\begin{equation}
\label{eq:total_success_closed_form}
    P_{\mathrm{tot}}^{(N,m)}=\frac{\mathbf{1}^{T}M_m^N\mathbf{q}^{(0,m)}}{3^N}.
\end{equation}

We next prove convergence for every fixed finite $m$. Assume first that $1/3<p<1$. From the definitions of $A_m,B_m,C_m$, we have
\begin{equation}
    A_m>B_m>C_m>0.
\end{equation}
Therefore, every entry of both $D_0$ and $D_1$ is strictly positive.

We claim that
\begin{equation}
\label{eq:Mm_primitive}
    M_m^2>0
\end{equation}
entrywise. Indeed, for every Bell label $(x,z)$,
\begin{equation}
    \Pi_S\mathbf e_{x,z} =\mathbf e_{x\oplus z,z}.
\end{equation}
Since the corresponding block $D_0$ or $D_1$ is strictly positive, one application of $M_m$ produces positive components along every vector
\begin{equation}
    \mathbf e_{x\oplus z,u},
    \qquad 
    u\in\mathbb Z_3.
\end{equation}
After the second preprocessing step, these labels become
\begin{equation}
    (x\oplus z\oplus u,u).
\end{equation}
As $u$ ranges over $\mathbb Z_3$, the first index ranges over all three values. The second check then produces every possible second index. Thus every Bell label can reach every other Bell label in two iterations, proving \eqref{eq:Mm_primitive}.

Hence $M_m$ is a primitive non-negative matrix. Let $\rho_m>0$ be its Perron eigenvalue, and let $\mathbf v_m>0$ and $\mathbf w_m>0$ be right and left Perron eigenvectors, respectively, normalized by
\begin{equation}
    \mathbf 1^T\mathbf v_m=1.
\end{equation}
The Perron--Frobenius theorem gives
\begin{equation}
\label{eq:perron_limit}
    \rho_m^{-n}M_m^n
    \longrightarrow
    \frac{\mathbf v_m\mathbf w_m^T}
         {\mathbf w_m^T\mathbf v_m}.
\end{equation}
For $1/3<p<1$, the initial depolarizing distribution is strictly positive, and hence
\begin{equation}
    \mathbf w_m^T\mathbf q^{(0,m)}>0.
\end{equation}
Combining \eqref{eq:perron_limit} with \eqref{eq:normalized_matrix_power} yields
\begin{equation}
\label{eq:fixed_m_convergence}
    \lim_{n\to\infty}\mathbf q^{(n,m)} =\mathbf q_\ast^{(m)} =\mathbf v_m.
\end{equation}

The same Perron--Frobenius limit and
\eqref{eq:total_success_closed_form} give
\begin{equation}
    P_{\mathrm{tot}}^{(N,m)}
    \sim
    \frac{\mathbf w_m^T\mathbf q^{(0,m)}}
         {\mathbf w_m^T\mathbf v_m}
    \left(\frac{\rho_m}{3}\right)^N.
\end{equation}
Right multiplication by $\Pi_S$ does not change the row sums. Therefore,
\begin{align}
    \rho_m
    &\leq A_m+2B_m+2C_m\\
    &\leq A_1+2B_1+2C_1
      =\frac{9p+3}{4}<3.
\end{align}
It follows that
\begin{equation}
    \lim_{n\to\infty} P_{\mathrm{succ}}\!\left(\mathbf q^{(n,m)}\right) =\frac{\rho_m}{3}<1,
    \qquad
    \lim_{N\to\infty}P_{\mathrm{tot}}^{(N,m)}=0.
\end{equation}

The endpoint $p=1$ is immediate. In this case $A_m=B_m=1$, $C_m=0$, and the initial state is $\mathbf e_{00}$. Since $\Pi_S\mathbf e_{00}=\mathbf e_{00}$,
we have
\begin{equation}
    \mathbf q^{(n,m)}=\mathbf e_{00},
    \qquad
    P_{\mathrm{tot}}^{(N,m)}=1
\end{equation}
for all $n,N$, and $m$.

For the remainder of the proof, assume $1/3<p<1$. Define
\begin{equation}
    \eta_m:=\frac{B_m}{A_m},
    \qquad
    \varepsilon_m:=\frac{C_m}{A_m}.
\end{equation}
Explicitly,
\begin{equation}
    \eta_m= \left( \frac{9p-1}{2(3p+1)} \right)^m,
    \qquad
    \varepsilon_m= \left( \frac{3(1-p)}{2(3p+1)} \right)^m.
\end{equation}
For $1/3<p<1$,
\begin{equation}
    \eta_m\longrightarrow0,
    \qquad
    \frac{\varepsilon_m}{\eta_m} =\frac{C_m}{B_m} = \left( \frac{3(1-p)}{9p-1} \right)^m \longrightarrow0.
\end{equation}
Hence $\varepsilon_m=o(\eta_m)$.

Let
\begin{equation}
    K:=\operatorname{diag} (2,2,2,-1,-1,-1,-1,-1,-1)
\end{equation}
in the standard Bell-label ordering, and let
\begin{equation}
    L:=\operatorname{diag} (J_3-I_3,J_3-I_3,J_3-I_3),
\end{equation}
where $J_3$ is the $3\times3$ all-ones matrix. Then
\begin{equation}
\label{eq:normalized_M_expansion}
    \frac{M_m}{A_m} = \left(I_9+\eta_mK+\varepsilon_mL\right)\Pi_S.
\end{equation}

The shear has order three:
\begin{equation}
    \Pi_S^3=I_9.
\end{equation}
Cubing \eqref{eq:normalized_M_expansion} and using $\varepsilon_m=o(\eta_m)$ gives
\begin{equation}
\label{eq:three_round_expansion}
    \left(\frac{M_m}{A_m}\right)^3 = I_9+\eta_mG+o(\eta_m),
\end{equation}
where
\begin{equation}
    G= K+\Pi_SK\Pi_S^{-1} +\Pi_S^2K\Pi_S^{-2}.
\end{equation}
A direct evaluation in the standard ordering $(00,01,02,10,11,12,20,21,22)$ yields
\begin{equation}
\label{eq:G_explicit}
    G= \operatorname{diag} (6,0,0,-3,0,0,-3,0,0).
\end{equation}
Thus the eigenvalues of $G$ are
\begin{equation}
    6\quad\text{(multiplicity one)},\qquad
    0\quad\text{(multiplicity six)},\qquad
    -3\quad\text{(multiplicity two)}.
\end{equation}
The unique eigenspace corresponding to the largest eigenvalue $6$ is
\begin{equation}
    \ker(G-6I_9)=\operatorname{span}\{\mathbf e_{00}\}.
\end{equation}

Define the effective generator
\begin{equation}
\label{eq:effective_generator}
    \mathcal K_m :=
    \frac{
        \left(M_m/A_m\right)^3-I_9
    }{\eta_m}.
\end{equation}
Equation~\eqref{eq:three_round_expansion} implies that
\begin{equation}
\label{eq:effective_generator_limit}
    \|\mathcal K_m-G\|
    \longrightarrow0
\end{equation}
in any fixed matrix norm. The matrices $\mathcal K_m$ and $\left(M_m/A_m\right)^3$ have the same left and right eigenvectors. More precisely, if $\mathcal K_m\mathbf v=\mu\mathbf v$, then
\begin{equation}
\label{eq:effective_generator_eigenvalue_relation}
    \left(\frac{M_m}{A_m}\right)^3\mathbf v = \left(1+\eta_m\mu\right)\mathbf v.
\end{equation}

Since the eigenvalue $6$ of $G$ is simple and isolated, finite-dimensional spectral perturbation theory implies that, for all sufficiently large $m$, $\mathcal K_m$ has a unique eigenvalue $\mu_\ast^{(m)}$ in a fixed neighborhood of $6$, and
\begin{equation}
\label{eq:effective_dominant_eigenvalue}
    \mu_\ast^{(m)}=6+o(1).
\end{equation}
The remaining eigenvalues converge, counting algebraic multiplicity, to either $0$ or $-3$. Because $\mathcal K_m$ is real and the neighborhood of $6$ contains exactly one eigenvalue, $\mu_\ast^{(m)}$ is real.

It follows from \eqref{eq:effective_generator_eigenvalue_relation} that the corresponding eigenvalue of $\left(M_m/A_m\right)^3$ is
\begin{equation}
\label{eq:cubed_dominant_eigenvalue}
    \lambda_\ast^{(m)} = 1+6\eta_m+o(\eta_m),
\end{equation}
whereas the remaining eigenvalues have the forms
\begin{equation}
\label{eq:cubed_remaining_eigenvalues}
    1+o(\eta_m)
    \qquad
    \text{or}
    \qquad
    1-3\eta_m+o(\eta_m).
\end{equation}
Since $\eta_m>0$, for all sufficiently large $m$, $\lambda_\ast^{(m)}$ is the unique eigenvalue of maximum modulus.

The entrywise positivity $M_m^2>0$ established in \eqref{eq:Mm_primitive} implies that $M_m$ is primitive and that $M_m^3>0$. Therefore, $\lambda_\ast^{(m)}$ is the Perron eigenvalue of $\left(M_m/A_m\right)^3$. The associated normalized Perron eigenvector converges to the eigenvector of $G$ corresponding to its isolated eigenvalue $6$. By \eqref{eq:G_explicit}, this eigenspace is $\operatorname{span}\{\mathbf e_{00}\}$. Hence,
\begin{equation}
\label{eq:fixed_point_vector_limit}
    \mathbf q_\ast^{(m)} \longrightarrow \mathbf e_{00}
    \qquad
    (m\to\infty).
\end{equation}
Since $M_m$ and $M_m^3$ have the same normalized Perron eigenvector, $\mathbf q_\ast^{(m)}$ is also the attracting fixed point of the original one-round recurrence. In particular,
\begin{equation}
    q_{\ast,00}^{(m)} \longrightarrow1.
\end{equation}
Combining this result with the fixed-$m$ convergence in \eqref{eq:fixed_m_convergence}, we obtain
\begin{equation}
\label{eq:ordered_limit_final}
    \lim_{m\to\infty} \left( \lim_{n\to\infty}q_{00}^{(n,m)} \right)=1,
\end{equation}
which proves Theorem~\ref{thm:mcaepp_convergence}.

For every $F_{\mathrm{tar}}<q_{\ast,00}^{(m)}$, the fixed-$m$ convergence guarantees that $N_{\mathrm{tar}}(m)$ is finite. Moreover, since $\lim_{m\to\infty}q_{\ast,00}^{(m)}=1$, every prescribed target fidelity $F_{\mathrm{tar}}<1$ can be reached with finite $m$ and finite $N_{\mathrm{tar}}(m)$, although the corresponding cumulative success probability may be small.

The order of limits in \eqref{eq:ordered_limit_final} is essential. Indeed,
\begin{equation}
    \frac{B_m}{A_m} \to 0,
    \qquad
    \frac{C_m}{A_m} \to 0,
\end{equation}
and hence $M_m/A_m\to\Pi_S$. Writing the $m$-independent initial depolarizing distribution as $\mathbf q^{(0)}$, we have $\Pi_S\mathbf q^{(0)}=\mathbf q^{(0)}$. Therefore, for every fixed $n$,
\begin{equation}
    \lim_{m\to\infty}\mathbf q^{(n,m)} = \mathbf q^{(0)}.
\end{equation}
Consequently,
\begin{equation}
    \lim_{n\to\infty} \left( \lim_{m\to\infty}q_{00}^{(n,m)} \right) = p_{00},
\end{equation}
so the two limits cannot be interchanged.

\section{Proof of Theorem \ref{thm:universal_purification}}
\label{app:ada}

All Bell and Pauli indices in this appendix are evaluated modulo $3$. We first derive the exact conditional map generated by one $m$-carrier check and then analyze the complete two-basis cycle. Throughout this appendix, $M_Z^{(m)}$ denotes the unnormalized transition matrix of the raw star-stabilizer check $\mathcal C_Z^{(m)}$. In particular, it does not include the bilateral shear $W_A^*\otimes W_B$ or the corresponding permutation matrix $\Pi_S$ used in Theorem~\ref{thm:mcaepp_convergence}.

Let the initial Pauli probabilities be $\{p_{xz}\}$. The operator $J$ introduced in Section \ref{subsec:ada} satisfies
\begin{equation}
    JX^xZ^zJ^\dagger=X^{-x}Z^{-z}.
\end{equation}
Consequently, the coordinated random choice between $I$ and $J$ maps the channel probabilities to
\begin{equation}
    \widetilde p_{xz} = \frac{p_{xz}+p_{-x,-z}}{2}.
\end{equation}
The same random operation on the two subsystems of the shared Bell pair produces the corresponding Bell-diagonal state. The target probability remains unchanged,
\begin{equation}
    \widetilde p_{00}=p_{00}.
\end{equation}
The random choices used for distinct carrier transmissions are independent. Hence, after averaging over these choices, the carrier errors remain independent and identically distributed according to $\{\widetilde p_{xz}\}$.

The four MUB line weights are unchanged by the pair symmetrization. By Lemma \ref{lemma:mub_dominance}, at least one line has weight strictly larger than $1/2$. After mapping a maximum-weight line to $z=0$, we relabel the resulting probabilities as $\{p_{xz}\}$. They have the form
\begin{equation}
\label{eq:appB_aligned_matrix}
    P=
    \begin{pmatrix}
        a&c&c\\
        b&d&e\\
        b&e&d
    \end{pmatrix},
\end{equation}
where
\begin{equation}
\label{eq:appB_normalization}
    a=p_{00},
    \qquad
    a+2(b+c+d+e)=1,
    \qquad
    b=\max\{b,c,d,e\}.
\end{equation}
The dominant MUB line is the column $z=0$, whose total weight is
\begin{equation}
\label{eq:appB_alpha}
    \alpha:=a+2b=L_{\max}>\frac{1}{2}.
\end{equation}
We further define
\begin{equation}
\label{eq:appB_beta_gamma}
    \beta:=a-b,
    \qquad
    \gamma:=c+d+e=\dfrac{1-a-2b}{2}=\dfrac{1-\alpha}{2}.
\end{equation}
The fidelity assumption $a>1/3$ gives
\begin{equation}
\label{eq:appB_spectral_inequality}
    \beta-\gamma = a-b-\frac{1-a-2b}{2} = \frac{3a-1}{2} > 0.
\end{equation}
If $b=0$, then the maximality of $b$ in \eqref{eq:appB_normalization} implies $c=d=e=0$ and $a=1$, in which case the channel is noiseless and the theorem is immediate. We therefore assume $b>0$ below. Equations \eqref{eq:appB_alpha}-\eqref{eq:appB_spectral_inequality} then imply
\begin{equation}
\label{eq:appB_ordering}
    0\leq\gamma<\beta<\alpha.
\end{equation}

Let
\begin{equation}
    \mathbf{q} = (q_{00},q_{01},q_{02},q_{10},q_{11},q_{12},q_{20},q_{21},q_{22})^T
\end{equation}
be the Bell probability vector of the shared state immediately before one $m$-carrier check. The $m$ carrier error pairs are denoted by $(x_j,z_j)$, $j=1,\ldots,m$, and are independently distributed according to \eqref{eq:appB_aligned_matrix}. If the shared-state input label is $(s,z_0)$, the stabilizer conditions in \eqref{condi:succ1} require
\begin{equation}
    z_1=z_2=\cdots=z_m
\end{equation}
and
\begin{equation}
s+\sum_{j=1}^{m}x_j\equiv0\pmod 3.
\end{equation}
Conditioned on success, the output Bell label is
\begin{equation}
x_{\mathrm{out}}=s,
\qquad
z_{\mathrm{out}}=z_0-z_1.
\end{equation}

For $s,l\in\mathbb{Z}_3$, define
\begin{equation}
\label{eq:appB_kappa_definition}
    \kappa_s^{(m)}(l) := \sum_{\substack{x_1,\ldots,x_m\in\mathbb{Z}_3\\ x_1+\cdots+x_m\equiv-s\pmod 3}} \prod_{j=1}^{m}p_{x_j l}.
\end{equation}
This is the probability that all carriers have phase label $l$ and that their shift labels satisfy the required modulo-$3$ sum. The exact unnormalized output distribution of one successful check is therefore
\begin{equation}
\label{eq:appB_exact_single_block}
    \bigl(M_Z^{(m)}\mathbf{q}\bigr)_{st} = \sum_{l=0}^{2} \kappa_s^{(m)}(l)q_{s,t+l}.
\end{equation}
The conditional output is obtained by dividing \eqref{eq:appB_exact_single_block} by
\begin{equation}
    P_Z^{(m)}(\mathbf{q}) = \mathbf{1}^TM_Z^{(m)}\mathbf{q}.
\end{equation}
Equation \eqref{eq:appB_exact_single_block} also shows that a finite-$m$ check does not leave the phase label invariant: carrier phase errors produce the shifts $t\mapsto t-l$. This is why the proof must analyze a complete alternating-basis cycle rather than two separate infinite sequences of checks.

Using the modulo-$3$ character identity, \eqref{eq:appB_kappa_definition} can equivalently be written as
\begin{equation}
\label{eq:appB_kappa_character}
    \kappa_s^{(m)}(l) = \frac{1}{3} \sum_{k=0}^{2} \omega^{ks} \left( \sum_{x=0}^{2}p_{xl}\omega^{kx} \right)^m.
\end{equation}
For the dominant column $l=0$, the distribution is $(a,b,b)^T$. Hence,
\begin{equation}
    \sum_{x=0}^{2}p_{x0}=\alpha
\end{equation}
and
\begin{equation}
    \sum_{x=0}^{2}p_{x0}\omega^x = \sum_{x=0}^{2}p_{x0}\omega^{2x} = a-b = \beta.
\end{equation}
Substitution into \eqref{eq:appB_kappa_character} yields
\begin{equation}
\label{eq:appB_kappa_zero}
    \kappa_0^{(m)}(0) = \frac{\alpha^m+2\beta^m}{3}
\end{equation}
and
\begin{equation}
\label{eq:appB_kappa_nonzero}
    \kappa_1^{(m)}(0) = \kappa_2^{(m)}(0) = \frac{\alpha^m-\beta^m}{3}.
\end{equation}

Each of the remaining columns $l=1,2$ has total probability $\gamma$. Therefore,
\begin{equation}
\label{eq:appB_kappa_bound}
    0\leq\kappa_s^{(m)}(1)\leq\gamma^m,
    \qquad
    0\leq\kappa_s^{(m)}(2)\leq\gamma^m
\end{equation}
for every $s\in\mathbb{Z}_3$.

Let $\mathbf{e}_{st}$ denote the standard basis vector associated with the Bell label $(s,t)$, and define the diagonal matrix $D_Z$ by
\begin{equation}
\label{eq:appB_DZ}
    D_Z\mathbf{e}_{st} =
    \begin{cases}
        2\mathbf{e}_{st},&s=0,\\
        -\mathbf{e}_{st},&s\in\{1,2\}.
    \end{cases}
\end{equation}

Define 
\begin{equation}
\label{eq:appB_rm}
    r_m := \left(\frac{\beta}{\alpha}\right)^m.
\end{equation}
After removing the common positive factor $\alpha^m/3$ from the single-block map, the exact relation \eqref{eq:appB_exact_single_block} becomes
\begin{align}
    \frac{3}{\alpha^m}\bigl(M_Z^{(m)}\mathbf{q}\bigr)_{st}
    &= \frac{3\kappa_s^{(m)}(0)}{\alpha^m}q_{st} + \frac{3}{\alpha^m} \left[ \kappa_s^{(m)}(1)q_{s,t+1} + \kappa_s^{(m)}(2)q_{s,t+2} \right]\notag\\
    \label{eq:appB_scaled_single_block_component}
    &= q_{st} + r_m\bigl(D_Z\mathbf{q}\bigr)_{st} + \frac{3}{\alpha^m} \left[ \kappa_s^{(m)}(1)q_{s,t+1} + \kappa_s^{(m)}(2)q_{s,t+2} \right],
\end{align}
where the second equality follows from \eqref{eq:appB_kappa_zero} and \eqref{eq:appB_kappa_nonzero}. Define the matrix $E_m$ by
\begin{equation}
    \bigl(E_m\mathbf{q}\bigr)_{st} := \frac{3}{\alpha^m} \left[ \kappa_s^{(m)}(1)q_{s,t+1} + \kappa_s^{(m)}(2)q_{s,t+2} \right].
\end{equation}
Equation \eqref{eq:appB_scaled_single_block_component} can then be written compactly as
\begin{equation}
\label{eq:appB_scaled_single_block}
    \widehat{M}_Z^{(m)} := \frac{3}{\alpha^m}M_Z^{(m)} = I+r_mD_Z+E_m.
\end{equation}
Since the matrix dimension is fixed and \eqref{eq:appB_kappa_bound} gives $\kappa_s^{(m)}(1),\kappa_s^{(m)}(2)\leq\gamma^m$, for any fixed matrix norm,
\begin{equation}
\label{eq:appB_Em_bound}
    \lVert E_m\rVert = O\left( \left(\frac{\gamma}{\alpha}\right)^m \right).
\end{equation}
Because $\gamma<\beta$, equations \eqref{eq:appB_rm} and \eqref{eq:appB_Em_bound} imply
\begin{equation}
\label{eq:appB_Em_small}
    \lVert E_m\rVert=o(r_m).
\end{equation}

Let $\Pi_R$ be the permutation matrix induced by the bilateral Fourier operation
\begin{equation}
    R_{AB} = H_{3,A}\otimes H_{3,B}^*.
\end{equation}
For the Bell-state convention
\begin{equation}
    \ket{\Phi^{s,t}} = \frac{1}{\sqrt{3}} \sum_{j=0}^{2} \omega^{tj}\ket{j,j+s},
\end{equation}
the bilateral Fourier operation gives
\begin{equation}
    R_{AB}\ket{\Phi^{s,t}} = \omega^{-st}\ket{\Phi^{t,-s}}.
\end{equation}
Since the phase factor $\omega^{-st}$ does not affect the Bell-state probability distribution, the corresponding permutation matrix satisfies
\begin{equation}
\label{eq:appB_fourier_permutation}
    \Pi_R\mathbf{e}_{st} = \mathbf{e}_{t,-s}.
\end{equation}
After the second check, the inverse bilateral Fourier operation
\begin{equation}
    R_{AB}^{-1} = H_{3,A}^\dagger\otimes H_{3,B}^T
\end{equation}
returns the state to the original Bell coordinate system. Therefore, the unnormalized map associated with one complete successful cycle is
\begin{equation}
\label{eq:appB_cycle_matrix}
    M_{\mathrm{cyc}}^{(m)} = \Pi_R^{-1}M_Z^{(m)}\Pi_RM_Z^{(m)}.
\end{equation}
The second occurrence of $M_Z^{(m)}$ corresponds to a fresh set of $m$ carriers transmitted independently through the same memoryless Pauli channel.

Define
\begin{equation}
\label{eq:appB_DX}
    D_X := \Pi_R^{-1}D_Z\Pi_R.
\end{equation}
This similarity transformation expresses the leading-order contribution of the second $Z$-basis check in the original Bell coordinate system. More explicitly, since \eqref{eq:appB_DZ} and $\Pi_R\mathbf{e}_{st}=\mathbf{e}_{t,-s}$, we have
\begin{align}
    D_X\mathbf{e}_{st}
    &= \Pi_R^{-1}D_Z\Pi_R\mathbf{e}_{st}\notag\\
    \label{eq:appB_DX_action}
    &= \Pi_R^{-1}D_Z\mathbf{e}_{t,-s}=
    \begin{cases}
        2\mathbf{e}_{st}, & t=0,\\
        -\mathbf{e}_{st}, & t\in\{1,2\}.
    \end{cases}
\end{align}
Thus, $D_Z$ distinguishes the sector $s=0$ from the sectors $s\neq0$, whereas $D_X$ distinguishes the sector $t=0$ from the sectors $t\neq0$. The notation $D_X$ therefore refers to the leading-order action of the Fourier-rotated check; it does not represent an additional physical operation.

We also define the matrix $\widetilde{E}_m$ by
\begin{equation}
\label{eq:appB_Etilde_definition}
    \widetilde{E}_m := \Pi_R^{-1}E_m\Pi_R.
\end{equation}
The matrix $E_m$ contains the subleading contributions to the first check in the original Bell coordinate system. Accordingly, $\widetilde E_m$ contains the same subleading contributions to the second check after the Fourier rotation and the subsequent return to the original Bell coordinate system. Since $\Pi_R$ is a fixed $9\times9$ permutation matrix independent of $m$, for any fixed matrix norm there exists a constant $C_R>0$, independent of $m$, such that
\begin{equation}
\label{eq:appB_Etilde_bound}
    \lVert\widetilde E_m\rVert = \left\lVert\Pi_R^{-1}E_m\Pi_R\right\rVert \leq C_R\lVert E_m\rVert.
\end{equation}
For the standard induced $1$-norm, induced infinity norm, spectral norm, and Frobenius norm, one may take $C_R=1$ because multiplication by a permutation matrix only permutes rows or columns. In particular, \eqref{eq:appB_Em_small} implies
\begin{equation}
\label{eq:appB_Etilde_small}
    \lVert\widetilde E_m\rVert = o(r_m).
\end{equation}

The scaled matrix of the Fourier-rotated second check is therefore
\begin{align}
    \Pi_R^{-1}\widehat M_Z^{(m)}\Pi_R
    &= \Pi_R^{-1}\left(I+r_mD_Z+E_m\right)\Pi_R\notag\\
    \label{eq:appB_rotated_single_block}
    &= I+r_mD_X+\widetilde E_m.
\end{align}
Because one complete cycle consists of the first check, followed by the Fourier-rotated second check, its scaled unnormalized matrix is
\begin{align}
    \widehat M_{\mathrm{cyc}}^{(m)}
    &:= \left(\frac{3}{\alpha^m}\right)^2M_{\mathrm{cyc}}^{(m)}\notag\\
    &= \Pi_R^{-1}\widehat M_Z^{(m)}\Pi_R\widehat M_Z^{(m)}\notag\\
    \label{eq:appB_scaled_cycle_product}
    &= \left(I+r_mD_X+\widetilde E_m\right) \left(I+r_mD_Z+E_m\right).
\end{align}
Expanding this product without assuming that the matrices commute gives
\begin{align}
    \widehat M_{\mathrm{cyc}}^{(m)} ={}
    & I+r_m(D_Z+D_X)\notag\\
    \label{eq:appB_scaled_cycle_expansion}
    &+E_m+\widetilde E_m +r_m^2D_XD_Z +r_mD_XE_m +r_m\widetilde E_mD_Z +\widetilde E_mE_m.
\end{align}
The order of the factors in the last four terms follows directly from the order of the two matrices in \eqref{eq:appB_scaled_cycle_product} and must be retained. In particular, no commutativity between $D_X$, $D_Z$, $E_m$, and $\widetilde E_m$ is required.

Collecting all terms that are asymptotically smaller than the leading correction $r_m(D_Z+D_X)$, define
\begin{equation}
\label{eq:appB_Fm_definition}
    F_m := E_m+\widetilde E_m +r_m^2D_XD_Z +r_mD_XE_m +r_m\widetilde E_mD_Z +\widetilde E_mE_m.
\end{equation}
Equation \eqref{eq:appB_scaled_cycle_expansion} can then be written in the compact form
\begin{equation}
\label{eq:appB_scaled_cycle}
    \widehat M_{\mathrm{cyc}}^{(m)} = I+r_m(D_Z+D_X)+F_m.
\end{equation}

It remains to verify that $F_m$ is indeed negligible compared with $r_m$. Using a fixed submultiplicative matrix norm, \eqref{eq:appB_Fm_definition} gives
\begin{align}
    \lVert F_m\rVert \leq{}
    & \lVert E_m\rVert +\lVert\widetilde E_m\rVert +r_m^2\lVert D_XD_Z\rVert\notag\\
    &+r_m\lVert D_X\rVert\lVert E_m\rVert +r_m\lVert\widetilde E_m\rVert\lVert D_Z\rVert +\lVert\widetilde E_m\rVert\lVert E_m\rVert.
\label{eq:appB_Fm_bound}
\end{align}
The matrices $D_X$ and $D_Z$ are fixed and independent of $m$, so their norms are finite constants. Moreover,
\begin{equation}
    \lVert E_m\rVert=o(r_m),
    \qquad
    \lVert\widetilde E_m\rVert=o(r_m).
\end{equation}
For a nontrivial noisy channel, $0<\beta/\alpha<1$, and hence
\begin{equation}
    r_m = \left(\frac{\beta}{\alpha}\right)^m \longrightarrow0.
\end{equation}
It follows that
\begin{equation}
    r_m^2=o(r_m),
\end{equation}
and consequently
\begin{align}
    r_m^2\lVert D_XD_Z\rVert&=o(r_m),\\
    r_m\lVert D_X\rVert\lVert E_m\rVert&=o(r_m),\\
    r_m\lVert\widetilde E_m\rVert\lVert D_Z\rVert&=o(r_m),\\
    \lVert\widetilde E_m\rVert\lVert E_m\rVert&=o(r_m).
\end{align}
Together with $\lVert E_m\rVert=o(r_m)$ and $\lVert\widetilde E_m\rVert=o(r_m)$, these estimates imply
\begin{equation}
\label{eq:appB_Fm_small}
    \lVert F_m\rVert = o(r_m).
\end{equation}
Therefore, the first nontrivial asymptotic correction to the identity matrix is completely determined by $D_Z+D_X$, whereas all remaining contributions are asymptotically smaller than $r_m$.

Set
\begin{equation}
    G:=D_Z+D_X.
\end{equation}
Equations \eqref{eq:appB_DZ}, \eqref{eq:appB_fourier_permutation}, and \eqref{eq:appB_DX} show that $G$ is diagonal in the Bell basis and satisfies
\begin{equation}
\label{eq:appB_G_spectrum}
    G\mathbf{e}_{st} =
    \begin{cases}
        4\mathbf{e}_{st},&(s,t)=(0,0),\\
        \mathbf{e}_{st},&s=0,\ t\neq0,\\
        \mathbf{e}_{st},&s\neq0,\ t=0,\\
        -2\mathbf{e}_{st},&s\neq0,\ t\neq0.
    \end{cases}
\end{equation}
Hence, $\mathbf{e}_{00}$ is the unique eigenvector of $G$ associated with its largest eigenvalue $4$, and the gap to the next eigenvalue is $3$.

To analyze the dominant eigendirection, define
\begin{equation}
\label{eq:appB_effective_generator}
    K_m := \frac{\widehat M_{\mathrm{cyc}}^{(m)}-I}{r_m} = G+\frac{F_m}{r_m},
    \qquad
    G:=D_Z+D_X.
\end{equation}
Since $\lVert F_m\rVert=o(r_m)$, we have
\begin{equation}
\label{eq:appB_generator_limit}
    \lim_{m\to\infty} \left\lVert K_m-G\right\rVert = 0.
\end{equation}
The matrices $K_m$ and $\widehat M_{\mathrm{cyc}}^{(m)}$ have the same left and right eigenvectors. Indeed, if $K_m\mathbf{v}=\mu\mathbf{v}$, then
\begin{equation}
    \widehat M_{\mathrm{cyc}}^{(m)}\mathbf{v} = \left(1+r_m\mu\right)\mathbf{v}.
\end{equation}

By \eqref{eq:appB_G_spectrum}, the eigenvalues of $G$ are $4$, $1$, and $-2$, with multiplicities $1$, $4$, and $4$, respectively, and the eigenspace associated with the simple eigenvalue $4$ is spanned by $\mathbf{e}_{00}$.

The eigenvalue $4$ of $G$ is simple and is separated from the remaining spectrum by the gap
\begin{equation}
    \min\{|4-1|,|4-(-2)|\}=3.
\end{equation}
Since $\lVert K_m-G\rVert\to0$, the stability of an isolated simple eigenvalue under finite-dimensional matrix perturbations implies that there exists an integer $m_0$ such that, for every $m\geq m_0$, $K_m$ has exactly one eigenvalue, counted with algebraic multiplicity, in the disk
\begin{equation}
    \mathcal{D}_4 := \{z\in\mathbb{C}:|z-4|<1\}.
\end{equation}
Denote this eigenvalue by $\mu_*^{(m)}$. Then
\begin{equation}
    \lim_{m\to\infty}\mu_*^{(m)} = 4.
\end{equation}
The remaining eight eigenvalues converge, counting algebraic multiplicity, to either $1$ or $-2$. Since $K_m$ is real and $\mathcal{D}_4$ contains exactly one eigenvalue of $K_m$, $\mu_*^{(m)}$ must be real; otherwise, its complex conjugate would be a distinct eigenvalue in the same disk.

Let $\mathbf{q}_*^{(m)}$ and $\mathbf{w}_*^{(m)}$ denote the right and left eigenvectors associated with $\mu_*^{(m)}$, respectively. They are normalized by
\begin{equation}
    \mathbf{1}^T\mathbf{q}_*^{(m)} = 1,
    \qquad
    \left(\mathbf{w}_*^{(m)}\right)^T \mathbf{q}_*^{(m)} = 1.
\end{equation}
Because $\mu_*^{(m)}$ converges to the isolated simple eigenvalue $4$ of $G$, the corresponding right and left eigendirections converge to those of $G$ associated with $4$. Both limiting eigendirections are spanned by $\mathbf{e}_{00}$. With the above normalization,
\begin{equation}
\label{eq:appB_perron_vector_limits}
    \lim_{m\to\infty}\mathbf{q}_*^{(m)} = \mathbf{e}_{00},
    \qquad
    \lim_{m\to\infty}\mathbf{w}_*^{(m)} = \mathbf{e}_{00}.
\end{equation}

The eigenvalue of $\widehat M_{\mathrm{cyc}}^{(m)}$ corresponding to $\mu_*^{(m)}$ is
\begin{equation}
    \lambda_*^{(m)} = 1+r_m\mu_*^{(m)} = 1+4r_m+o(r_m).
\end{equation}
The remaining eigenvalues of $\widehat M_{\mathrm{cyc}}^{(m)}$ are of the forms
\begin{equation}
    1+r_m+o(r_m)
    \qquad
    \text{or}
    \qquad
    1-2r_m+o(r_m).
\end{equation}
Because $r_m>0$ and $r_m\to0$, there exists an integer $m_0$ such that, for every $m\geq m_0$, $\lambda_*^{(m)}$ is the unique eigenvalue of maximum modulus. The exact matrix $\widehat M_{\mathrm{cyc}}^{(m)}$ is nonnegative, since it is a positive rescaling of the unnormalized Bell-probability transition matrix $M_{\mathrm{cyc}}^{(m)}$. Hence, its unique maximum-modulus eigenvalue is its Perron eigenvalue, and $\mathbf{q}_*^{(m)}$ is its normalized nonnegative right Perron eigenvector. Since $M_{\mathrm{cyc}}^{(m)}$ and $\widehat M_{\mathrm{cyc}}^{(m)}$ differ only by a positive scalar factor, they have the same left and right Perron eigenvectors. In particular,
\begin{equation}
\label{eq:appB_perron_limit}
    \lim_{m\to\infty}\mathbf{q}_*^{(m)} = \mathbf{e}_{00}.
\end{equation}

The normalized conditional recurrence associated with complete successful cycles is
\begin{equation}
\label{eq:appB_normalized_cycle}
    \mathbf{q}^{(n+1,m)} = \frac{M_{\mathrm{cyc}}^{(m)}\mathbf{q}^{(n,m)}}{\mathbf{1}^TM_{\mathrm{cyc}}^{(m)}\mathbf{q}^{(n,m)}}.
\end{equation}
Repeated substitution in \eqref{eq:appB_normalized_cycle} shows that all intermediate normalization factors cancel, yielding
\begin{equation}
\label{eq:appB_cycle_power}
    \mathbf{q}^{(n,m)} = \frac{\left(M_{\mathrm{cyc}}^{(m)}\right)^n\mathbf{q}^{(0,m)}}{\mathbf{1}^T\left(M_{\mathrm{cyc}}^{(m)}\right)^n\mathbf{q}^{(0,m)}}.
\end{equation}

Let $\lambda_{\mathrm{P}}^{(m)}$ denote the Perron eigenvalue of $M_{\mathrm{cyc}}^{(m)}$. For every $m\geq m_0$, the simplicity and strict dominance of $\lambda_{\mathrm{P}}^{(m)}$ imply
\begin{equation}
    \left(M_{\mathrm{cyc}}^{(m)}\right)^n\mathbf{q}^{(0,m)} = \left(\lambda_{\mathrm{P}}^{(m)}\right)^n \mathbf{q}_*^{(m)} \left(\mathbf{w}_*^{(m)}\right)^T\mathbf{q}^{(0,m)} + o\left(\left|\lambda_{\mathrm{P}}^{(m)}\right|^n\right)
\end{equation}
as $n\to\infty$.

For the fixed maximum-line alignment used in this proof, the MUB preprocessing is independent of $m$. Hence the initial distribution is also independent of $m$, and we write $\mathbf{q}^{(0,m)}=\mathbf{q}^{(0)}$. Moreover, the preprocessing preserves the target Bell probability, so
\begin{equation}
    \mathbf{e}_{00}^T\mathbf{q}^{(0)} =q_{00}^{(0)} =a>\frac{1}{3}.
\end{equation}
Using the left-eigenvector limit in
\eqref{eq:appB_perron_vector_limits}, we obtain
\begin{equation}
    \lim_{m\to\infty} \left(\mathbf{w}_*^{(m)}\right)^T\mathbf{q}^{(0)} = \mathbf{e}_{00}^T\mathbf{q}^{(0)} =a>0.
\end{equation}
After increasing $m_0$ if necessary, it follows that
\begin{equation}
    \left(\mathbf{w}_*^{(m)}\right)^T\mathbf{q}^{(0)}>0
\end{equation}
for every $m\geq m_0$.

Finally, taking the subsequent large-$m$ limit and using \eqref{eq:appB_perron_limit}, we obtain
\begin{equation}
\label{eq:appB_universal_ordered_limit}
    \lim_{m\to\infty} \left( \lim_{n\to\infty} \mathbf{q}^{(n,m)} \right) = \mathbf{e}_{00},
\end{equation}
and hence
\begin{equation}
    \lim_{m\to\infty} \left( \lim_{n\to\infty} q_{00}^{(n,m)} \right) = \lim_{m\to\infty}q_{*,00}^{(m)} = 1.
\end{equation}
Thus, for every qutrit Pauli channel satisfying the assumptions of the theorem, the conditional fixed-point fidelity of the complete-cycle recurrence approaches unity as the carrier number tends to infinity. This proves Theorem~\ref{thm:universal_purification}.

\section{Proof of Theorem~\ref{thm:prime_power_mcaepp}}
\label{app:prime_power_mcaepp}

All indices and arithmetic operations in this appendix are evaluated over the finite field $\mathbb F_d$, where $d=\ell^r$ is a prime power. Let $\chi$ be the nontrivial additive character defined in Section~\ref{sec:higher}. We first derive the exact conditional map generated by one successful $m$-carrier check and then analyze the complete two-check cycle. Throughout this appendix, an $m$-carrier check means the raw finite-field star-stabilizer check consisting of encoding, carrier transmission, decoding, and all-zero syndrome post-selection. No additional bilateral shear is included in either check of the complete cycle.

We also justify the Clifford implementation used in \eqref{eq:finite_field_line_alignment_matrices}. Choose an $\mathbb F_\ell$ basis $e_1,\ldots,e_r$ of $\mathbb F_d$ and its trace-dual basis $f_1,\ldots,f_r$, satisfying $\operatorname{Tr}(e_if_j)=\delta_{ij}$. Writing $x=\sum_jx_je_j$ and $z=\sum_jz_jf_j$ and identifying $\ket{x}$ with $\bigotimes_j\ket{x_j}$ gives
\begin{equation}
    X(x)Z(z)=\bigotimes_{j=1}^r X_\ell(x_j)Z_\ell(z_j),
\end{equation}
because $\operatorname{Tr}(zu)=\sum_jz_ju_j$. Every $S\in SL(2,\mathbb F_d)$ is $\mathbb F_\ell$-linear and preserves the traced symplectic form $\operatorname{Tr}(zx'-xz')$; its coordinate matrix lies in $Sp(2r,\mathbb F_\ell)$. The symplectic--Clifford construction in Secs.~II.C--IV of \cite{Hostens2005Clifford}, including the modulo-$2\ell$ phase data required in even local dimension, yields a Clifford operation on the  corresponding $r$ systems of dimension $\ell$. Pulling it back through the basis identification implements $S_a$ or $S_\infty$ on the finite-field Pauli labels up to phases. Finally, $(U^*\otimes U)\ket{\Phi^{0,0}}=\ket{\Phi^{0,0}}$ and Pauli conjugation give
\begin{equation}
    (U^*\otimes U)\ket{\Phi^{x,z}} =e^{\mathrm i\vartheta(x,z)}\ket{\Phi^{x',z'}},
\end{equation}
where $(x',z')=S(x,z)$. This proves existence of the required implementation, not optimality of its experimental gate cost.

After the maximum-weight MUB line has been mapped to the axis $z=0$ and the multiplicative symmetrization has been applied, let $\{\widetilde p_{xz}\}_{x,z\in\mathbb F_d}$ denote the effective Pauli probabilities. Define
\begin{equation}
\label{eq:ff_parameters}
    a:=\widetilde p_{00}=p_{00},
    \qquad
    b:=\widetilde p_{x0}
    \quad\text{for every }x\neq0,
\end{equation}
and
\begin{equation}
\label{eq:ff_alpha_beta_gamma}
    \alpha:=a+(d-1)b,
    \qquad
    \beta:=a-b,
    \qquad
    \gamma:=\frac{1-\alpha}{d-1}.
\end{equation}
The multiplicative symmetrization also gives
\begin{equation}
\label{eq:ff_nonzero_column_weight}
    \sum_{x\in\mathbb F_d}\widetilde p_{xz}
    =\gamma
    \qquad\text{for every }z\neq0.
\end{equation}

We first dispose of the case $b=0$. In this case, the aligned axis $z=0$ has total weight $\alpha=a$. Since the multiplicative symmetrization preserves the total weight of this axis, its nonzero points also had zero total weight before the symmetrization. If any nonzero phase-space point outside the aligned axis had positive probability, the unique MUB line through this point and the origin would have weight strictly larger than $a$, contradicting the maximality of the aligned line. Hence every nonidentity Pauli probability vanishes and $a=1$. The channel is then noiseless, the initial Bell distribution is $\mathbf e_{00}$, and every successful cycle leaves it unchanged with unit success probability. The theorem is therefore immediate in this case.

For the remainder of the proof, assume $b>0$. Since $a=p_{00}>1/d$,
\begin{equation}
\label{eq:ff_beta_gamma_difference}
    \beta-\gamma
    =a-b-\frac{1-a-(d-1)b}{d-1}
    =\frac{da-1}{d-1}>0.
\end{equation}
Moreover,
\begin{equation}
    \alpha-\beta=db>0.
\end{equation}
Consequently,
\begin{equation}
\label{eq:ff_strict_parameter_order}
    0\leq\gamma<\beta<\alpha.
\end{equation}

Let
\begin{equation}
\label{eq:ff_q_vector}
    \mathbf q
    =
    \bigl(q_{st}\bigr)_{s,t\in\mathbb F_d}
    \in\mathbb R^{d^2}
\end{equation}
be the Bell-probability vector immediately before one $m$-carrier check. The errors on the $m$ independently transmitted carriers are denoted by $(x_j,z_j)$, $j=1,\ldots,m$, and are independently distributed according to $\{\widetilde p_{xz}\}$.

We first derive the syndrome and output label directly. For an input label $(s,z_0)$, Alice's Fourier and inverse-SUM encoding gives
\begin{equation}
\label{eq:ff_after_alice_encoding}
    \frac{1}{d^{m/2}}\sum_{u,r_1,\ldots,r_{m-1}} \chi(z_0u)\ket{u,u+s}_{AB} \ket{r_1,\ldots,r_{m-1},-u-\sum_{i=1}^{m-1}r_i}.
\end{equation}
Under $X(x)Z(z)\ket r=\chi(zr)\ket{r+x}$, carrier errors transform this state into
\begin{align}
    \frac{1}{d^{m/2}}\sum_{u,\mathbf r}
    &\chi\!\left((z_0-z_m)u+\sum_{i=1}^{m-1}(z_i-z_m)r_i\right) \ket{u,u+s}_{AB}\notag\\
    \label{eq:ff_after_carrier_errors}
    &\otimes\ket{r_1+x_1,\ldots,r_{m-1}+x_{m-1}, -u-\sum_{i=1}^{m-1}r_i+x_m}.
\end{align}
Bob applies the forward SUM from $B$ to carrier $m$ and then from every carrier $i<m$ to carrier $m$. Its label becomes
\begin{equation}
    -u-\sum_i r_i+x_m+(u+s)+\sum_i(r_i+x_i) =s+\sum_{j=1}^m x_j.
\end{equation}
For the first $m-1$ carriers,
\begin{equation}
    H_d\ket{r_i+x_i}=\frac{1}{\sqrt d}\sum_{y_i} \chi\bigl(y_i(r_i+x_i)\bigr)\ket{y_i},
\end{equation}
and character orthogonality gives
\begin{equation}
    \sum_{r_i}\chi\bigl((z_i-z_m+y_i)r_i\bigr) =d\,\delta_{y_i,z_m-z_i}.
\end{equation}
Therefore the decoded state is
\begin{equation}
\label{eq:ff_raw_star_propagation}
    \chi\!\left(\sum_{i=1}^{m-1}x_i(z_m-z_i)\right) \ket{\Phi^{s,z_0-z_m}}_{AB} \bigotimes_{i=1}^{m-1}\ket{z_m-z_i}_i \left|s+\sum_{j=1}^{m}x_j\right\rangle_m.
\end{equation}
The normalization is $d^{-m/2}d^{-(m-1)/2}d^{m-1}=d^{-1/2}$, and the prefactor is a global phase. Hence the all-zero outcome is obtained exactly when
\begin{equation}
\label{eq:ff_success_conditions}
    z_1=\cdots=z_m=l
\end{equation}
and
\begin{equation}
\label{eq:ff_shift_condition}
    s+\sum_{j=1}^{m}x_j=0.
\end{equation}
The accepted Bell label is
\begin{equation}
\label{eq:ff_output_label}
    x_{\mathrm{out}}=s,\qquad z_{\mathrm{out}}=z_0-l.
\end{equation}
For $m=1$, the products over the first $m-1$ carriers are empty, the phase is one, and \eqref{eq:ff_raw_star_propagation} reduces to $\ket{\Phi^{s,z_0-z_1}}\ket{s+x_1}$.

For $s,l\in\mathbb F_d$, define
\begin{equation}
\label{eq:ff_kappa_definition}
    \kappa_s^{(m)}(l) := \sum_{\substack{x_1,\ldots,x_m\in\mathbb{F}_d\\ x_1+\cdots+x_m=-s}}
    \prod_{j=1}^{m}\widetilde p_{x_jl}.
\end{equation}
This is the joint probability that all carriers have phase label $l$ and that their shift labels satisfy \eqref{eq:ff_shift_condition}.

Using the finite-field character identity
\begin{equation}
\label{eq:ff_character_delta}
    \delta_0(y) = \frac{1}{d} \sum_{k\in\mathbb F_d}\chi(ky),
\end{equation}
\eqref{eq:ff_kappa_definition} becomes
\begin{align}
\label{eq:ff_kappa_character}
    \kappa_s^{(m)}(l)
    &=
    \frac{1}{d} \sum_{k\in\mathbb F_d} \chi(ks) \sum_{x_1,\ldots,x_m\in\mathbb F_d} \prod_{j=1}^{m} \widetilde p_{x_jl}\chi(kx_j) \nonumber\\
    &=
    \frac{1}{d} \sum_{k\in\mathbb F_d} \chi(ks)
    \left(
        \sum_{x\in\mathbb F_d}
        \widetilde p_{xl}\chi(kx)
    \right)^m.
\end{align}

For the aligned column $l=0$, \eqref{eq:ff_parameters} and \eqref{eq:ff_alpha_beta_gamma} give
\begin{equation}
\label{eq:ff_zero_column_trivial_character}
    \sum_{x\in\mathbb F_d}\widetilde p_{x0} =\alpha.
\end{equation}
For every $k\neq0$, multiplication by $k$ permutes $\mathbb F_d^\times$, and hence
\begin{align}
\label{eq:ff_zero_column_nontrivial_character}
    \sum_{x\in\mathbb F_d}\widetilde p_{x0}\chi(kx)
    &= a+b\sum_{x\neq0}\chi(kx) \nonumber\\
    &=a-b = \beta.
\end{align}
We also use
\begin{equation}
\label{eq:ff_character_sum_nonzero}
    \sum_{k\neq0}\chi(ks)
    =
    \begin{cases}
        d-1, & s=0,\\
        -1,  & s\neq0.
    \end{cases}
\end{equation}
Substitution into \eqref{eq:ff_kappa_character} yields
\begin{equation}
\label{eq:ff_kappa_zero_s0}
    \kappa_0^{(m)}(0) = \frac{\alpha^m+(d-1)\beta^m}{d},
\end{equation}
and, for every $s\neq0$,
\begin{equation}
\label{eq:ff_kappa_zero_snonzero}
    \kappa_s^{(m)}(0) = \frac{\alpha^m-\beta^m}{d}.
\end{equation}

For a nonzero column $l\neq0$, the constraint in \eqref{eq:ff_kappa_definition} can only decrease the total probability. Equation~\eqref{eq:ff_nonzero_column_weight} therefore gives
\begin{equation}
\label{eq:ff_kappa_nonzero_bound}
    0\leq\kappa_s^{(m)}(l) \leq
    \left(
        \sum_{x\in\mathbb F_d}\widetilde p_{xl}
    \right)^m
    =\gamma^m.
\end{equation}

Let $M_Z^{(m)}$ denote the unnormalized Bell-probability transition matrix generated by one successful check. From \eqref{eq:ff_output_label}, its exact action is
\begin{equation}
\label{eq:ff_single_check_map}
    \bigl(M_Z^{(m)}\mathbf q\bigr)_{st} = \sum_{l\in\mathbb F_d} \kappa_s^{(m)}(l)q_{s,t+l}.
\end{equation}
The corresponding success probability is
\begin{equation}
\label{eq:ff_single_check_success}
    P_Z^{(m)}(\mathbf q) = \mathbf 1^TM_Z^{(m)}\mathbf q.
\end{equation}

Let $\mathbf e_{st}$ denote the standard basis vector associated with the Bell label $(s,t)$, and define the diagonal matrix $D_Z$ by
\begin{equation}
\label{eq:ff_DZ_definition}
    D_Z\mathbf e_{st} =
    \begin{cases}
        (d-1)\mathbf e_{st}, & s=0,\\
        -\mathbf e_{st},     & s\neq0.
    \end{cases}
\end{equation}
Define
\begin{equation}
\label{eq:ff_rm_definition}
    r_m := \left(\frac{\beta}{\alpha}\right)^m.
\end{equation}
By \eqref{eq:ff_strict_parameter_order},
\begin{equation}
\label{eq:ff_rm_limit}
    0<r_m<1,
    \qquad
    r_m\longrightarrow0.
\end{equation}

We rescale the single-check matrix by setting
\begin{equation}
\label{eq:ff_MZ_hat_definition}
    \widehat M_Z^{(m)} := \frac{d}{\alpha^m}M_Z^{(m)}.
\end{equation}
Equations~\eqref{eq:ff_kappa_zero_s0} and \eqref{eq:ff_kappa_zero_snonzero} show that the contribution from $l=0$ is exactly $I+r_mD_Z$. Define the remaining matrix $E_m$ by
\begin{equation}
\label{eq:ff_Em_definition}
    (E_m\mathbf q)_{st} := \frac{d}{\alpha^m} \sum_{l\neq0}\kappa_s^{(m)}(l)q_{s,t+l}.
\end{equation}
Then
\begin{equation}
\label{eq:ff_single_check_expansion}
    \widehat M_Z^{(m)} = I+r_mD_Z+E_m.
\end{equation}

For the induced infinity norm, \eqref{eq:ff_kappa_nonzero_bound} gives
\begin{equation}
\label{eq:ff_Em_norm_bound}
    \|E_m\|_\infty \leq d(d-1) \left(\frac{\gamma}{\alpha}\right)^m.
\end{equation}
Since $\gamma<\beta$,
\begin{equation}
\label{eq:ff_Em_small_o}
    \frac{\|E_m\|_\infty}{r_m} \leq d(d-1) \left(\frac{\gamma}{\beta}\right)^m \longrightarrow0.
\end{equation}
Since the matrix dimension $d^2$ is fixed, all matrix norms are equivalent. Consequently,
\begin{equation}
\label{eq:ff_Em_o_rm}
    \|E_m\|=o(r_m)
\end{equation}
in any fixed matrix norm.

Let $\Pi_R$ denote the permutation matrix induced by the bilateral finite-field Fourier transformation
\begin{equation}
\label{eq:ff_RAB_definition}
    R_{AB}:=H_{d,A}\otimes H_{d,B}^*.
\end{equation}
Its action on the Bell basis is
\begin{equation}
\label{eq:ff_RAB_Bell_action}
    R_{AB}\ket{\Phi^{s,t}} = \chi(-st)\ket{\Phi^{t,-s}}.
\end{equation}
Since the phase factor $\chi(-st)$ does not affect Bell probabilities,
\begin{equation}
\label{eq:ff_PiR_action}
    \Pi_R\mathbf e_{st} = \mathbf e_{t,-s}.
\end{equation}
Define
\begin{equation}
\label{eq:ff_DX_Etilde_definition}
    D_X:=\Pi_R^{-1}D_Z\Pi_R,
    \qquad
    \widetilde E_m:=\Pi_R^{-1}E_m\Pi_R.
\end{equation}
It follows from \eqref{eq:ff_DZ_definition} and \eqref{eq:ff_PiR_action} that
\begin{equation}
\label{eq:ff_DX_action}
    D_X\mathbf e_{st} =
    \begin{cases}
        (d-1)\mathbf e_{st}, & t=0,\\
        -\mathbf e_{st},     & t\neq0.
    \end{cases}
\end{equation}
Since conjugation by $\Pi_R$ only permutes rows and columns,
\begin{equation}
\label{eq:ff_Etilde_small_o}
    \|\widetilde E_m\|=o(r_m).
\end{equation}

The unnormalized matrix associated with one complete successful cycle is
\begin{equation}
\label{eq:ff_cycle_matrix}
    M_{\mathrm{cyc}}^{(m)} = \Pi_R^{-1}M_Z^{(m)}\Pi_RM_Z^{(m)}.
\end{equation}
The second occurrence of $M_Z^{(m)}$ corresponds to a fresh, independent set of $m$ carriers. Define the rescaled cycle matrix by
\begin{equation}
\label{eq:ff_cycle_matrix_hat}
    \widehat M_{\mathrm{cyc}}^{(m)} := \frac{d^2}{\alpha^{2m}}M_{\mathrm{cyc}}^{(m)}.
\end{equation}
Using \eqref{eq:ff_single_check_expansion}, we obtain
\begin{align}
\label{eq:ff_cycle_expansion}
    \widehat M_{\mathrm{cyc}}^{(m)}
    &= \Pi_R^{-1}\widehat M_Z^{(m)} \Pi_R\widehat M_Z^{(m)} \nonumber\\
    &= \bigl(I+r_mD_X+\widetilde E_m\bigr) \bigl(I+r_mD_Z+E_m\bigr) \nonumber\\
    &= I+r_m(D_Z+D_X)+F_m,
\end{align}
where
\begin{align}
\label{eq:ff_Fm_definition}
    F_m :={}
    & E_m+\widetilde E_m +r_mD_XE_m +r_m\widetilde E_mD_Z \nonumber\\
    &+ \widetilde E_mE_m +r_m^2D_XD_Z.
\end{align}
Equations~\eqref{eq:ff_rm_limit}, \eqref{eq:ff_Em_o_rm}, and \eqref{eq:ff_Etilde_small_o} imply
\begin{equation}
\label{eq:ff_Fm_small_o}
    \|F_m\|=o(r_m).
\end{equation}

Set
\begin{equation}
\label{eq:ff_G_definition}
    G:=D_Z+D_X.
\end{equation}
Equations~\eqref{eq:ff_DZ_definition} and \eqref{eq:ff_DX_action} show that $G$ is diagonal in the Bell basis and satisfies
\begin{equation}
\label{eq:ff_G_action}
    G\mathbf e_{st} =
    \begin{cases}
        2(d-1)\mathbf e_{st},
        &(s,t)=(0,0),\\[1mm]
        (d-2)\mathbf e_{st},
        &\text{exactly one of $s$ and $t$ is zero},\\[1mm]
        -2\mathbf e_{st},
        &s\neq0,\ t\neq0.
    \end{cases}
\end{equation}
Thus the eigenvalues of $G$ are
\begin{equation}
\label{eq:ff_G_spectrum}
\begin{aligned}
    &2(d-1) &\text{with multiplicity } &1,\\
    &d-2 &\text{with multiplicity } &2(d-1),\\
    &-2 &\text{with multiplicity } &(d-1)^2.
\end{aligned}
\end{equation}
The eigenspace associated with the largest eigenvalue is
\begin{equation}
\label{eq:ff_G_dominant_eigenspace}
    \ker\bigl(G-2(d-1)I\bigr) = \operatorname{span}\{\mathbf e_{00}\}.
\end{equation}
In particular, the largest eigenvalue $2(d-1)$ is simple, and its gap to the next eigenvalue is
\begin{equation}
\label{eq:ff_G_gap}
    2(d-1)-(d-2)=d.
\end{equation}

To analyze the dominant eigendirection, define
\begin{equation}
\label{eq:ff_Km_definition}
    K_m := \frac{\widehat M_{\mathrm{cyc}}^{(m)}-I}{r_m} = G+\frac{F_m}{r_m}.
\end{equation}
By \eqref{eq:ff_Fm_small_o},
\begin{equation}
\label{eq:ff_Km_to_G}
    \|K_m-G\|\longrightarrow0.
\end{equation}
The matrices $K_m$ and $\widehat M_{\mathrm{cyc}}^{(m)}$ have the same left and right eigenvectors. Indeed, if $K_m\mathbf v=\mu\mathbf v$, then
\begin{equation}
\label{eq:ff_Km_Mm_eigenvalue_relation}
    \widehat M_{\mathrm{cyc}}^{(m)}\mathbf v = (1+r_m\mu)\mathbf v.
\end{equation}

Since the eigenvalue $2(d-1)$ of $G$ is simple and isolated, standard finite-dimensional eigenvalue perturbation theory implies that there exists an integer $m_0$ such that, for every $m\geq m_0$, $K_m$ has exactly one eigenvalue, counted with algebraic multiplicity, in the disk
\begin{equation}
\label{eq:ff_dominant_disk}
    \mathcal D_\ast :=
    \left\{
        z\in\mathbb C:
        |z-2(d-1)|<\frac{d}{3}
    \right\}.
\end{equation}
Denote this eigenvalue by $\mu_\ast^{(m)}$. Then
\begin{equation}
\label{eq:ff_mu_star_limit}
    \mu_\ast^{(m)} = 2(d-1)+o(1).
\end{equation}
Because $K_m$ is real and $\mathcal D_\ast$ contains exactly one eigenvalue, $\mu_\ast^{(m)}$ is real. The remaining eigenvalues converge, counting algebraic multiplicity, to either $d-2$ or $-2$.

The corresponding eigenvalue of $\widehat M_{\mathrm{cyc}}^{(m)}$ is
\begin{equation}
\label{eq:ff_lambda_star_expansion}
    \lambda_\ast^{(m)} = 1+2(d-1)r_m+o(r_m).
\end{equation}
The remaining eigenvalues have the forms
\begin{equation}
\label{eq:ff_other_eigenvalues}
    1+(d-2)r_m+o(r_m)
    \qquad
    \text{or}
    \qquad
    1-2r_m+o(r_m).
\end{equation}
Since $r_m>0$, after increasing $m_0$ if necessary, $\lambda_\ast^{(m)}$ is the unique eigenvalue of maximum modulus for every $m\geq m_0$.

The exact matrix $M_{\mathrm{cyc}}^{(m)}$ is nonnegative because it is the unnormalized Bell-probability transition matrix. Its positive rescaling $\widehat M_{\mathrm{cyc}}^{(m)}$ is therefore also nonnegative. Hence $\lambda_\ast^{(m)}$ is the Perron eigenvalue of $\widehat M_{\mathrm{cyc}}^{(m)}$.

Let $\mathbf q_\ast^{(m)}$ and $\mathbf w_\ast^{(m)}$ be the corresponding right and left Perron eigenvectors, normalized by
\begin{equation}
\label{eq:ff_perron_normalization}
    \mathbf 1^T\mathbf q_\ast^{(m)}=1,
    \qquad
    \bigl(\mathbf w_\ast^{(m)}\bigr)^T \mathbf q_\ast^{(m)}=1.
\end{equation}
The convergence of the isolated eigendirections in \eqref{eq:ff_Km_to_G}, together with \eqref{eq:ff_G_dominant_eigenspace}, gives
\begin{equation}
\label{eq:ff_perron_vectors_limit}
    \mathbf q_\ast^{(m)} \longrightarrow\mathbf e_{00},
    \qquad
    \mathbf w_\ast^{(m)} \longrightarrow\mathbf e_{00}.
\end{equation}
Equivalently, every accumulation point $\mathbf v$ of the normalized right Perron vectors satisfies
\begin{equation}
    G\mathbf v=2(d-1)\mathbf v,
    \qquad
    \mathbf 1^T\mathbf v=1,
\end{equation}
whose unique solution is $\mathbf v=\mathbf e_{00}$.

The normalized conditional recurrence associated with complete successful cycles is
\begin{equation}
\label{eq:ff_normalized_cycle_recurrence}
    \mathbf q^{(N+1,m)} =
    \frac{
        M_{\mathrm{cyc}}^{(m)}
        \mathbf q^{(N,m)}
    }{
        \mathbf 1^T
        M_{\mathrm{cyc}}^{(m)}
        \mathbf q^{(N,m)}
    }.
\end{equation}
Repeated substitution shows that
\begin{equation}
\label{eq:ff_normalized_cycle_power}
    \mathbf q^{(N,m)} =
    \frac{
        \bigl(M_{\mathrm{cyc}}^{(m)}\bigr)^N
        \mathbf q^{(0)}
    }{
        \mathbf 1^T
        \bigl(M_{\mathrm{cyc}}^{(m)}\bigr)^N
        \mathbf q^{(0)}
    }.
\end{equation}

The preprocessing preserves the target Bell probability, so the $m$-independent initial distribution satisfies
\begin{equation}
\label{eq:ff_initial_target_weight}
    \mathbf e_{00}^T\mathbf q^{(0)} = q_{00}^{(0)} = a > \frac{1}{d}.
\end{equation}
Using \eqref{eq:ff_perron_vectors_limit}, we obtain
\begin{equation}
\label{eq:ff_initial_projection_limit}
    \lim_{m\to\infty} \bigl(\mathbf w_\ast^{(m)}\bigr)^T \mathbf q^{(0)} = \mathbf e_{00}^T\mathbf q^{(0)} = a>0.
\end{equation}
After increasing $m_0$ if necessary,
\begin{equation}
\label{eq:ff_initial_projection_positive}
    \bigl(\mathbf w_\ast^{(m)}\bigr)^T \mathbf q^{(0)}>0
\end{equation}
for every $m\geq m_0$. Thus the initial distribution has a nonzero projection onto the dominant Perron eigenspace.

Let $\lambda_P^{(m)}$ denote the Perron eigenvalue of the unscaled matrix $M_{\mathrm{cyc}}^{(m)}$. Its simplicity and strict dominance give
\begin{equation}
\label{eq:ff_perron_power_expansion}
    \bigl(M_{\mathrm{cyc}}^{(m)}\bigr)^N \mathbf q^{(0)} = \bigl(\lambda_P^{(m)}\bigr)^N \mathbf q_\ast^{(m)} \bigl(\mathbf w_\ast^{(m)}\bigr)^T \mathbf q^{(0)} + o\!\left( |\lambda_P^{(m)}|^N \right)
\end{equation}
as $N\to\infty$. Substituting \eqref{eq:ff_perron_power_expansion} into \eqref{eq:ff_normalized_cycle_power} yields
\begin{equation}
\label{eq:ff_fixed_m_convergence}
    \lim_{N\to\infty}\mathbf q^{(N,m)} = \mathbf q_\ast^{(m)}
\end{equation}
for every fixed $m\geq m_0$.

Taking the subsequent large-$m$ limit and using \eqref{eq:ff_perron_vectors_limit}, we obtain
\begin{equation}
\label{eq:ff_ordered_vector_limit}
    \lim_{m\to\infty}
    \left(
        \lim_{N\to\infty}\mathbf q^{(N,m)}
    \right)
    = \mathbf e_{00}.
\end{equation}
In particular,
\begin{equation}
\label{eq:ff_ordered_fidelity_limit}
    \lim_{m\to\infty}
    \left(
        \lim_{N\to\infty}q_{00}^{(N,m)}
    \right)
    = \lim_{m\to\infty}q_{\ast,00}^{(m)} = 1.
\end{equation}

It remains to verify the finite-resource assertion. Let
\begin{equation}
\label{eq:ff_cycle_success_probability}
    P_{\mathrm{cyc}}^{(m)}(\mathbf q) := \mathbf 1^T M_{\mathrm{cyc}}^{(m)} \mathbf q
\end{equation}
be the probability that one complete cycle succeeds. Whenever $q_{00}>0$, consider the event in which the input shared state has Bell label $(0,0)$ and all $2m$ carriers used in the two checks experience the identity error $(0,0)$. This event has probability $q_{00}a^{2m}$, is accepted by both checks, and leaves the Bell label equal to $(0,0)$. Therefore,
\begin{equation}
\label{eq:ff_success_positive_bound}
    P_{\mathrm{cyc}}^{(m)}(\mathbf q) \geq q_{00}a^{2m}>0,
\end{equation}
and
\begin{equation}
\label{eq:ff_target_component_positive_bound}
    \bigl(
        M_{\mathrm{cyc}}^{(m)}\mathbf q
    \bigr)_{00}
    \geq q_{00}a^{2m}>0.
\end{equation}
Since $q_{00}^{(0)}=a>0$, induction shows that every finite cycle has strictly positive success probability and that $q_{00}^{(N,m)}>0$ for every finite $N$.

For any prescribed $F_{\mathrm{tar}}<1$, \eqref{eq:ff_ordered_fidelity_limit} allows us to choose a finite $m\geq m_0$ such that
\begin{equation}
    q_{\ast,00}^{(m)}>F_{\mathrm{tar}}.
\end{equation}
The fixed-$m$ convergence in \eqref{eq:ff_fixed_m_convergence} then gives a finite $N$ such that
\begin{equation}
    q_{00}^{(N,m)}\geq F_{\mathrm{tar}}.
\end{equation}
The corresponding cumulative success probability is
\begin{equation}
\label{eq:ff_cumulative_success}
    P_{\mathrm{tot}}^{(N,m)} = \prod_{n=0}^{N-1} P_{\mathrm{cyc}}^{(m)} \bigl(\mathbf q^{(n,m)}\bigr)>0,
\end{equation}
because it is a finite product of strictly positive factors. This proves Theorem~\ref{thm:prime_power_mcaepp}.

\bibliographystyle{unsrt}
\bibliography{ref}

\end{document}